%% file: main.tex
\documentclass{article}

\usepackage{geometry}          
\makeatletter
\renewcommand{\paragraph}{%
  \@startsection{paragraph}{4}{\z@}%
  {0em}
  {-1em}
  {\normalfont\normalsize\bfseries}
}
\makeatother

\usepackage{etoolbox}
\usepackage{amsmath} 

\makeatletter
\let\oldparagraph\paragraph
\renewcommand{\paragraph}[1]{\oldparagraph{#1\@addpunct{.}}}
\makeatother

\usepackage[sort&compress,numbers]{natbib}
\usepackage[utf8]{inputenc} 
\usepackage[T1]{fontenc}    
\usepackage[colorlinks=true, linkcolor=red, citecolor=teal, urlcolor=blue, breaklinks=true, bookmarksopen=true]{hyperref}
\usepackage{hyperref}       
\usepackage{url}            
\usepackage{booktabs}       
\usepackage{amsfonts}       
\usepackage{nicefrac}       
\usepackage{microtype}      
\usepackage{xcolor}         %

\usepackage{algorithm}
\usepackage[noend]{algpseudocode}
\algrenewcommand\algorithmicrequire{\textbf{Input:}}  
\algrenewcommand\algorithmicensure{\textbf{Output:}}  

\usepackage{amsmath, amssymb, mathtools, bm, amsthm}

\newcommand{\x}[0]{\mathbf x}
\newcommand{\xrho}{\mathcal{X}(\rho)}

\newcommand{\xtot}[0]{\norm{\x}}
\newcommand{\xtotmi}[0]{\norm{\x_{-i}}}
\newcommand{\xr}[1]{\mathcal{X}(#1)}

\newcommand{\s}[0]{\mathbf s}
\newcommand{\stot}[0]{\mathbf{x}^{\top} \mathbf{s}}

\newcommand{\qai}[0]{Q_{\text{AI}}}
\newcommand{\Tx}[0]{T(\mathbf x)}

\usepackage{float}           
\usepackage{subcaption}      
\usepackage{url}

\newcommand{\proofof}[1]{\textbf{\textup{Proof of \Cref{#1}}}}

\DeclareMathOperator*{\argmax}{arg\,max}

\usepackage{amsmath, mathtools, bm}
\usepackage{nicefrac}

\usepackage{thmtools}
\usepackage{thm-restate}

\declaretheorem[name=Theorem]{theorem}
\declaretheorem[name=Lemma, sibling=theorem]{lemma} 
\declaretheorem[name=Proposition, sibling=theorem]{proposition} 
\declaretheorem[name=Remark, sibling=theorem]{remark} 

\declaretheorem[name=Definition, style=definition]{definition}

\usepackage{cleveref}  

\Crefname{theorem}{Theorem}{Theorems}
\Crefname{proposition}{Proposition}{Propositions}
\Crefname{lemma}{Lemma}{Lemmas}

\crefname{theorem}{theorem}{theorems}
\crefname{proposition}{proposition}{propositions}
\crefname{lemma}{lemma}{lemmas}

\crefname{inequality}{inequality}{inequalities}
\Crefname{inequality}{Inequality}{Inequalities}

\newcommand{\abs}[1]{\left| #1 \right|}
\newcommand{\norm}[1]{\left\| #1 \right\|}

\newcommand{\floor}[1]{\left\lfloor #1 \right\rfloor}

\title{Strategic Content Creation with GenAI: To Share or Not to Share?}

\author{
Gur Keinan%
\thanks{%
    {Technion---Israel Institute of Technology (\url{gur.keinan@campus.technion.ac.il})}}
\and Omer Ben{-}Porat%
\thanks{%
    {Technion---Israel Institute of Technology (\url{omerbp@technion.ac.il})}}
}
\date{}

\begin{document}

\maketitle

\begin{abstract}
We introduce a game-theoretic framework that examines strategic interactions between a platform and its content creators in the presence of AI-generated content. Our model's main novelty is in capturing creators' dual strategic decisions: Their investment in content quality and their (possible) consent to share their content with the platform's GenAI, both of which significantly impact their utility. To incentivize creators, the platform strategically allocates a portion of its GenAI-driven revenue to creators who share their content. We focus on the class of \emph{full-sharing} equilibrium profiles, in which all creators willingly share their content with the platform's GenAI system. Such equilibria are highly desirable both theoretically and practically. Our main technical contribution is formulating and efficiently solving a novel optimization problem that approximates the platform's optimal revenue subject to inducing a full-sharing equilibrium. A key aspect of our approach is identifying conditions under which full-sharing equilibria exist and a surprising connection to the Prisoner's Dilemma. Finally, our simulations demonstrate how revenue-allocation mechanisms affect creator utility and the platform's revenue.
\end{abstract}

\section{Introduction}

Online Content platforms such as YouTube, Medium, and SoundCloud have become ubiquitous, hosting vast ecosystems of creators who aim to maximize their reach, engagement, and revenue. The decision-making challenges faced by these creators have been the focus of a growing body of work, examining aspects such as revenue-sharing mechanisms~\cite{bhargava2022creator}, fairness in content visibility and diversification~\cite{10.1145/3292500.3330745,10.5555/3524938.3525586}, and the economic incentives shaping content creation~\cite{10.5555/3326943.3327046,ben2020content,Dean2024,acharya2025producers,hron2022modeling}. As these platforms evolve, especially with advances in Generative AI (GenAI), new considerations arise.

GenAI enables platforms to synthesize content by distilling and recombining existing creator contributions. For instance, a news platform could use GenAI tools to generate personalized summaries of the day's news tailored to individual users. This enhances user satisfaction and opens new monetization channels. However, it also introduces two significant challenges: Sharing consent and content quality.\footnote{For instance, YouTube allows its creators to opt out of having their content used for GenAI training~\cite{youtube_ai_training_policy}.} First, using creator content for GenAI raises legal and ethical questions around intellectual property and copyright infringement~\cite{samuelson2023generative}. Second, AI-generated content may cannibalize attention from human-created work, reducing creators' incentives to invest in content quality.

To address these challenges, platforms may offer a remedy: Sharing GenAI-driven revenue with creators who allow the platform's GenAI to use their content. This mirrors existing revenue-sharing programs on platforms like YouTube and X (formerly Twitter)~\cite{x_creator_revenue_sharing,youtube_partner_earnings}, and could incentivize both sharing consent and content quality. The main question thus becomes: \emph{How should the platform allocate its GenAI-driven revenue to maximize its (retained) revenue while incentivizing the creation and sharing of high-quality content?}

This paper pioneers the study of revenue sharing in the presence of strategic content creators and AI-generated content. We present a game-theoretic model consisting of a platform and strategic creators, in which the platform can generate content using GenAI tools and allocate its GenAI-driven revenue based on creator actions. Our framework incorporates creators' dual strategic decisions: (i) the extent to which they share their content with the platform's GenAI, ranging from full withholding to full sharing, and (ii) the body of content they create, each incurring its own cost. The platform selects a revenue-sharing mechanism rewarding creators based on these two axes of contribution, with the goal of optimizing for content quality, overall revenue, or a tradeoff between them.

\subsection{Our Contribution}

Our contribution is twofold. First and foremost, we introduce the first game-theoretic model that captures both content production and data sharing in Generative AI ecosystems, where creators \emph{and} the platform are strategic. User traffic is allocated through a Tullock contest structure~\cite{tullock2001efficient}, where creators and the platform's GenAI system compete based on their relative content quality. Within this environment, each creator decides not only the quality of their content but also how much of it to share with the GenAI. The GenAI's content quality is determined by the pool of shared creator content; thus, the platform should incentivize creators to produce and share content. To that end, we assume the platform can select a proportional allocation rule (see \Cref{eq:f-rho-definition}) parametrized by $\rho \in [0,1]$, and redistribute a $\rho$-portion of its GenAI-driven revenue to the creators. The platform is also strategic, and can use $\rho$ in a way that steers creators towards revenue-maximizing outcomes.

Our second contribution is technical. We study equilibria in the induced game, and focus on a subclass of equilibrium profiles we term \emph{full-sharing equilibria} (FSE for abbreviation; see Definition~\ref{def:fse}). In an FSE, creators are better off by sharing their content. FSEs are desirable due to several properties that we discuss in \Cref{subsec:solution-concepts}. We present a comprehensive characterization of FSE in \Cref{sec:fse}, showing that any $\rho$ induces at most one FSE. In \Cref{sec:platform-optimization}, we study the platform's revenue optimization problem (see~\eqref{eq:platform-optimization}). The platform can choose any value $\rho$ to maximize its revenue under the induced FSE, if it exists. Such bilevel optimization is notoriously NP-hard~\cite{hansen1992new}, and in our case, the difficulty is further exacerbated. Among other factors, finding equilibrium is non-trivial as best-response computation is non-concave, and the domain of feasible solutions does not have a closed-form. Nevertheless, we show the following:

\begin{theorem}[Informal version of \Cref{thm:alg_main}]
There exists an efficient algorithm that, for any $\varepsilon > 0$, computes a revenue allocation inducing an $\varepsilon$-FSE and achieving platform revenue within $\varepsilon$ of the maximum attainable under any FSE.
\end{theorem}

Our approach hinges on smoothness analysis of fixed points in an auxiliary game we call \emph{enforced sharing game}, where creators are forced to share their content. Finally, we present synthetic experiments to examine how revenue-allocation mechanisms affect content production, creator utility, and the platform's revenue.

\subsection{Related Work}
Strategic content provider dynamics have received substantial attention in recent years, with numerous models analyzing how creators compete for attention and adapt their strategies in algorithmically mediated ecosystems \cite{10.1145/3292500.3330745, ben2020content, Dean2024, acharya2025producers, hron2022modeling, 10.5555/3666122.3666764}. These works address key challenges such as incentivizing content quality \cite{hu2023incentivizing, 10.5555/3666122.3669383, 10.5555/3618408.3620064, 10.1145/3589334.3645353}, ensuring fairness and diversity \cite{10.5555/3600270.3602160, 10.5555/3524938.3525586, 10.1145/3490486.3538346, Yao2022, Yao2022Learning}, and aligning creator incentives with platform-level objectives \cite{10.1609/aaai.v38i20.30266, DBLP:journals/corr/abs-2305-11381}. The recent proliferation of Generative AI (GenAI) has introduced new complexity into these ecosystems, prompting a line of research exploring how GenAI alters competitive dynamics. For instance, some works consider environments in which either GenAI or human agents act strategically, but not both \cite{10.1609/aaai.v39i13.33548, esmaeili2025strategizehumancontentcreation}, while others examine GenAI systems that act strategically to preserve human participation \cite{taitler2025selectiveresponsestrategiesgenai}.

The work most closely related to ours is \citet{10.5555/3692070.3694417}, who introduce a model of competition between human creators and a GenAI agent within a Tullock contest~\cite{ewerhart2015mixed,EWERHART2017168,Ghosh2023,chowdhury2011generalized}, a celebrated and well-established economic model of competition. Our model, like that of \citet{10.5555/3692070.3694417}, captures how user traffic is distributed based on the relative quality of content produced by humans and a GenAI system. The effectiveness of the GenAI system itself is derived from human-generated content. However, our paper significantly departs from \citet{10.5555/3692070.3694417}. First, we allow creators to decide whether and to what extent their content is shared with the GenAI system, making data sharing itself a strategic choice. Second, we model the platform as a strategic actor that designs a revenue-allocation mechanism to influence creator behavior. These modifications result in a doubly strategic setting, where both creators and the platform engage in interdependent optimization. From a technical perspective, the creators' strategy spaces become more complex, and the game is no longer monotone~\cite{rosen1965existence}. Furthermore, modeling the platform as a strategic entity gives rise to a novel bilevel optimization problem of revenue maximization, which has not been addressed in previous works. Consequently, our research questions are novel, and our results significantly extend the existing literature by addressing strategic interactions that \citet{10.5555/3692070.3694417} have not explored.

More broadly, since our work addresses (strategic) content sharing with GenAI, it relates to works on data valuation~\cite{pmlr-v89-jia19a, baghcheband2025shapley, pmlr-v206-wang23e} and data markets~\cite{acemoglu2022too}. Furthermore, we consider the allocation of GenAI-driven revenue to content creators, thereby relating to copyright challenges with GenAI~\cite{gans2024copyright, pasquale2024consent,yang2024generativeaicopyrightdynamic}, attribution to specific data sources~\cite {bohacek2025aimodeltrainedimages}, and royalties~\cite{deng2024computational}.

\section{Model}
We consider a platform that hosts content produced by a set of human creators, alongside content generated by a GenAI system controlled by the platform. We formally define the components of the interaction below.

\paragraph{Creators}
We consider $n\geq 2$ human content creators competing for user traffic and denote the set of creators by $[n]=\{1,2,\dots,n\}$, with a typical creator indexed by $i\in[n]$.  Each creator produces a \emph{body of content} that encompasses multiple attributes: Quality, quantity, trend alignment, domain expertise or specialty, and so on.  As is standard in prior work (e.g.,~\cite{10.5555/3737916.3740675,10.5555/3692070.3694417}), we aggregate these attributes into a single scalar, and without loss of generality, call it \emph{content quality}. Formally, creator $i$ selects a quality $x_i\in[0,\infty)$.

For each creator $i$, creating a content $x_i$ incurs a production cost $c_i(x_i)$. While we do not assume a particular form of the cost function, we have several structural assumptions that are standard in prior work. We assume that for every $i \in [n]$, the cost function $c_i$ is non-decreasing, twice differentiable, and strictly convex. Moreover, we assume that $c_i(0) = 0$,  $\lim_{x \to 0^+} \frac{d c_i(x)}{dx}=0$, and $\lim_{x \to\infty} \frac{d c_i(x)}{dx}=\infty$. These assumptions ensure that not creating content is costless, and that the marginal cost of improving the quality by one unit is initially zero but becomes unbounded.

Beyond choosing the content quality, each creator also chooses their \emph{sharing level} $s_i \in [0, 1]$. The sharing level represents the extent to which creator $i$ allows their content $x_i$ to be used by the GenAI system, ranging from complete withholding ($s_i=0$) to full sharing ($s_i=1$). Therefore, the strategy of each creator $i\in[n]$ is a pair $(x_i,s_i)$. Let $\mathbf{x} = (x_1, \ldots, x_n)$ and $\mathbf{s} = (s_1, \ldots, s_n)$ denote the quality and sharing profiles, respectively. Additionally, for a creator $i \in [n]$, we let $\mathbf{x}_{-i} = (x_1, \ldots, x_{i-1}, x_{i+1}, \ldots, x_n)$ and $\mathbf{s}_{-i} = (s_1, \ldots, s_{i-1}, s_{i+1}, \ldots, s_n)$ denote the quality and sharing profiles after omitting the $i$'th entry.

\paragraph{Generative AI}
The platform operates a GenAI system, modeled as a non-strategic agent whose output quality depends on the aggregate quality of the content shared from creators, $\stot = \sum_{j=1}^{n} x_j s_j$. Its quality is given by $\qai(\mathbf{x}, \mathbf{s}) = \alpha \cdot \stot$, where $\alpha > 0$ represents the efficiency of data usage.\footnote{We discuss extensions of some of our results to a more general functional form of $\qai$ in \Cref{sec:extensions}.} Notably, if no data is shared (i.e., $\s = 0$), then $\qai(\mathbf{x}, \mathbf{s}) = 0$. Whenever $\mathbf{x}$ and $\mathbf{s}$ are clear from the context, we simplify the notation and use  $\qai$.

\paragraph{User traffic and engagement}
The platform attracts user traffic denoted by $\Tx$, which increases monotonically with the total quality of human content, $\xtot = \sum_{j=1}^{n} x_j$, where we use $\|\cdot\|$ without a subscript to denote the $\ell_1$ norm throughout the paper. We model traffic as $\Tx = \mu \cdot \xtot^{\gamma}$,
where $\mu > 0$ is a baseline traffic parameter and $\gamma \in [0, 1]$ captures the elasticity of traffic with respect to the total human content quality~\cite{TAFESSE2023103357, 10.25300/MISQ/2014/38.3.04}.\footnote{We assume that users are attracted to the platform due to its human-generated content, not the AI-generated content. The reason is that users tend to value originality and exhibit bias against AI-generated content when it lacks it~\cite{bellaiche2023humans}, which is the case in emerging domains where historical data is scarce~\cite{Li_Flanigan_2024}. We discuss extensions of our results to other form of traffic in \Cref{sec:extensions}.}

Given a strategy profile $(\mathbf x,\mathbf s)$ and the induced user traffic $\Tx$, users allocate their attention to human and AI-generated content according to a Tullock-like decision based on relative quality~\cite{tullock2001efficient}. Specifically, the fraction of attention allocated to creator $i$'s content is $\frac{x_i}{\xtot + \qai(\mathbf{x}, \mathbf{s})}$, while the fraction allocated to AI-generated content is $\frac{\qai(\mathbf{x}, \mathbf{s})}{\xtot + \qai(\mathbf{x}, \mathbf{s})}$. In both cases, the denominator represents the total quality of created and generated content available on the platform.

\paragraph{Platform strategy}
The platform may have multiple sources of revenue. However, to ease presentation, we focus on the user traffic it receives from its AI-generated content. Note that it depends on $\qai$, which in turn depends on human content. To obtain high-quality AI-generated content, the platform can allocate some of this revenue back to creators. Formally, the platform commits to an \emph{allocation rule} $f : [0,\infty)^n \times [0,1]^n \to \Delta^{n+1}$. It maps every strategy profile $(\x, \s)$ to a vector $(f_1(\mathbf{x},\mathbf{s}), \ldots, f_n(\mathbf{x},\mathbf{s}), f_P(\mathbf{x},\mathbf{s}))$, where $f_i(\mathbf{x},\mathbf{s})$ represents the fraction of GenAI-driven revenue allocated to creator $i \in [n]$, and $f_P(\mathbf{x},\mathbf{s})$ is the fraction retained by the platform. Here, $\Delta^{n+1}$ denotes the $(n+1)$-dimensional simplex.

While various allocation mechanisms are possible, we focus on the class of parameterized \emph{proportional allocation rules $f_\rho$}, formally defined below, and address other designs in \Cref{sec:other_allocation_rules}. The revenue-allocation parameter $\rho \in [0,1]$, or simply the \emph{allocation parameter}, specifies the fraction of GenAI-driven revenue allocated to creators, with the remainder $1 - \rho$ retained by the platform. Under this rule, the platform distributes the $\rho$-fraction among creators in proportion to their effective shared contributions $x_i s_i$, yielding

\begin{equation}\label{eq:f-rho-definition}
f_{i, \rho}(\mathbf{x}, \mathbf{s}) =
\begin{cases}
\rho \cdot \frac{x_i s_i}{\stot} &  \stot > 0  \\
0 & \stot = 0
\end{cases},
\end{equation}
and $f_{P, \rho}(\mathbf{x}, \mathbf{s}) = 1 - \sum_{i = 1}^{n} f_{i, \rho}(\mathbf{x}, \mathbf{s})$. We note that whenever $\stot > 0$, it holds that  $\sum_{i = 1}^{n} f_{i, \rho}(\mathbf{x}, \mathbf{s}) =\rho$ and $f_{P, \rho}(\mathbf{x}, \mathbf{s}) = 1 -\rho$.

\paragraph{Revenue and utilities}
Each human creator $i \in [n]$ receives revenue from two sources and incurs a creation cost. The first revenue source is direct user traffic to their content, yielding a revenue of $\Tx \frac{x_i}{\xtot + \qai}$. The second revenue source is their allocated share of the GenAI-driven revenue, given by $\Tx \frac{\qai}{\xtot + \qai} \cdot f_{i, \rho}(\mathbf{x}, \mathbf{s})$. Overall, creator $i$'s utility is
\[
U_{i}(\mathbf{x},\mathbf{s}; f_{\rho}) = \Tx \frac{x_i}{\xtot + \qai} + \Tx \frac{\qai}{\xtot + \qai} \cdot f_{i, \rho}(\mathbf{x}, \mathbf{s}) - c_i(x_i).
\]
The platform's revenue corresponds to the portion of GenAI revenue it retains under the allocation rule~$f_{\rho}$, which is given by
\[
U_{P}(\mathbf{x},\mathbf{s}; f_{\rho}) = \Tx \frac{\qai}{\xtot + \qai} \cdot f_{P, \rho}(\mathbf{x}, \mathbf{s}).
\]

\paragraph{Stackelberg game}
We model the interaction as a Stackelberg game, with the platform as the leader and the creators as followers. The platform first commits to a revenue-allocation rule $f_{\rho}$. This choice induces a subgame among the creators, denoted by \emph{$G(\rho)$}. In this game, each creator $i \in [n]$ simultaneously selects a strategy $(x_i, s_i)$ to maximize their utility given the fixed rule $f_{\rho}$. The resulting profile $(\mathbf{x}, \mathbf{s})$ determines the user traffic $\Tx$, the quality of AI-generated content $\qai(\mathbf{x}, \mathbf{s})$, and the resulting payoffs $U_i(\mathbf{x}, \mathbf{s}; f_{\rho})$ and $U_P(\mathbf{x}, \mathbf{s}; f_{\rho})$. This structure enables the platform to steer creator behavior through the choice of the allocation parameter $\rho$.

\subsection{Solution Concepts}\label{subsec:solution-concepts}

We are interested in stable outcomes in the subgame $G(\cdot)$, where stability refers to pure Nash equilibrium (PNE). We highlight a class of PNEs that we call \emph{\textbf{F}ull \textbf{S}haring PN\textbf{E}}, or FSE for abbreviation. In FSE, full content sharing arises endogenously: Each creator chooses to share their content because the platform's revenue-allocation mechanism renders sharing a best response.

\begin{definition}[FSE]\label{def:fse}
    A strategy profile $(\mathbf{x}, \mathbf{1})$ is an \emph{FSE} in the subgame $G(\rho)$ if for every creator $i \in [n]$,
    \[
        U_{i}(\mathbf{x}, \mathbf{1}; f_{\rho}) \ge \sup_{x_i' \in [0, \infty),\, s_i' \in [0,1]} U_{i}((x_i', \mathbf{x}_{-i}), (s_i', \mathbf{1}_{-i}); f_{\rho}).
    \]
\end{definition}

FSEs are a special case of PNEs in which all creators fully share their content. While other equilibria involving partial or no sharing may exist, FSEs are appealing from the platform's perspective: They maximize the quality of AI-generated content given the created content. Focusing on FSE is also practically motivated, as it prevents unauthorized use of content and disputes or resentment from creators who believe their content was used without consent. FSEs naturally extend to approximate equilibria: A strategy profile $(\mathbf{x}, \mathbf{1})$ is an $\varepsilon$-FSE if no creator $i$ has an $\varepsilon$-beneficial deviation.

To facilitate the analysis in the forthcoming sections, we also present the following auxiliary equilibrium concept called \emph{\textbf{E}nforced \textbf{S}haring PN\textbf{E}}, or simply ESE.

\begin{definition}[ESE]\label{def:ese}
    A quality profile $\mathbf{x}$ is an \emph{ESE} in the subgame $G(\rho)$ if for every creator $i \in [n]$,
    \[
        U_{i}(\mathbf{x}, \mathbf{1}; f_{\rho}) \ge \sup_{x_i' \in [0, \infty)} U_{i}((x_i', \mathbf{x}_{-i}), \mathbf{1}; f_{\rho}).
    \]
\end{definition}

To motivate \Cref{def:ese}, imagine the auxiliary \emph{enforced sharing game}, where $f_{\rho}$ is fixed and creators \emph{must} share their content, i.e., $\mathbf{s} = \mathbf{1}$. ESEs allow us to analyze game dynamics by focusing solely on creators' quality profiles $\mathbf x$. We remark that any quality profile $\mathbf x$ such that $(\mathbf{x}, \mathbf{1})$ is an FSE also constitutes an ESE.

\section{FSE Analysis}\label{sec:fse}

In this section, we analyze how creators respond to an allocation rule. Subsection~\ref{subsec:best-sharing-response} assumes a fixed allocation rule $f_\rho$ and studies the structure of ESEs in $G(\rho)$. We show that every $\rho$ induces a unique ESE, although different values of $\rho$ may induce different ESEs. In Subsection~\ref{subsec:fse-existence}, we move to study FSEs. Namely, \Cref{thm:fse-existence} characterizes a sufficient condition under which the unique ESE $\x^\star$ also constitutes the unique FSE $(\x^\star, \mathbf{1})$.

\subsection{Best Sharing Response and Uniqueness of ESE}\label{subsec:best-sharing-response}

We begin by analyzing the creators' best responses under a fixed allocation rule $f_\rho$. First, we provide a full characterization of the optimal sharing level $s_i$ for a given profile $(\x, \mathbf s_{-i})$, showing that it takes a threshold form.

\begin{restatable}{lemma}{LemSChoiceThreshold}
    \label{lem:s_choice_threshold}
    Fix an arbitrary creator $i$ with content $x_i >0$, and a profile $(\mathbf x_{-i}, \mathbf{s}_{-i})$. Further, let $\tau_i$  such that $\tau_i = \frac{x_i}{x_i + \xtotmi + \alpha \mathbf{x}^{\top}_{-i} \mathbf{s}_{-i}}$. The optimal sharing level $s$ of creator $i$ w.r.t. $(x_i, \mathbf x_{-i}, \mathbf s_{-i})$, i.e., $s\in \argmax_{s\in [0,1]} U_i ((x_i,\mathbf x_{-i}),(s, \mathbf{s}_{-i}); f_{\rho}) $ is given by
    \begin{equation}\label{eq:sharing-threshold-form}
        s^\star =
        \begin{cases}
            1 &    \rho > \tau_i \\
            \textnormal{any number from } [0,1] &    \rho  = \tau_i \\
            0 &    \rho < \tau_i \\
        \end{cases}.
    \end{equation}
\end{restatable}

Due to space constraints, all proofs in this and subsequent sections are deferred to the appendix. The proof of \Cref{lem:s_choice_threshold} involves showing that creator $i$'s utility is strictly increasing in $s_i$ when $\rho > \tau_i$, strictly decreasing when $\rho < \tau_i$, and constant when $\rho = \tau_i$, yielding \Cref{eq:sharing-threshold-form}.

\Cref{lem:s_choice_threshold} reveals an interesting asymmetry: Creators producing a large share of $\xtot$ have higher thresholds; thus, they are more inclined to withhold their data, as doing so weakens a powerful GenAI competitor. In contrast, creators with relatively low quality have stronger incentives to share. Notably, the sharing decision becomes effectively binary---each creator either fully shares or fully withholds, with no incentive to choose an intermediate level.

Another insight from \Cref{lem:s_choice_threshold} concerns the case $\mathbf{s} = \mathbf{0}$, where all creators initially withhold their data. Fix a quality profile $\x$, and notice that if the number of creators $n$ is moderately large, the threshold $\tau_i = \frac{x_i}{\xtot}$ becomes small for most creators. Consequently, even a low revenue-allocation parameter $\rho  \approx \nicefrac{1}{n}$ can make it beneficial for the least productive creator, denoted by $i$, to deviate and set $s_i = 1$. This deviation increases $\stot$, thereby lowering the threshold $\tau_j$ for other creators $j \neq i$ to $\frac{x_j}{\xtot + \alpha x_i} <\frac{x_j}{\xtot}=\tau_j$. As a result,  creator $i$'s deviation can potentially trigger a cascade of full-sharing responses from other players, an effect similar to the classic prisoner's dilemma.

Next, we study creators' quality choices under enforced sharing.

\begin{restatable}{lemma}{LemUniqueESE}
    \label{lem:unique-ese}
    For any fixed $\rho$, the subgame $G(\rho)$ admits a unique ESE.
\end{restatable}

The proof relies on Rosen's diagonal strict concavity (DSC) condition \cite{rosen1965existence}, a standard approach for establishing existence and uniqueness of PNEs in concave games.

Given the existence and uniqueness established in \Cref{lem:unique-ese}, we can define the ESE correspondence: Let $\mathcal X: [0,1] \to [0,\infty)^n$ denote the mapping from $\rho$ to its corresponding unique ESE profile, $\xr{\rho}$, in the subgame $G(\rho)$.

\subsection{Existence of FSE}\label{subsec:fse-existence}

We now examine conditions under which full sharing can arise in a PNE. For a profile to be an FSE, it must induce a quality profile that constitutes an ESE under enforced sharing. The uniqueness of the ESE implies that, if an FSE exists, both the creators' content qualities and their sharing decisions are uniquely determined. Therefore, the remaining question is whether the provided revenue allocation is enough to induce such a profile in equilibrium.

First, we consider the two extreme cases, $\rho=0$ and $\rho = 1$. The former is trivial---reators can only be worse off from sharing; hence, withholding content is a dominant strategy, and no FSE exists in $G(0)$. The latter case is more interesting. If $\rho=1$, the platform allocates its entire GenAI-driven revenue to the creators. While this might not be practical, it provides a useful reference point for creator-centered ecosystems.

\begin{restatable}{proposition}{PropUniquePNERhoOne}
    \label{prop:unique_pne_rho1}
    The subgame $G(1)$ admits a unique PNE, which is an FSE.
\end{restatable}

The result builds on a key structural property: When $\rho = 1$, full sharing becomes a dominant strategy for every creator. \Cref{lem:s_choice_threshold} offers intuition: As long as $\x$ is strictly positive, $\rho=1$ will surpass every possible threshold $\tau_i$, no matter the choice of $x_i$. This guarantees that any PNE must involve full sharing, and existence and uniqueness of the resulting quality profile then follow from \Cref{lem:unique-ese}. Extensions of \Cref{prop:unique_pne_rho1} to other forms of $\qai$ are discussed in \Cref{sec:extensions}.

Generalizing \Cref{prop:unique_pne_rho1} beyond $\rho=1$ requires careful attention. Although \Cref{lem:s_choice_threshold,lem:unique-ese} provide a tight characterization of the (decoupled) best responses in $s$ and $x$, the main challenge is ensuring that no combined deviation in both is beneficial. The following~\Cref{thm:fse-existence} offers a sufficient condition for a value of $\rho$ to induce an FSE.

\begin{restatable}{theorem}{ThmFSEExistence}
    \label{thm:fse-existence}
    Fix any $\rho \in [0,1]$. For every $i \in [n]$, define $x_i^{\max}$ to be the unique positive value satisfying the equation $c_i(z) = \mu z^\gamma$. If $\rho > \max_{i\in[n]} \left\{ \frac{x_i^{\max}}{x_i^{\max} + (1+\alpha) \norm{\mathcal{X}_{-i}(\rho)}} \right\}$, then $(\xr{\rho}, \mathbf{1})$ is the unique FSE of the subgame $G(\rho)$.
\end{restatable}

\begin{proof}[\textbf{\textup{Proof Sketch of \Cref{thm:fse-existence}}}]
    Fix some creator $i$. The proof relies on $x_i^{\max}$, a quality such that any choice of $x_i > x_i^{\max}$ will result in a strictly negative utility for $i$, regardless of its sharing level. When $\mathbf s_{-i} = \mathbf{1}_{-i}$, the threshold from \Cref{lem:s_choice_threshold} becomes $\frac{x_i}{x_i + (1+ \alpha)\xtotmi}$. This threshold is strictly increasing in $x_i$; thus, if $\rho > \frac{x_i^{\max}}{x_i^{\max} + (1+\alpha) \norm{\mathcal{X}_{-i}(\rho)}}$, any profitable deviation of $i$ must involve $s_i=1$. If this holds for every $i$, \Cref{lem:unique-ese} guarantees no profitable deviations in $\x$ exist as well.
\end{proof}

\subsection{Analysis under Power Cost Functions}

Although the condition in \Cref{thm:fse-existence} is theoretically meaningful, it is not constructive. The key challenge is that the quantity $\norm{\mathcal{X}_{-i}(\rho)}$, determining the right-hand side of the inequality, depends on the ESE $\xr{\rho}$, which in turn is a function of $\rho$. Generally, it is not possible to obtain a closed-form expression for $\xr{\rho}$ in terms of $\rho$, making the inequality difficult to evaluate. To enable more precise analysis, we focus on a tractable yet flexible class of \emph{power cost functions}, a standard structure in economic and game-theoretic models of effort and production~\cite{yildiran2015game, roughgarden2005selfish, 10.5555/3692070.3694417}.

\begin{definition}\label{def:power-costs-instance}
    An instance is called \emph{$(n, a_{\min}, a_{\max}, \theta)$-power cost} if it consists of $n$ creators, where each creator $i \in [n]$ has a cost function $c_i(x_i) = a_i x_i^{\theta}$ with common exponent $\theta \in (1,2]$ and cost coefficients $a_i \in [a_{\min}, a_{\max}]$.
\end{definition}

Power cost functions enable more precise analysis. The next proposition reveals that as creator ecosystems grow large, inducing full sharing becomes increasingly tractable, requiring only modest revenue-sharing commitments from the platform.

\begin{restatable}{proposition}{PropAsymptoticStabilityInterval}
    \label{prop:asymptotic_stability_interval}
    For any $a_{\min}, a_{\max}, \theta$ and $w \in (0,1]$, there exists $N(w) \in \mathbb{N}$ such that for all $n > N(w)$, in any $(n, a_{\min}, a_{\max}, \theta)$-power cost instance, the subgame $G(\rho)$ admits an FSE for every $\rho \in (w,1]$.
\end{restatable}

The proof of \Cref{prop:asymptotic_stability_interval} involves the construction of a lower bound on $\norm{\mathcal{X}_{-i}(\rho)}$ that grows with the number of creators $n$. We then obtain the desired result by substituting this lower bound in \Cref{thm:fse-existence} and extracting the necessary $n$.

While \Cref{prop:asymptotic_stability_interval} suggests that platforms can achieve FSE with low revenue sharing, this raises a natural question about content production incentives. Lower allocation parameters may reduce creators' motivation to invest in quality, potentially decreasing overall platform traffic. The next proposition examines this relationship.

\begin{restatable}{proposition}{PropMonotonicityHomogeneous}
    \label{prop:monotonicity_homogeneous}
    Fix any homogeneous power cost instance, i.e., $(n, a, a, \theta)$ for some $a > 0$. Then, $\norm{\xr{\rho}}$ is strictly increasing in $\rho$.
\end{restatable}

\Cref{prop:asymptotic_stability_interval} and \Cref{prop:monotonicity_homogeneous} reveal competing incentives in the choice of $\rho$: Retaining a larger share of existing GenAI revenue versus increasing total revenue through larger traffic. The next section navigates this tradeoff algorithmically.

\section{The Platform's Revenue Maximization Problem}
\label{sec:platform-optimization}

We now turn to the platform's strategic problem of selecting the revenue-allocation parameter $\rho \in [0,1]$ that maximizes its retained revenue, subject to the constraint that creators respond by fully sharing their data. This induces a bilevel optimization problem: The platform acts as a Stackelberg leader, anticipating the equilibrium behavior of creators in response to its choice of $\rho$.

\subsection{Formulation and Relaxation}\label{subsec:formulation_and-relaxation}

The platform's optimization problem, which we term \textbf{P}latform \textbf{O}ptimization, is given by:

\begin{equation}\tag{PO}\label{eq:platform-optimization}
    \sup_{\rho \in [0,1]} \ U_P(\xr{\rho}, \mathbf{1}; f_\rho)  \quad
\text{subject to } (\xr{\rho}, \mathbf{1}) \text{ is an FSE}.
\end{equation}

Further, let $\text{OPT}$ denote the supremum, which is not necessarily attainable by any feasible solution. Solving \ref{eq:platform-optimization} is challenging for two main reasons. First, checking whether a specific $\rho$ induces an FSE requires solving a complex fixed-point problem involving a generally non-convex system of equations, making exact computation of FSE intractable in most cases. Second, even with the ability to verify the above efficiently, optimizing $U_P(\xr{\rho}, \mathbf{1}; f_\rho)$ over $\rho$ analytically is nontrivial due to its implicit dependence on $\xr{\rho}$. Practical approaches might involve a grid search over $\rho$, but this approach does not guarantee the approximate solution will be feasible with respect to the FSE constraint. To that end, we relax the constraint and focus on approximate FSEs. Finally, to allow efficient approximation, we impose additional regularity assumptions on the cost functions by assuming they are $\beta$-strongly convex---a standard condition in strategic optimization \cite{NEURIPS2024_2d52879e, pmlr-v89-goel19a}.

\subsection{Algorithm}\label{subsec:general-algorithm}

\begin{algorithm}[t]
\caption{Approximation Algorithm for \ref{eq:platform-optimization}}
\label{alg:main}
\begin{algorithmic}[1] 
\Require Game instance $\mathcal{I}$, $\delta$, $\eta$
\Ensure $\hat\rho$

\State $\hat \rho \gets 0, \hat U \gets 0$ \label{line:init}
\For{$t \gets 1$ \textbf{to} $\lceil 1/\delta \rceil$} \label{line:loop-start}
    \State Let $\rho \gets \min\left\{\dfrac{t}{\lceil 1/\delta \rceil}, 1 \right\}$ \label{line:rho}
    \State Find $\x$ such that $\norm{\x - \xr{\rho}}_2 \leq \delta$ \label{line:compute-ese}
    \If{$(\x,\mathbf{1})$ is an $\eta$-FSE in $G(\rho)$ and $U_P(\x,\mathbf{1}; f_\rho) > \hat U$} \label{line:check-fse}
        \State $\hat U \gets U_P(\x,\mathbf{1}; f_\rho), \hat \rho \gets \rho$ \label{line:update-best}
    \EndIf
\EndFor \label{line:loop-end}

\State \Return $\hat\rho$ \label{line:return}
\end{algorithmic}
\end{algorithm}

We propose a simple search-based procedure for computing an approximate solution of \ref{eq:platform-optimization}, illustrated in \Cref{alg:main}. It receives as input an instance $\mathcal I$ and two parameters, $\delta$ and $\eta$. We carefully determine the values of $\delta$ and $\eta$ to reach a certain approximation level; we elaborate later in \Cref{thm:alg_main}. \Cref{alg:main} performs a grid search over the interval $[0,1]$, evaluating a discrete set of allocations parameters at resolution $\delta$. Starting with initialized variables (Line~\ref{line:init}), the algorithm iterates through candidate values $\rho = \nicefrac{t}{\floor{1/\delta}}$ (Lines~\ref{line:loop-start}--\ref{line:loop-end}). For each candidate, it approximates the corresponding ESE (Line~\ref{line:compute-ese}) and verifies whether it constitutes an $\eta$-FSE (Line~\ref{line:check-fse}). If this stability criterion is met, it updates the current best solution whenever the resulting platform revenue exceeds the previous maximum (Line~\ref{line:update-best}). The algorithm ultimately returns the best-performing $\rho$ across all valid candidates (Line~\ref{line:return}). The next result establishes that the algorithm returns a near-optimal solution with provable guarantees on both approximation quality and runtime.

\begin{restatable}{theorem}{ThmAlgMain}
    \label{thm:alg_main}
    Fix any instance $\mathcal{I}$ with strongly convex cost functions. There exist constants $A,B$ such that for any $\varepsilon > 0$, the execution of \Cref{alg:main} with $\delta = \frac{\varepsilon}{A}$ and $\eta = \frac{\varepsilon}{4}$ can be implemented in time $O\left(\frac{A^2 B}{\varepsilon^2}\right)$. Moreover, its output $\hat{\rho}$ induces an $\varepsilon$-FSE and satisfies $U_P(\xr{\hat{\rho}}, \mathbf{1}; f_{\hat{\rho}}) \geq \mathrm{OPT} - \varepsilon$.
\end{restatable}

The exact specification of $A=A(\mathcal I)$ and $B=B(\mathcal I)$ involves somewhat technical details; however, both constants depend on upper and lower bounds of the equilibrium quality $\norm{\xr{\rho}}$, which under mild regularity assumptions are themselves polynomial in the number of creators $n$. 

\subsection{Proof Sketch of \Cref{thm:alg_main}}

\begin{figure}
    \centering
    \includegraphics[scale=1.]
    {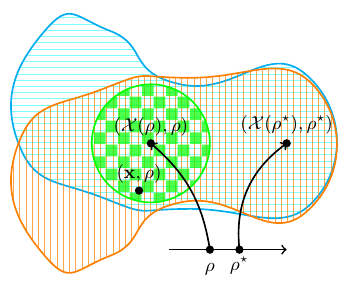}
    \caption{Illustrating the proof sketch.}
    \label{fig:approx}
\end{figure}

Fix any instance $\mathcal{I}$ and $\varepsilon > 0$, and let $\delta,\eta$ denote the parameters in \Cref{thm:alg_main} that are passed to the algorithm. Next, focus on any feasible solution $\rho^\star$ to \ref{eq:platform-optimization}, and let $\rho$ denote the closest value to $\rho^\star$ the algorithm considers in Line~\ref{line:rho}. Notably, $\abs{\rho - \rho^\star} \leq \delta$. \Cref{fig:approx} illustrates how $\rho^\star$ is mapped to the unique ESE of $G(\rho^\star)$, $\xr{\rho^\star}$.  Similarly, $\rho$ is mapped to $\xr{\rho}$, the unique ESE of $G(\rho)$.

\paragraph{Step 1: Beneficial deviations} We define the function $\phi$ such that $\phi(\mathbf{\tilde x}, \tilde \rho)$ is the maximal gain from any creator's deviation from $(\mathbf{\tilde x},\mathbf{1})$ in $G(\tilde \rho)$. Because $\rho^\star$ is feasible, $(\xr{\rho^\star}, \mathbf{1})$ is the FSE of  $G(\rho^\star)$; therefore, $\phi(\xr{\rho^\star},\rho^\star)=0$. Using technical lemmas, we show that $\phi$ is Lipschitz continuous and bound its Lipschitz constant. Additionally, we use the Implicit Function Theorem to bound the Lipschitz constant of $\xr \cdot$, circumventing its lack of a closed-form expression. The blue region in \Cref{fig:approx} (filled with horizontal lines) depicts the set of pairs $(\tilde {\mathbf x},\tilde \rho)$ such that $\phi(\tilde {\mathbf x},\tilde \rho) \leq \eta < \varepsilon$. A careful selection of $\delta$ implies that $\rho$ and $\rho^\star$ are close, and so are $\xr{\rho}$ and $\xr{\rho^\star}$; therefore, $(\xr{\rho}, \rho)$ lies within the blue region, implying that the condition in Line~\ref{line:check-fse} holds and $\rho$ is an $\varepsilon$-FSE of $G(\rho)$.

\paragraph{Step 2: Revenue approximation}
Similarly, we show that the function $U_P(\tilde{\mathbf x},\mathbf{1} ; f_{\tilde \rho})$ is Lipschitz continuous and upper bound its Lipschitz constant. This ensures that there exists a neighborhood around $(\xr{\rho^\star},\rho^\star)$ with pairs $(\tilde {\mathbf x},\tilde \rho)$ that obtain $\varepsilon$-close revenue. We let the orange region in \Cref{fig:approx}, filled with vertical lines,  depict all such pairs. Picking $\delta$ appropriately ensures that $(\xr{\rho},\rho)$ belongs to the orange region. Thus, in case $\rho^\star$ achieves a revenue arbitrarily close to $\text{OPT}$, at least one nearby point $\rho$ on the discrete grid yields a revenue within $\varepsilon$ of optimal.

\paragraph{Step 3: Resolving output mismatch}
Although we have demonstrated in Steps 1 and 2 that $\xr{\rho}$ satisfies both the $\varepsilon$-FSE and revenue guarantees when $\rho$ lies sufficiently close to $\rho^\star$, the algorithm does not evaluate $\xr{\rho}$ directly. Instead, it computes an approximate profile $\mathbf{x}$ such that $\norm{\mathbf{x} - \xr{\rho}}_2 \leq \delta$ (recall Line~\ref{line:compute-ese}). This creates a mismatch: Even if $\xr{\rho}$ is a valid and near-optimal solution, the algorithm might overlook it because the nearby profile $\mathbf{x}$ fails one of the checks in Line~\ref{line:check-fse}. To resolve this, let the green region (with a checkerboard pattern) in \Cref{fig:approx} denote all pairs $(\tilde{\mathbf{x}}, \rho)$ such that $\norm{\tilde{\mathbf{x}} - \xr{\rho}}_2 \leq \delta$. Unlike the orange and blue regions, the green region is a Euclidean ball. Note that $(\mathbf{x}, \rho)$ is in this region for any arbitrary $\mathbf x$ that may have been found in Line~\ref{line:compute-ese}. By choosing $\delta$ small enough so that the green ball is fully contained within both the blue and orange regions, we ensure that all such $\mathbf{x}$ exhibit the same stability and revenue guarantees as $\xr{\rho}$, and vice versa. Thus, the algorithm safely evaluates $\mathbf{x}$ in place of $\xr{\rho}$, and its output $\rho$ satisfies the desired approximation guarantees.

\paragraph{Efficient computation} First, we prove that the enforced sharing game belongs to the class of strongly monotone games~\cite{rosen1965existence}. This property allows us to apply multi-agent mirror descent to find the approximate quality profile required in Line~\ref{line:compute-ese} in $O(1/\delta)$ iterations~\cite{10.5555/3327345.3327469}, each requiring $O(n)$ time to evaluate creator gradients. The $\delta$-FSE condition can also be efficiently verified. For each creator, it suffices to solve a single one-dimensional optimization: Maximize utility over $x_i \ge 0$ with $s_i=0$, holding the strategies of all other creators fixed. This follows from \Cref{lem:s_choice_threshold}, which ensures that any profitable deviation is dominated by a deviation that has $s_i=1$ or $s_i=0$. The former is implicitly covered by the approximation guarantee of the computed profile $\x$, and the latter can be efficiently approximated, as we demonstrate that the utility function is Lipschitz continuous under each fixed value of $s$. We choose $\delta'$ small enough so that approximation errors from the ESE computation and the deviation check are jointly bounded by $\delta$. After finding the maximal utility under $s_i=0$, it remains only to test whether the resulting gain exceeds $\delta$, thereby completing the verification.

\subsection{Refinements for Power Cost Functions}\label{subsec:refinements-power-cost}

Having established the correctness and efficiency of \Cref{alg:main}, we now examine its behavior under specific cost structures. For the class of power cost functions (recall \Cref{def:power-costs-instance}), we obtain explicit bounds on the algorithm's runtime dependence on the number of creators.

\begin{restatable}{proposition}{PropPowerCostScaling}
    \label{prop:power_cost_scaling}
    Consider any $(n, a_{\min}, a_{\max}, \theta)$-power cost instance $\mathcal{I}$ with $a_{\min}, a_{\max} = O(1)$. Then the constants $A(\mathcal{I})$ and $B(\mathcal{I})$ from \Cref{thm:alg_main} satisfy $A(\mathcal{I}) = O(n^{2})$ and $B(\mathcal{I}) = O(n)$.
\end{restatable}

Although \Cref{alg:main} offers an efficient method for achieving an anytime approximation of $\mathrm{OPT}$, its runtime inevitably depends on $n$. Thus, for large platforms, it may be preferable to use a fixed allocation rule rather than find one algorithmically, as the latter option might be impractical. It turns out that for many power cost instances, there is a closed-form choice of $\rho$ that induces an FSE and provides a multiplicative approximation of the platform's optimal revenue. Interestingly, this choice of $\rho$ does not depend on $n$ or the coefficients $\left(a_i\right)_{i=1}^{n}$.

\begin{restatable}{proposition}{PropApproxFixedRho}
    \label{prop:approx_fixed_rho}
    Consider large enough $n, \mu$ and any power cost instance such that $(\alpha+1)\gamma > \theta$. Then, choosing $\rho = \frac{\alpha \gamma + \gamma - \theta}{\alpha \theta}$ ensures that $\rho$ induces an FSE, and yields platform revenue satisfying $U_P(\xr{\rho}, \mathbf{1} ; f_{\rho}) \ge (2\alpha + 2)^{-\frac{\gamma}{\theta - \gamma}} \cdot \mathrm{OPT}$.
\end{restatable}

Notably, larger values of $\theta$ decrease the expression of $\rho$ in \Cref{prop:approx_fixed_rho}, and, perhaps counterintuitively, larger values of $\alpha$ increase it. We detail more about these phenomena in \Cref{sec:experiments}, where they also arise from simulations on the optimal allocation parameter. The proof of \Cref{prop:approx_fixed_rho} builds on a lower bound of $\norm{\mathcal{X}(\rho)}$ from~\cite{10.5555/3692070.3694417}, which we adapt to our setting. We then derive the closed-form expression for $\rho$ by substituting this lower bound into the platform's revenue function and optimizing it analytically. \Cref{prop:asymptotic_stability_interval} ensures that the resulting profile is a valid FSE, while the approximation factor reflects the gap between the bound and the true value of $\norm{\mathcal{X}(\rho)}$.

\section{Experiments}\label{sec:experiments}

\begin{figure*}[t]
    \centering
    \includegraphics[width=0.95\textwidth]{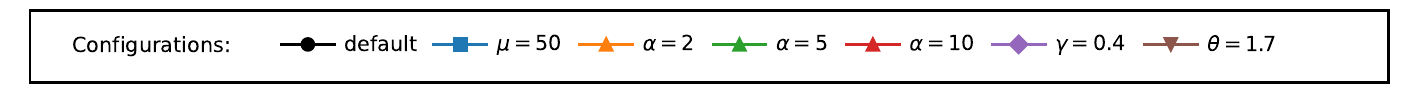}

    \begin{minipage}{0.9\textwidth}
        \centering
        \begin{subfigure}[t]{0.32\textwidth}
            \includegraphics[width=\textwidth]{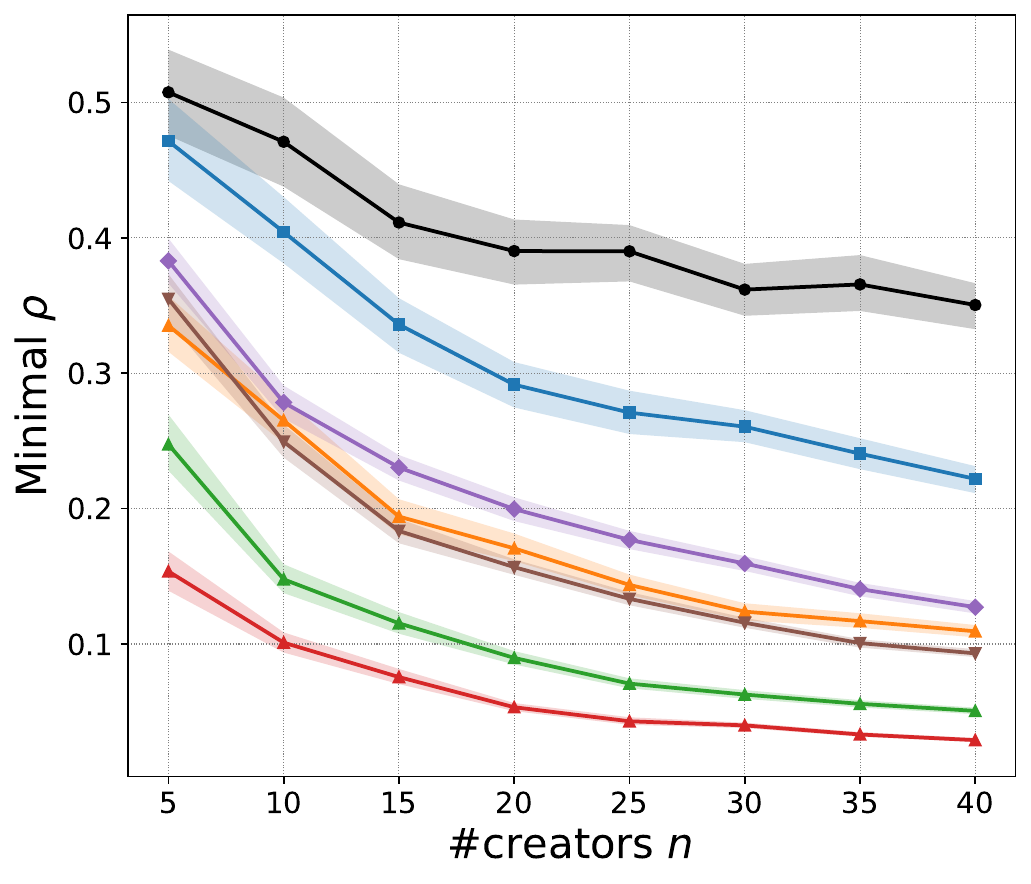}
            \caption{Min. $\rho$ inducing $\varepsilon$-FSE}
            \label{fig:subfig1}
    \end{subfigure}
        \hfill
        \begin{subfigure}[t]{0.32\textwidth}
            \includegraphics[width=\textwidth]{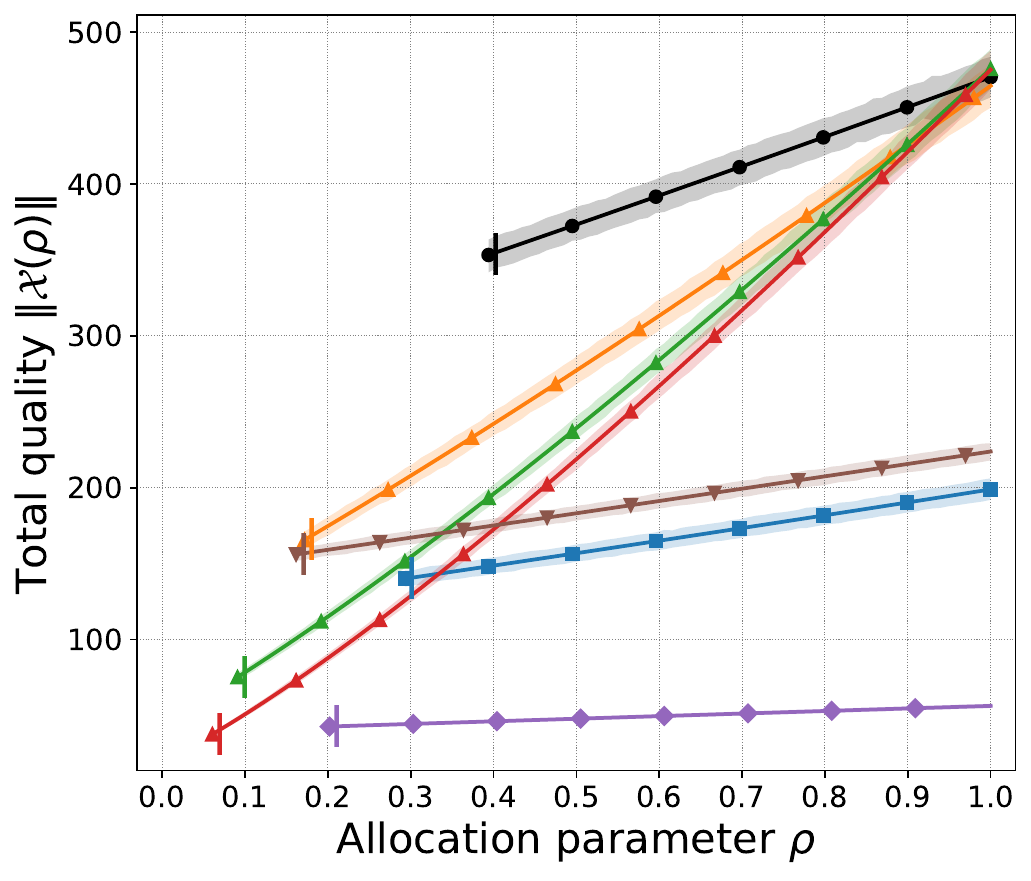}
            \caption{$\rho$ vs. $\norm{\xr{\rho}}$ in FSE}
            \label{fig:subfig2}
        \end{subfigure}
        \hfill
        \begin{subfigure}[t]{0.32\textwidth}
            \includegraphics[width=\textwidth]{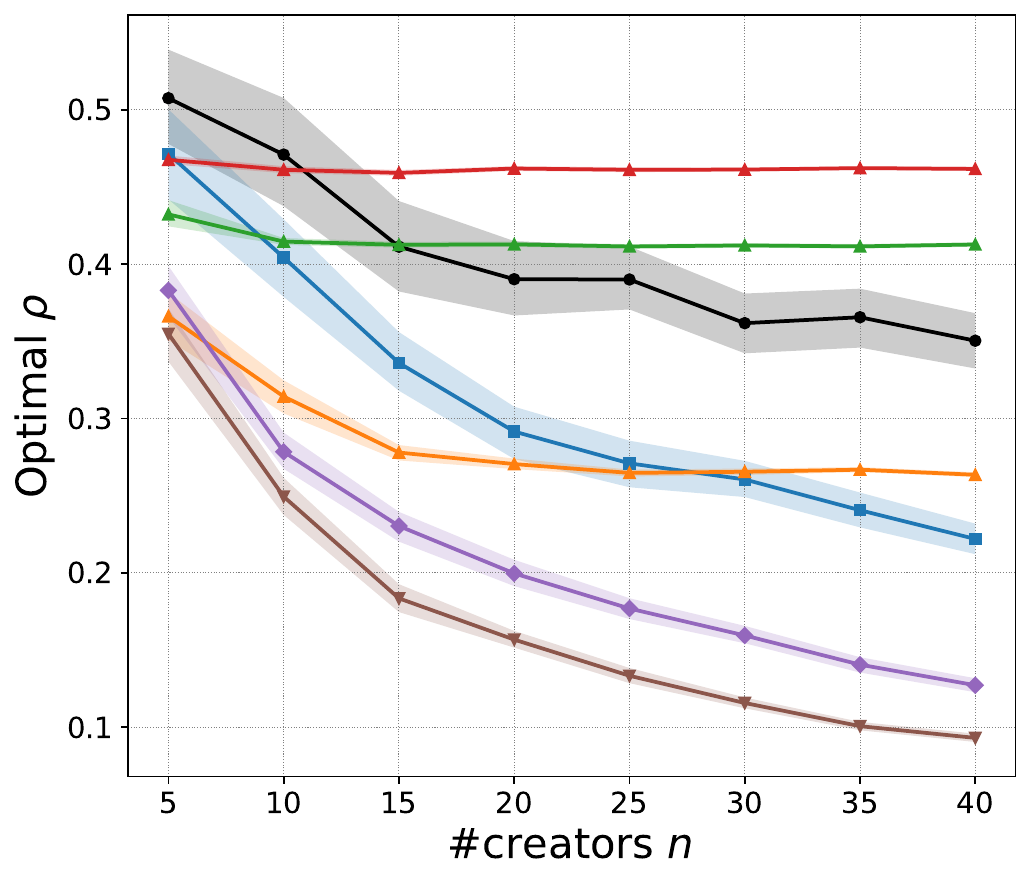}
            \caption{Optimal allocation parameter $\rho$}
            \label{fig:subfig3}
        \end{subfigure}
    \end{minipage}

    \begin{minipage}{0.95\textwidth}
        \centering
        \begin{subfigure}[t]{0.32\textwidth}
            \includegraphics[width=\textwidth]{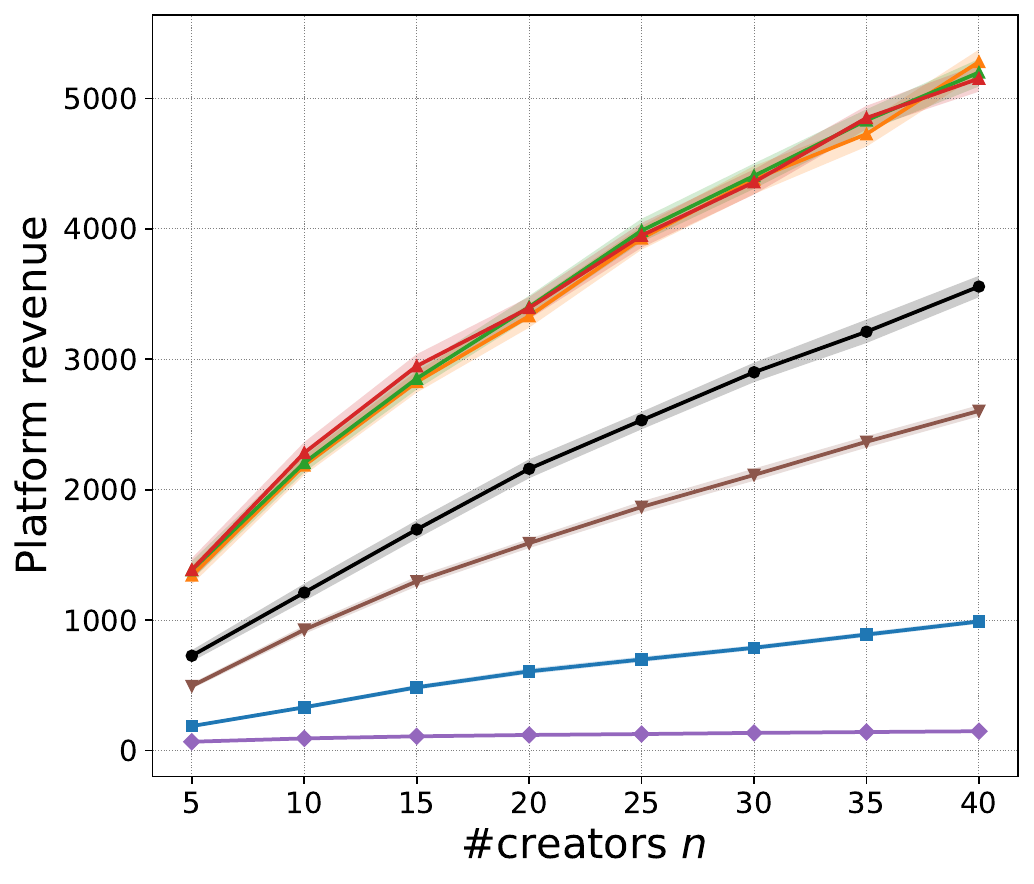}
            \caption{$U_P$ at optimal allocation}
            \label{fig:subfig4}
        \end{subfigure}
        \hfill
        \begin{subfigure}[t]{0.32\textwidth}
            \includegraphics[width=\textwidth]{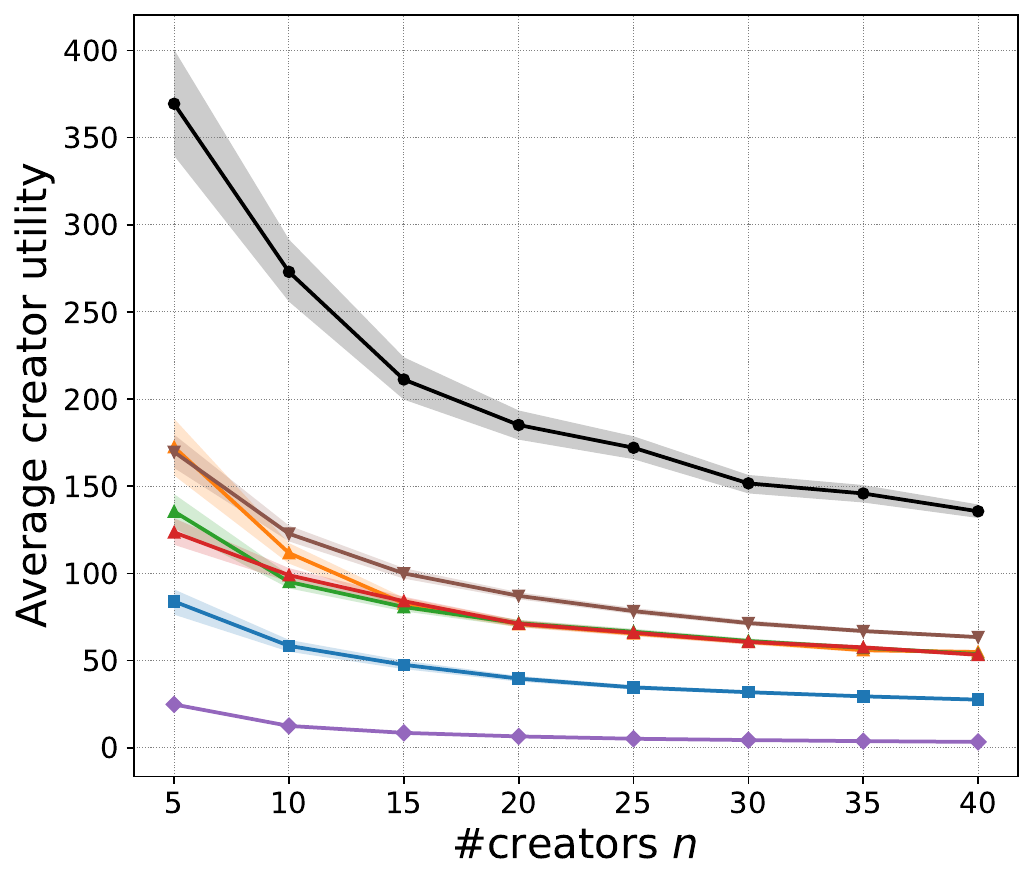}
            \caption{Average $U_i$ at optimal allocation}
            \label{fig:subfig5}
        \end{subfigure}
        \hfill
        \begin{subfigure}[t]{0.32\textwidth}
            \includegraphics[width=\textwidth]{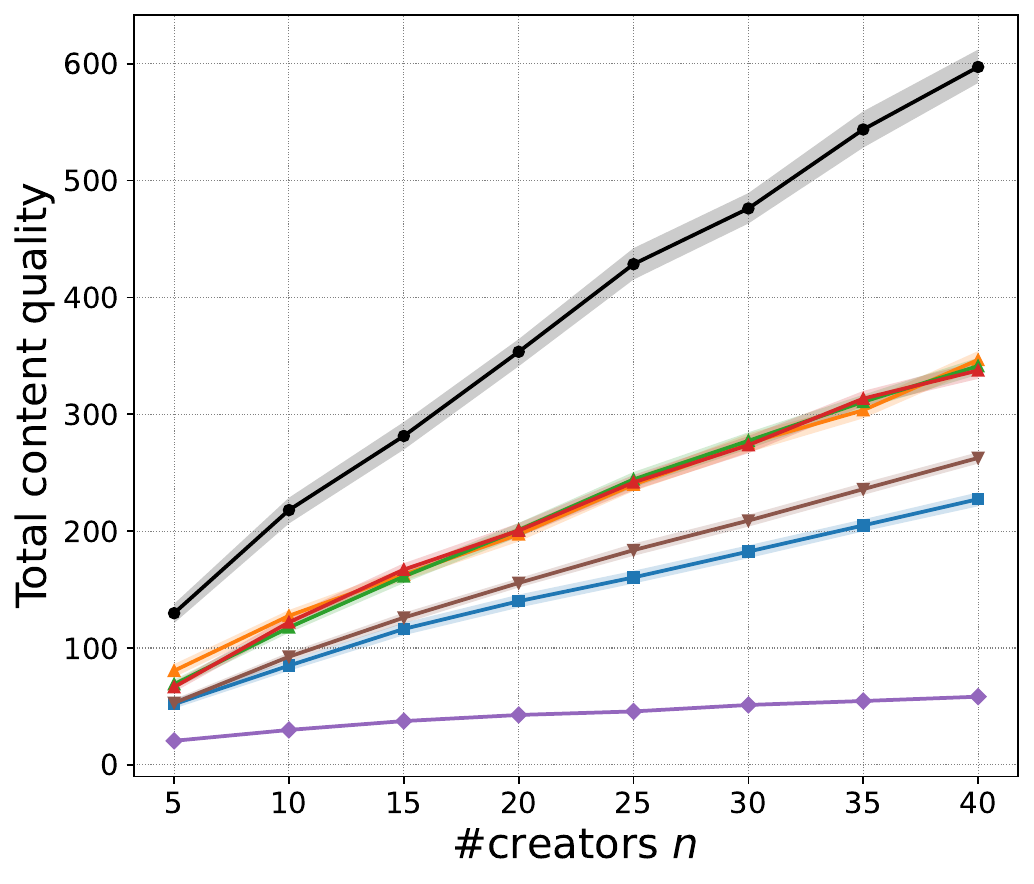}
            \caption{$\norm{\xr{\rho}}$ at optimal allocation}
            \label{fig:subfig6}
        \end{subfigure}
    \end{minipage}

    \caption{Revenue-allocation effects on equilibria, quality, revenue, and utilities.}
    \label{fig:combined}
\end{figure*}

In this section, we conduct experiments to understand equilibrium outcomes and, more generally, how the allocation parameter $\rho$ impacts creator behavior. We consider cost functions of the form $c_i(x) = a_i x^\theta$. In the default configuration, the parameters are $n=20$, $\gamma=0.8$, $\mu=100$, $\theta=1.5$, $\alpha=0.5$, and the coefficients $\{a_i\}_{i=1}^{n}$ are independently drawn from the uniform distribution $Uni[1,10]$. In addition to the default configuration, we varied one parameter at a time while keeping all others fixed at their default values to assess its individual impact. Each point in \Cref{fig:combined} represents the average outcome across 150 independent game instances, with error bars indicating 95\% confidence intervals obtained via bootstrapping. We discretized the possible values of $\rho$ into 100 equally-spaced points in the range $[0,1]$. For each $\rho$, we check if it induces an $\varepsilon$-FSE in $G(\rho)$ for $\varepsilon=10^{-4}$. We ran all the executions on a standard laptop, completing the entire process in a few hours. We include all implementation details in \Cref{sec:sim_details}.

\paragraph{Minimal revenue-allocation threshold}
\Cref{fig:subfig1} illustrates how the minimal allocation parameter $\rho$ required for inducing an FSE varies with ecosystem parameters. All curves decrease with the number of creators, as \Cref{thm:fse-existence} implies: The content of each creator becomes insignificant w.r.t. all other created content, so withholding has a high alternative cost. Furthermore, more powerful GenAI (larger $\alpha$) enables lower $\rho$ choices for stable FSE, as the curves associated with a change in $\alpha$ (red, green, orange, and black curves) are ordered. The black curve $(\mu = 100)$ lies above the blue curve $(\mu = 50)$, suggesting that creators demand more compensation as traffic increases.

\paragraph{Content quality vs. allocation parameter}
\Cref{fig:subfig2} shows how different values of $\rho$ affect $\norm{\xr{\rho}}$, the total content quality in FSE. Since not all values of $\rho$ induce an FSE (see \Cref{fig:subfig1}), each line begins at the minimal value of $\rho$ for which an FSE exists under the corresponding parameter setting. As expected, larger values of $\rho$ encourage creators to produce more content regardless of the setting, as they get better compensation from the platform. Notably, more powerful GenAI (larger $\alpha$) reduces creator content production. Furthermore, different values of $\alpha$ align as $\rho$ approaches $1$. To see why, recall that if the platform allocates most GenAI-driven revenue to creators, creators do not suffer from the competition with GenAI.

\paragraph{Optimal allocation}
\Cref{fig:subfig3} presents the platform's optimal allocation parameter. We see that all the curves are non-increasing, implying that gradually smaller allocation parameters are enough. Focusing on the impact of high GenAI quality parameter $\alpha$ (red, green, and orange curves), we see an interesting pattern: The allocation parameter plateaus as we increase the number of creators. That is, the platform's optimal choice tends toward higher $\rho$ regardless of ecosystem size, with larger $\alpha$ corresponding to larger optimal $\rho$. We find this surprising as \Cref{fig:subfig1} indicates that more powerful GenAI permits lower values of the allocation parameter to induce an FSE. This counterintuitive finding indicates that platforms with more powerful GenAI must compensate creators more generously in terms of the allocation. For weaker GenAI, the optimal allocation parameter is similar to the minimal allocation threshold required to induce an FSE: For $\alpha=0.5$, the black curves in \Cref{fig:subfig1} (minimal parameter) and \Cref{fig:subfig3} (optimal parameter) are almost identical.

\paragraph{Platform-optimal equilibrium outcomes}
Figures~\ref{fig:subfig4}, \ref{fig:subfig5}, and~\ref{fig:subfig6} examine the platform's revenue, creator utilities, and total quality under the platform-optimal allocation. As expected, larger creator populations increase platform revenue and total quality (increasing curves in \Cref{fig:subfig4} and \Cref{fig:subfig6}) while decreasing the average creator returns (decreasing curves in \Cref{fig:subfig5}). Interestingly, powerful GenAI (red, green, and orange curves) yield nearly identical patterns in all graphs, suggesting that beyond a certain threshold, more powerful GenAI does not significantly impact economic outcomes. Moreover, weak GenAI (black curve) is bad for the platform compared to powerful GenAI, and good for creator utility and total content quality.

\section{Extensions}\label{sec:extensions}

In this section, we provide several natural extensions of our framework and discuss their implications on our main results.

\subsection{Multiple Topics}

The following extension enables creators to allocate different efforts across various content topics, while the platform maintains a single revenue-allocation rule. Consider $K$ distinct content topics. Each creator $i$ now chooses a quality vector $\mathbf{x}_i = (x_{i,1}, \ldots, x_{i,K}) \in [0,\infty)^K$ and a sharing vector $\mathbf{s}_i = (s_{i,1}, \ldots, s_{i,K}) \in [0,1]^K$, where $x_{i,k}$ and $s_{i,k}$ represents their quality and sharing level in topic $k$, respectively. The cost function, which now gets $\mathbf{x}_i$ as input, is convex but need not be separable across topics, i.e., $c_i(\mathbf{x}_i) = \sum_k c_{i,k}(x_{i,k})$, for most results to hold. 

Each topic $k$ generates independent user traffic $T_k$, and the GenAI system's quality is $Q_{\text{AI},k}$; both depend only on the restriction of $\mathbf{x}$ and $\mathbf{s}$ to topic $k$ and have the same functional forms as in the single-topic case. The platform commits to a single allocation parameter $\rho \in [0,1]$ across all topics. Under proportional allocation, creator $i$'s revenue from topic $k$ is given by
\[
    R_{i,k}(\mathbf{x}, \mathbf{s}; f_\rho) =  \frac{T_k(\mathbf{x}) x_{i,k}}{\sum_{j=1}^n x_{j,k} + Q_{\text{AI},k}} + \frac{T_k(\mathbf{x}) Q_{\text{AI},k}}{\sum_{j=1}^n x_{j,k} + Q_{\text{AI},k}} \cdot \frac{\rho x_{i,k} s_{i,k}}{\sum_{j=1}^n x_{j,k} s_{j,k}},
\]
and its overall utility is $U_i(\mathbf{x}, \mathbf{s}; f_\rho) = \sum_{k=1}^K R_{i,k}(\mathbf{x}, \mathbf{s}; f_\rho) - c_i(\mathbf{x}_i)$.

\paragraph{Results} \Cref{lem:s_choice_threshold} extends to each topic independently: For fixed quality choices, the optimal sharing decision $s_{i,k}$ for topic $k$ follows the same threshold structure in \Cref{{eq:sharing-threshold-form}}, with $\tau_{i,k} = \frac{x_{i,k}}{x_{i,k} + \sum_{j \neq i} x_{j,k} + \alpha_k \sum_{j \neq i} x_{j,k} s_{j,k}}$. \Cref{lem:unique-ese} holds as well---since the only interplay between different topics is through the cost function, the same analysis still applies. \Cref{thm:fse-existence} requires some modifications, but the general idea remains the same: A large enough $\rho$ admits FSE. The main technical difference is establishing per-topic rational quality bounds $\mathbf{x}_{i,k}^{\max}$ when the cost functions couple topics together. The approach of \Cref{alg:main} generalizes to this setting too, but testing the condition in Line~\ref{line:check-fse} becomes computationally intensive, as we need to check each of the $2^K-1$ possible sharing deviations. This can be avoided when the cost functions are separable across topics.

\subsection{Prior GenAI Data}

Next, we allow the GenAI system to possess a fixed body of prior data, independent of the creators' sharing contributions. Formally, we introduce a non-strategic creator indexed by $0$ with fixed quality $x_0 > 0$ and full sharing level $s_0 = 1$. The quality of the GenAI system then becomes $Q_{\text{AI}}(\mathbf{x}, \mathbf{s}) = \alpha \left( x_0 + \sum_{j=1}^n x_j s_j \right)$, so that even if no strategic creator shares content, the GenAI still attains positive quality $\alpha x_0$. This modification also affects user traffic. Essentially, we assume that there is another source of $x_0$ body of content in the platform that is always available to the GenAI system.

\paragraph{Results} From creator $i$'s perspective, the new creator is just another content contributor with fixed quality and full sharing; thus, the threshold from \Cref{lem:s_choice_threshold} updates accordingly to 
\[
    \tau_i = \frac{x_i}{x_i + \xtotmi + \alpha \mathbf{x}^{\top}_{-i} \mathbf{s}_{-i}+(1+\alpha)x_0}.
\]
Notably, prior data $x_0 > 0$ strictly decreases $\tau_i$, making full sharing more attractive for creators. The proof of \Cref{lem:unique-ese} requires a minor modification, as the dummy creator's strategy is fixed. Nonetheless, the resulting subgame among the creators maintains the DSC property. All other results extend directly to this setting.

\subsection{Alternative Optimization Objectives}

While our analysis in \Cref{sec:platform-optimization} focused on maximizing platform revenue (recall \eqref{eq:platform-optimization}), there are other desirable objectives to consider. For instance, a platform may wish to balance its retained value with total content production and creator welfare. Our framework naturally extends to such cases, where the target function is Lipschitz continuous in the quality profile $\mathbf{x}$ and the allocation parameter $\rho$. The platform’s problem becomes  
\[
\sup_{\rho \in [0,1]} \ f(\xr{\rho}, \rho) 
\quad \text{subject to } (\xr{\rho}, \mathbf{1}) \text{ is an FSE}.
\]  
Natural choices for $f$ include total content production $f(\mathbf{x}, \rho) = \|\mathbf{x}\|_1$, creator welfare $f(\mathbf{x}, \rho) = \sum_{i=1}^n U_i(\mathbf{x}, \mathbf{1}; f_\rho)$, or regularized combinations such as 
\[
f(\mathbf{x}, \rho) = (1-\lambda) U_P(\mathbf{x}, \mathbf{1}; f_\rho) + \lambda \|\mathbf{x}\|_1.
\]

\paragraph{Results} The same algorithmic approach from \Cref{alg:main} applies directly to these generalized objectives, with the approximation guarantees of \Cref{thm:alg_main} carrying over after adjusting the constants $A$ and $B$ to account for the Lipschitz properties of $f$.

\subsection{GenAI-Dependent Traffic}

In our main model, user traffic $\Tx$ depends solely on the quality of human-generated content. However, high-quality AI-generated content may also attract users to the platform. To capture this effect, we extend our traffic model to incorporate both human and AI contributions:
\[
    T(\x, \s) = \mu \left( \xtot + \qai(\x, \s) \right)^\gamma.
\]
All other components of the model remain unchanged.

\paragraph{Results} Under this extended traffic model, the creator's optimal sharing decision retains its threshold structure, with an adjusted threshold given by $
\tau_i = \frac{1}{\alpha} \left[ \left( \frac{\xtot + \alpha \left(x_i + \sum_{j\neq i}^n x_j s_j\right)}{\xtot + \alpha \sum_{j \neq i} x_j s_j} \right)^{1-\gamma} - 1 \right]$.
Notably, larger values of $\gamma$ decrease $\tau_i$, making sharing more attractive. This occurs because higher traffic elasticity amplifies the positive externality that sharing exerts on the platform's traffic.

The analysis of ESE remains tractable under this extension. In the enforced sharing game where $\s = \mathbf{1}$, the traffic function simplifies to $T(\x, \mathbf{1}) = \mu (1 + \alpha)^\gamma \xtot^\gamma$. This expression is structurally identical to our original traffic model, differing only by a multiplicative scaling of the baseline parameter $\mu$ by $(1 + \alpha)^\gamma$. Consequently, \Cref{lem:unique-ese} applies directly: The subgame $G(\rho)$ under enforced sharing admits a unique ESE. Since the utility function at $\s = \mathbf{1}$ retains its structural properties, our main results carry over to this setting with adjusted constants. Specifically, the FSE existence condition (\Cref{thm:fse-existence}) holds by redefining $x_i^{\max}$ as the unique positive solution to $c_i(x) = \mu (1 + \alpha)^\gamma x^\gamma$. The approximation guarantees of \Cref{thm:alg_main} follow similarly.

\subsection{GenAI Quality Beyond Linearity}
\label{subsec:genai-quality-beta}

In our baseline model, the quality of the GenAI system is linear in the aggregate amount of shared content, given by $\qai(\mathbf{x}, \mathbf{s}) = \alpha \sum_{i=1}^n x_i s_i$. This specification is both analytically convenient and captures the essence of the strategic dynamics in the system. However, in practice, GenAI systems often exhibit diminishing marginal returns as the amount of training data grows~\cite{kaplan2020scalinglawsneurallanguage}. To capture this effect, we consider the following generalization of $\qai$:
\[
    \qai(\mathbf{x}, \mathbf{s}) = \alpha \left( \sum_{i=1}^n x_i s_i \right)^{\beta},
    \qquad \beta \in [0,1].
\]
This form, introduced by~\citet{10.5555/3692070.3694417}, captures diminishing returns in GenAI quality. Notice that when $\beta = 1$, this formulation reduces to our original model.

\paragraph{Theoretical results.}
When $\beta < 1$, the enforced sharing game is no longer concave and therefore does not satisfy Rosen's DSC condition. As a consequence, we lose general guarantees on the existence and uniqueness of an ESE. Indeed, in numerical experiments, we observe instances that admit multiple approximate ESEs, which can differ in the induced quality profiles. From a more positive perspective, several qualitative insights from our baseline model remain valid. First, \Cref{prop:unique_pne_rho1} transfers directly: When $\rho = 1$, full sharing is a dominant strategy for every creator, and the subgame admits a unique PNE that is an FSE. Second, a weaker analogue of \Cref{lem:s_choice_threshold} continues to hold. In particular, when no creator shares content, creator $i$ has an incentive to deviate and share whenever $\rho > \frac{x_i}{\sum_{j=1}^n x_j}$. This mirrors the baseline threshold and gives rise to a prisoner's dilemma-like structure in large systems.

\paragraph{Simulation methodology.}
To study equilibria under the generalized GenAI quality function with $\beta < 1$, we rely on numerical methods. We focus on the enforced sharing game, where all creators are required to share ($\mathbf{s}=\mathbf{1}$), and compute an approximate ESE via best-response dynamics. Starting from an arbitrary initial quality profile, creators update their quality choices sequentially. Each best response in $x_i$ is computed via exact one-dimensional optimization over the interval $[0, x_i^{\max}]$, where $x_i^{\max}$ is the analytical upper bound derived in the paper. To ensure robust convergence, we employ a damping factor that interpolates between the current quality and the computed best response. We declare convergence once the maximum change in any creator's quality falls below a tolerance $\varepsilon$. Across all instances we tested, these dynamics converged quickly to a stable quality profile, which we treat as an $\varepsilon$-ESE.

To test whether a converged profile constitutes an $\varepsilon$-FSE, we explicitly check unilateral deviations in the sharing decision. For each creator, we consider both $s_i = 0$ and $s_i = 1$, and for each case, compute the optimal quality choice via one-dimensional optimization. If no such deviation yields a utility gain exceeding $\varepsilon$, we classify the profile as an $\varepsilon$-FSE. Notably, although multiple approximate equilibria may exist in this setting, our simulations reveal that they exhibit similar qualitative properties across different initializations.

\begin{figure*}[t]
    \centering
    \includegraphics[width=0.95\textwidth]{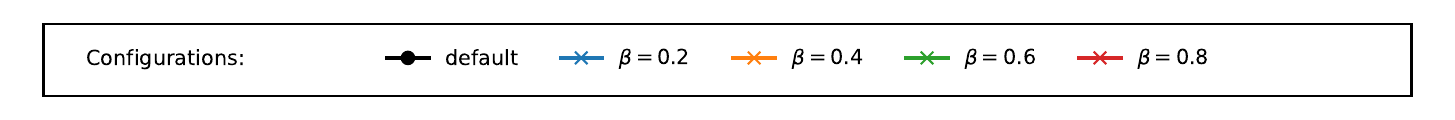}
    \vspace{0.3cm}

    \begin{minipage}{0.9\textwidth}
        \centering
        \textbf{$\alpha = 0.5$}
        \vspace{0.2cm}

        \centering
        \begin{subfigure}[t]{0.32\textwidth}
            \includegraphics[width=\textwidth]{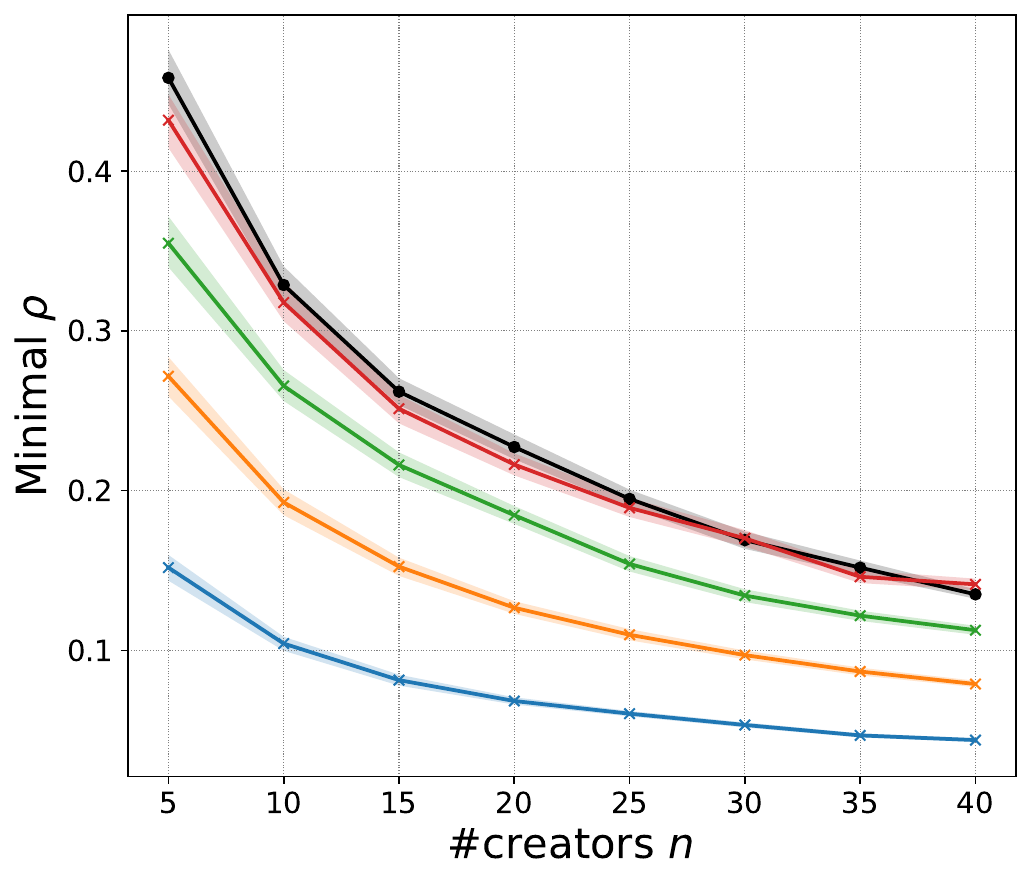}
            \caption{Min. $\rho$ inducing $\varepsilon$-FSE}
            \label{fig:beta_subfig1}
        \end{subfigure}
        \hfill
        \begin{subfigure}[t]{0.32\textwidth}
            \includegraphics[width=\textwidth]{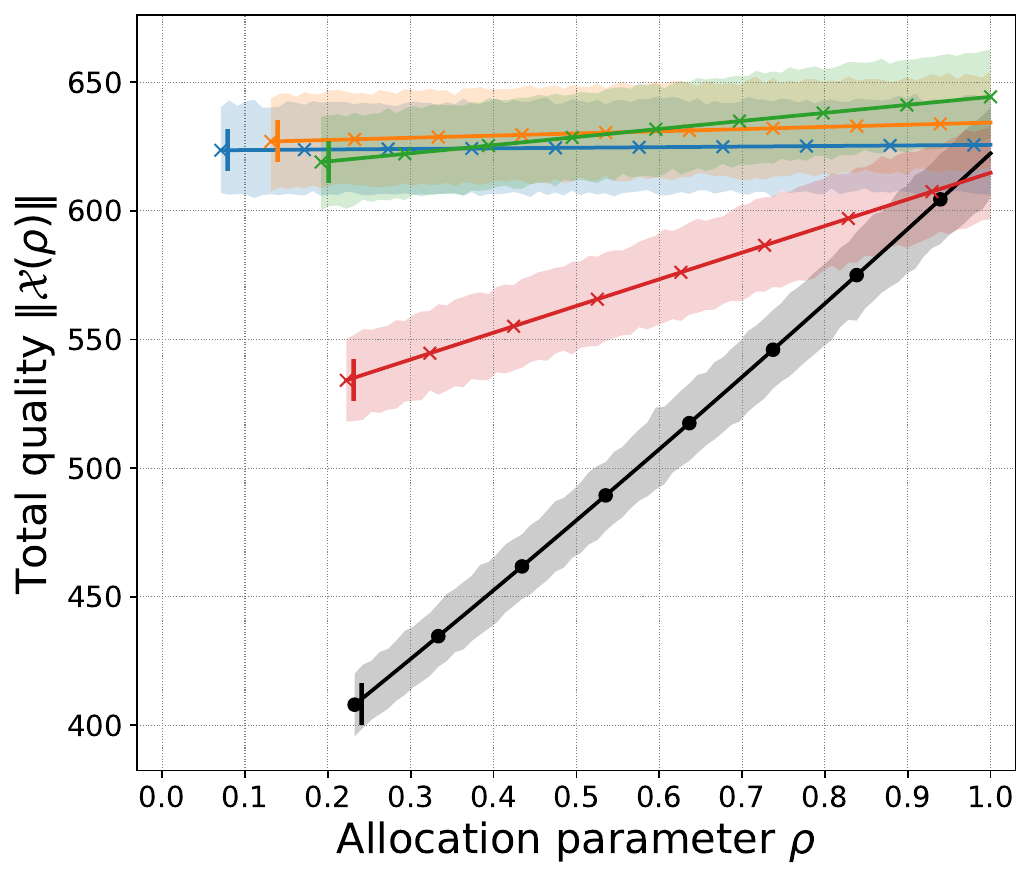}
            \caption{$\rho$ vs. $\norm{\xr{\rho}}$ in FSE}
            \label{fig:beta_subfig2}
        \end{subfigure}
        \hfill
        \begin{subfigure}[t]{0.32\textwidth}
            \includegraphics[width=\textwidth]{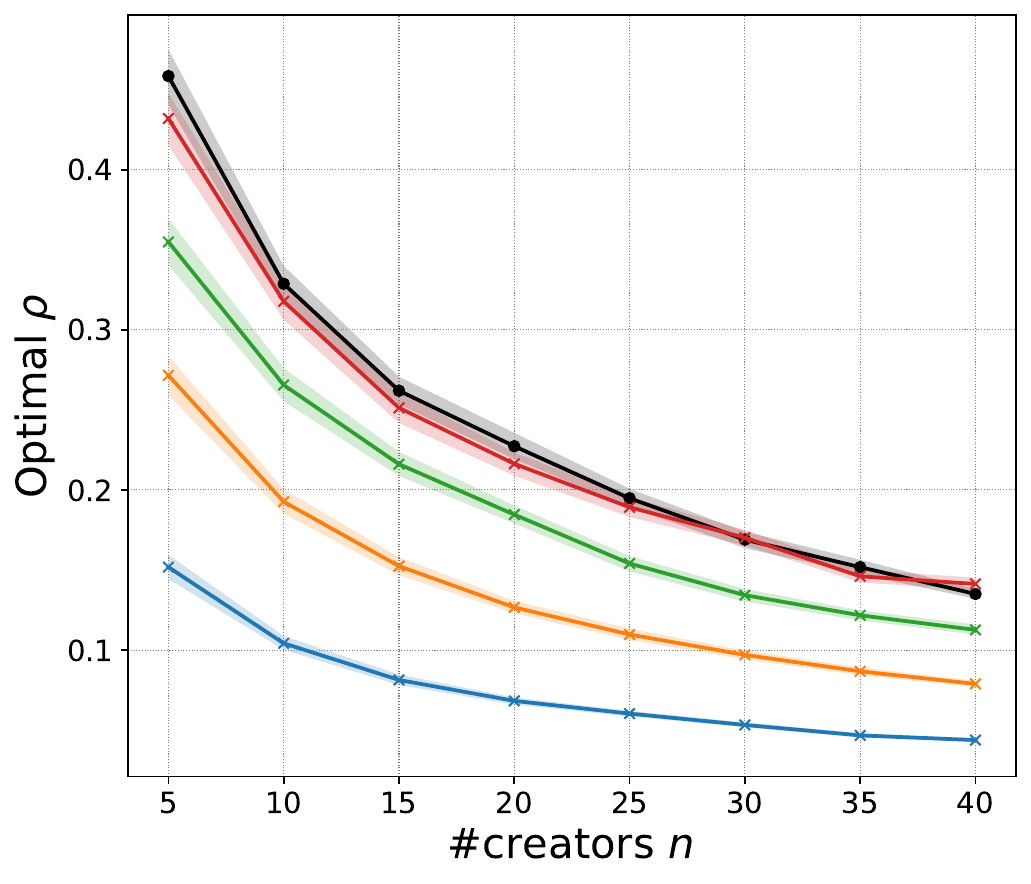}
            \caption{Optimal allocation parameter $\rho$}
            \label{fig:beta_subfig3}
        \end{subfigure}
    \end{minipage}

    \begin{minipage}{0.95\textwidth}
        \centering
        \textbf{$\alpha = 5$}
        \vspace{0.2cm}
        
        \centering
        \begin{subfigure}[t]{0.32\textwidth}
            \includegraphics[width=\textwidth]{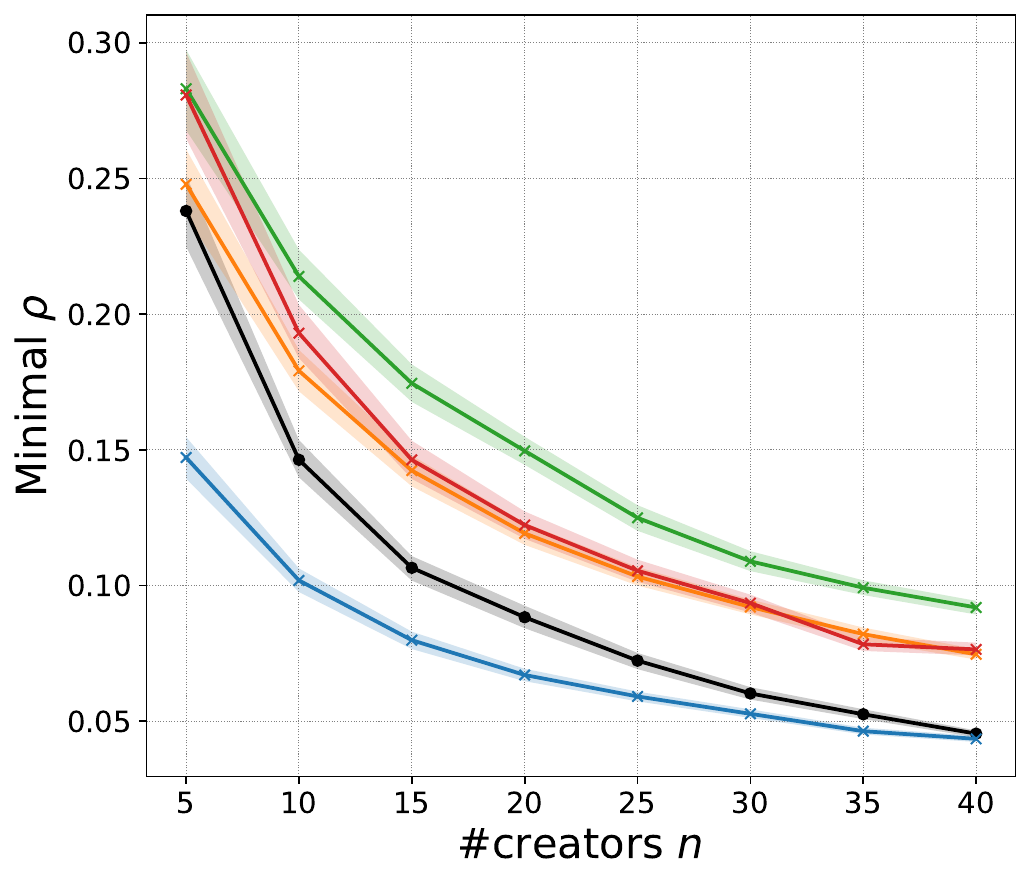}
            \caption{Min. $\rho$ inducing $\varepsilon$-FSE}
            \label{fig:beta_subfig4}
        \end{subfigure}
        \hfill
        \begin{subfigure}[t]{0.32\textwidth}
            \includegraphics[width=\textwidth]{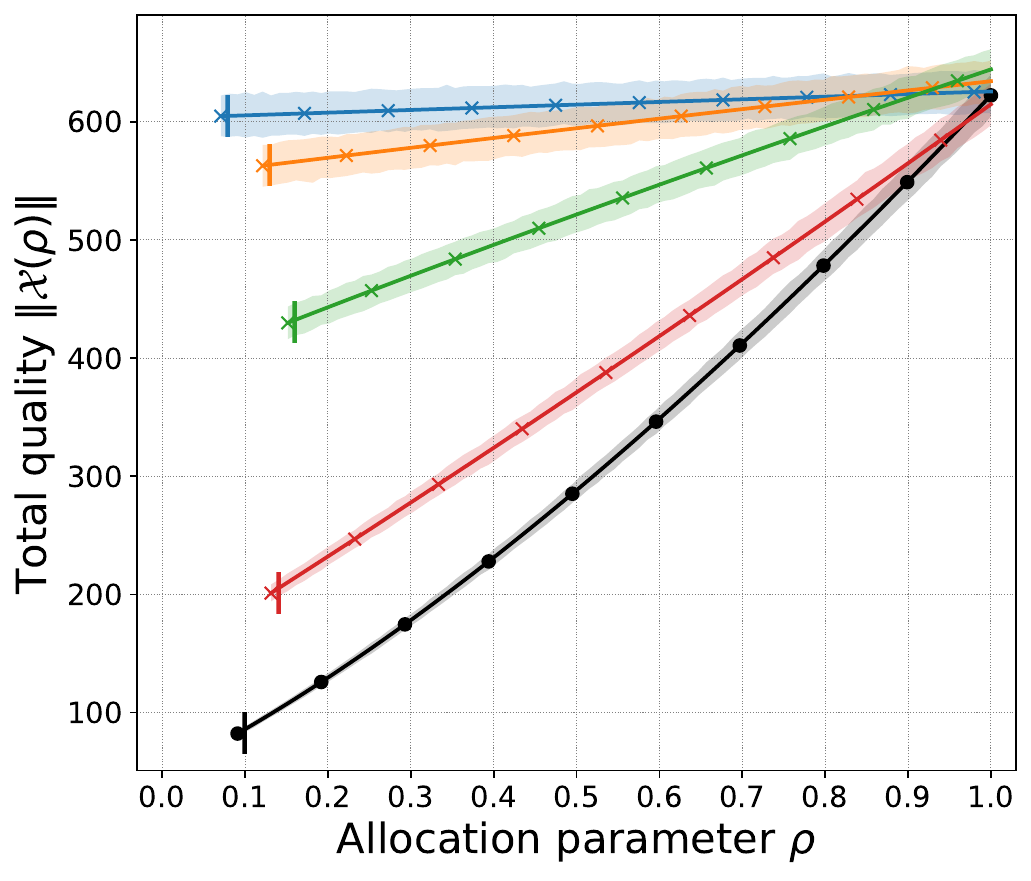}
            \caption{$\rho$ vs. $\norm{\xr{\rho}}$ in FSE}
            \label{fig:beta_subfig5}
        \end{subfigure}
        \hfill
        \begin{subfigure}[t]{0.32\textwidth}
            \includegraphics[width=\textwidth]{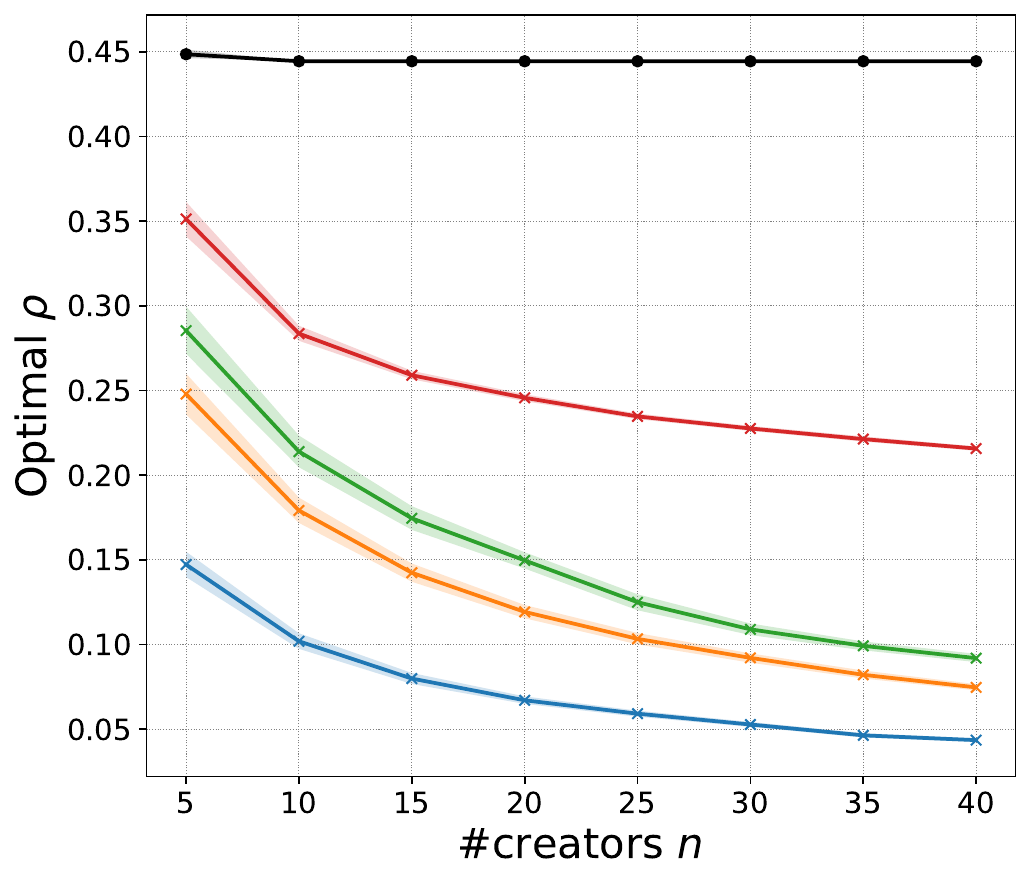}
            \caption{Optimal allocation parameter $\rho$}
            \label{fig:beta_subfig6}
        \end{subfigure}
    \end{minipage}

    \caption{Effect of the diminishing returns parameter $\beta$ on equilibrium outcomes for $\alpha = 0.5$ (top row) and $\alpha = 5$ (bottom row).}
    \label{fig:beta_combined}
\end{figure*}

\paragraph{Simulation results.}
We examine the effect of $\beta \in \{0.2, 0.4, 0.6, 0.8, 1.0\}$ on equilibrium outcomes, using two values of the GenAI efficiency parameter: $\alpha = 0.5$ (moderate GenAI) and $\alpha = 5$ (powerful GenAI). The results are presented in \Cref{fig:beta_combined}.

For moderate GenAI ($\alpha = 0.5$), several baseline trends persist. The minimal allocation parameter required to induce an FSE decreases with the number of creators, consistent with the intuition from \Cref{thm:fse-existence}. Content quality continues to increase with $\rho$. However, we observe that the minimal $\rho$ \emph{increases} with $\beta$, which is somewhat counterintuitive: In our baseline analysis, stronger GenAI (larger $\alpha$) decreased the minimal $\rho$, yet here a different notion of ``stronger'' GenAI (larger $\beta$, corresponding to better data utilization) has the opposite effect. The optimal allocation parameter in this regime closely tracks the minimal threshold, suggesting that the platform retains as much revenue as possible while maintaining an FSE.

For powerful GenAI ($\alpha = 5$), the relationship between $\beta$ and the minimal $\rho$ becomes non-monotonic, though the decrease with the number of creators persists. A clearer pattern emerges in the optimal allocation: Larger $\beta$ leads to higher optimal $\rho$, and the gap between optimal and minimal allocation parameters widens substantially for larger $\beta$. This indicates that with powerful GenAI exhibiting near-linear returns, the platform benefits from more generous revenue sharing to incentivize content production. The content-quality curves confirm that larger $\beta$ reduces total content at any given $\rho$, though quality remains increasing in the allocation parameter.

Overall, the core tradeoffs identified in our baseline model--balancing revenue retention against content production incentives--persist across the range of $\beta$ values. However, the interaction between $\alpha$ and $\beta$ introduces additional complexity, and some of the clean monotonicity properties from the linear case are lost when both parameters vary.

\section{Discussion and Future Work}

We explored strategic data sharing and revenue allocation in content platforms with GenAI. We modeled the interaction between creators and the platform, where each creator chooses both the quality of their content and the extent to which they share it with the platform. The platform, which commits to an allocation rule, allocates them a portion of its GenAI-driven revenue based on their content qualities and sharing decisions. We characterized conditions for inducing an FSE, and proposed an efficient approximation algorithm for solving the platform's revenue maximization problem~\eqref{eq:platform-optimization}. Our simulations examined sensitivity and equilibrium outcomes.

We see several promising directions for future research. First, while our analysis focused on proportional allocation rules, other functional forms may lead to better outcomes. For instance, allocation rules that exhibit a staircase structure, providing high revenue for certain quality levels but diminishing afterward, could potentially yield more nuanced and efficient outcomes. Second, we restricted our attention to FSEs, but other PNEs (non-FSEs) may exist. Indeed, \Cref{lem:s_choice_threshold} hints that asymmetric sharing, where only some creators choose to share content, can occur in equilibrium. Understanding which allocation rules can induce such equilibria and whether they lead to improved platform revenue remains an important open question. Third, extending our theoretical guarantees to more general GenAI quality functions with diminishing returns ($\beta < 1$ in \Cref{subsec:genai-quality-beta}) remains open. The loss of concavity invalidates our uniqueness results, and while simulations suggest equilibria persist with similar qualitative behavior, formal existence and characterization guarantees are lacking. Finally, our model considers a static interaction; incorporating dynamic elements such as creators entering or exiting the platform, evolving GenAI capabilities, or repeated interactions with learning could yield richer insights into long-term platform-creator relationships.

\section*{Acknowledgments}

This research was supported by the Israel Science Foundation (ISF; Grant No. 3079/24). We thank the anonymous reviewers for their valuable feedback.

\bibliographystyle{plainnat}
\bibliography{main}

\appendix

\newpage
\input{appendix/omitted_simulation_details}

\newpage
\input{appendix/section3_proofs}

\newpage
\input{appendix/section4_proofs}

\newpage
\input{appendix/extensions_proofs}

\newpage
\input{appendix/other_allocation_rules}

\newpage
\input{appendix/modeling_assumptions_and_limitations}

\end{document}

%% file: appendix/omitted_simulation_details.tex
\section{Omitted Simulation Details}\label{sec:sim_details}

This section provides technical details about the algorithms and computational resources used in our experimental simulations.\footnote{Our code is available at \url{https://github.com/GurKeinan/Code_for_Strategic_Content_Creation_with_GenAI_To_Share_or_Not_to_Share_Paper}.}

\subsection{ESE Computation}

\paragraph{Theoretical algorithm.} \Cref{alg:mamd_ese} provides a theoretically justified method for approximating the unique ESE for a given revenue-allocation parameter $\rho$. It employs multi-agent mirror descent with gradient-based updates to iteratively refine each creator's quality choice $x_i$ while maintaining non-negativity constraints. The strong monotonicity of the enforced sharing game (established in \Cref{lem:unique-ese}) guarantees convergence at a rate of $O(1/T)$~\cite{10.5555/3327345.3327469}, where $T$ is the number of iterations.

\begin{algorithm}[t]
\caption{Multi-Agent Mirror Descent for ESE Quality Profile (adapted from \citet{10.5555/3327345.3327469})}
\label{alg:mamd_ese}
\begin{algorithmic}[1]
\Require Number of creators $n$, cost functions $\{c_i\}_{i=1}^n$, model parameters $\mu, \alpha, \gamma, \rho$, number of iterations $T$, learning rate schedule $\{\eta_t\}_{t=0}^{T-1}$, initial quality profile $x^{(0)} \in \mathbb{R}^n_{\ge 0}$.
\Ensure Approximate ESE quality profile $x^{(T)}$.

\State Initialize quality profile $x \gets x^{(0)}$
\For{$t \gets 0$ \textbf{to} $T-1$}
    \State Calculate gradient vector $v^{(t)} \in \mathbb{R}^n$. For each creator $i$:
    \For{$i \gets 1$ \textbf{to} $n$}
        \State Let $X^{(t)} = \sum_{j=1}^n x_j^{(t)}$ and $X_{-i}^{(t)} = X^{(t)} - x_i^{(t)}$
        \State $v_i^{(t)} \gets \frac{\mu(1+\alpha\rho)}{1+\alpha} (X^{(t)})^{\gamma-2}\,[\gamma x_i^{(t)} + X_{-i}^{(t)}] - c_i'(x_i^{(t)})$
    \EndFor
    \State Update candidate profile: $x_{cand} \gets x^{(t)} + \eta_t v^{(t)}$
    \State Project onto feasible set: $x^{(t+1)} \gets \max(0, x_{cand})$
\EndFor
\State \Return $x^{(T)}$
\end{algorithmic}
\end{algorithm}

\paragraph{Practical implementation.} While \Cref{alg:mamd_ese} provides provable convergence guarantees, we adopted a more practical approach in our simulations. Since each creator's utility is strictly concave in their own quality $x_i$ (as established in \Cref{lem:unique-ese}), the ESE is uniquely characterized by the first-order conditions: For each creator $i \in [n]$, the equilibrium quality satisfies
\[
\frac{\partial U_i}{\partial x_i}\bigg|_{(\mathbf{x}^\star, \mathbf{1})} = 0 \quad \text{or} \quad x_i^\star = 0 \text{ with } \frac{\partial U_i}{\partial x_i}\bigg|_{(\mathbf{x}^\star, \mathbf{1})} \leq 0.
\]
We solve this system numerically using \texttt{scipy.optimize}'s root-finding routines. Under our cost assumptions (in particular, $\lim_{x \to 0^+} c_i'(x) = 0$), interior solutions are guaranteed, and the strict concavity ensures that any profile with near-zero gradients closely approximates the true ESE. This approach leverages highly optimized numerical solvers while maintaining theoretical soundness. We use a tolerance of $10^{-8}$ for gradient norms in all experiments.

\subsection{FSE Stability Check}

\Cref{alg:fse_stability_check_simulations_appendix} verifies whether a given profile is an FSE by checking if any creator can profitably deviate to not sharing ($s_i = 0$). This is necessary and sufficient due to \Cref{lem:s_choice_threshold}: any profitable deviation is dominated by one with $s_i \in \{0,1\}$, and deviations with $s_i = 1$ are already ruled out by the ESE property. For each creator, we solve a one-dimensional optimization problem to find their best response when withholding content. The profile is deemed stable only if no creator benefits from such deviation beyond the tolerance threshold $\varepsilon_{tol}$.

\begin{algorithm}[t]
\caption{FSE Stability Check via $s_i=0$ Deviation}
\label{alg:fse_stability_check_simulations_appendix}
\begin{algorithmic}[1]
\Require Candidate ESE quality profile $x^\star \in \mathbb{R}^n_{>0}$, cost functions $\{c_i(x_i) = a_i x_i^{\theta_i}\}_{i=1}^n$, model parameters $\mu, \alpha, \gamma, \rho$, stability tolerance $\varepsilon_{tol} > 0$.
\Ensure Boolean indicating if $(x^\star, \mathbf{1})$ is stable (\textbf{True}) or not (\textbf{False}).

\For{$i \gets 1$ \textbf{to} $n$}
    \State Calculate utility at candidate profile: $U_i^\star \gets U_i(x_i^\star, 1; x_{-i}^\star, \mathbf{1}_{-i})$
    \State Calculate rational quality upper bound for deviation: $x_i^{max} \gets (\mu / a_i)^{1/(\theta_i-\gamma)}$
    \State $U_{dev}(s_i=0) \gets \max_{x_i' \in [0, x_i^{max}]} \{ U_i(x_i', 0; x_{-i}^\star, \mathbf{1}_{-i}) \}$ \Comment{Requires 1D solver}
    \State Check stability condition:
    \If{$U_{dev}(s_i=0) - U_i^\star > \varepsilon_{tol}$}
        \State \Return \textbf{False} \Comment{Player $i$ has a profitable deviation}
    \EndIf
\EndFor
\State \Return \textbf{True} \Comment{No player has a profitable deviation to $s_i=0$}
\end{algorithmic}
\end{algorithm}

\subsection{Simulation Pipeline and Computational Resources}

For each parameter combination, we initialized 150 independent game instances with randomly sampled cost coefficients. For each instance, we evaluated 100 equally-spaced values of $\rho \in [0,1]$. The computational workflow proceeds as follows: For each $\rho$ value, we compute the unique ESE by solving the first-order conditions using \texttt{scipy.optimize}, leveraging the strict concavity established in our theoretical analysis. Subsequently, \Cref{alg:fse_stability_check_simulations_appendix} verifies whether each computed ESE also constitutes an FSE by checking for profitable deviations to non-sharing strategies. Results are averaged across instances with 95\% confidence intervals computed via bootstrapping.

The complete simulation suite was executed on a standard MacBook Pro with 18 GB of RAM and 8 CPU cores. We used parallel processing across 6 cores to accelerate computation. The full set of experiments required approximately 6 hours of computation time for ESE computation and FSE verification across all parameter combinations. Statistical analysis, bootstrapping for confidence intervals, and visualization required less than 1 minute. The implementation uses standard Python scientific computing libraries (NumPy, SciPy) without GPU acceleration.

%% file: appendix/section3_proofs.tex
\section{Omitted Proofs from \Cref{sec:fse}}\label{sec:sec3_proofs}

This section is divided into two parts. In the first part, we state and prove all the auxiliary lemmas necessary for the proofs of the main results presented in \Cref{sec:fse}. The second part provides the complete proofs of these main results.

\begin{remark}
    For the rest of the paper, we denote $x_i^{\max}$ as the unique positive solution to $c_i(x) = \mu x^\gamma$. This value is well-defined because the function $g(x) = c_i(x) - \mu x^\gamma$ is strictly convex (as $c_i$ is strictly convex and $x^\gamma$ is concave) and satisfies $g(0^+)<0$ and $\lim_{x\to\infty} g(x) = \infty$. The boundary conditions imply that $g(x)$ must eventually become positive, ensuring it crosses zero at least once. Meanwhile, strict convexity ensures that the slope is strictly increasing, meaning $g(x)$ can cross the x-axis exactly once.
\end{remark}

\subsection{Auxiliary Lemmas for \Cref{sec:fse}'s Main Results}

\begin{lemma}
\label{lem:utility_full_sharing}
Consider a profile of the form $(\x, \mathbf{1})$. Then, the utility of creator $i \in [n]$ under the proportional allocation rule $f_\rho$ is given by:
\[
U_i(\x, \mathbf{1}; f_\rho) = \frac{\mu (1 + \alpha \rho)}{1 + \alpha} \cdot (\xtot)^{\gamma - 1} \cdot x_i - c_i(x_i).
\]
\end{lemma}

\begin{proof}[\proofof{lem:utility_full_sharing}]
    Under full sharing, the total shared content equals the total human content, $\stot = \xtot$. Consequently, Creator $i$ receives a direct traffic share of $\frac{x_i}{(1+\alpha)\xtot}$, and their share of AI-driven revenue under the rule $f_\rho$ is $\rho \cdot \frac{x_i}{\xtot}$.

    Combining both revenue sources, the utility of creator $i$ is:
    \begin{align*}
        U_i(\x, \mathbf{1}; f_\rho)
        &= \mu (\xtot)^\gamma \cdot \frac{x_i}{(1+\alpha)\xtot}
        + \mu (\xtot)^\gamma \cdot \frac{\alpha \xtot}{(1+\alpha)\xtot} \cdot \rho \cdot \frac{x_i}{\xtot}
        - c_i(x_i) \\
        &= \mu (\xtot)^{\gamma - 1} \cdot \left( \frac{1 + \alpha \rho}{1 + \alpha} \right) \cdot x_i - c_i(x_i) \\
        &= \frac{\mu (1 + \alpha \rho)}{1 + \alpha} \cdot (\xtot)^{\gamma - 1} \cdot x_i - c_i(x_i),
    \end{align*}
    as claimed.
\end{proof}

\begin{lemma}
    \label{lem:rational_quality_bound}
    Consider a profile $(\x, \mathbf{1})$ and some arbitrary creator $i \in [n]$. Then, any deviation $(x_i', s_i')$ by creator $i$ where $x_i' > x_i^{\max}$ yields negative utility.
\end{lemma}

\begin{proof}[\proofof{lem:rational_quality_bound}]
    Consider a profile $(\x, \mathbf{1})$ with arbitrary $\x \in [0,\infty)^n$, and a unilateral deviation $(x_i', s_i')$ by creator $i$. For the deviation to yield non-negative utility, the revenue must exceed the cost:
    \[
        R_i(x_i', s_i') \coloneqq \Tx \frac{x_i'}{\xtot' + \qai'} + \Tx \frac{\qai'}{\xtot' + \qai'} \cdot f_i(\x', \mathbf{s}') \geq c_i(x_i').
    \]

    By \Cref{lem:s_choice_threshold}, the optimal sharing level for any quality $x_i'$ must be either $s_i'=0$ or $s_i'=1$, with other values being weakly dominated. We analyze both cases and show that in each of them, the revenue is bounded above by $\mu (x_i')^{\gamma}$. Thus, any deviation $(x_i', s_i')$ with $x_i' > x_i^{\max}$ leads to negative utility, as for those values the cost $c_i(x_i')$ exceeds the revenue.

    \paragraph{Case 1: Deviation with full sharing ($s_i'=1$)}
    With $s_i'=1$ and $s_j=1$ for $j \neq i$, we obtain via \Cref{lem:utility_full_sharing}:
    \[
        R_i(x_i', 1) = \frac{\mu (1 + \alpha \rho)}{1+\alpha} (x_i' + \xtotmi)^{\gamma-1} x_i'.
    \]

    Since $\rho \in [0,1]$, we have $\frac{1 + \alpha \rho}{1+\alpha} \leq 1$. Additionally, as $\gamma \in [0,1]$ and $x_i' + \xtotmi \geq x_i'$, it follows that $(x_i' + \xtotmi)^{\gamma-1} \leq (x_i')^{\gamma-1}$. Consequently:
    \[
        R_i(x_i', 1) \leq \mu (x_i')^{\gamma-1} x_i' = \mu (x_i')^{\gamma}.
    \]

    \paragraph{Case 2: Deviation with no sharing ($s_i'=0$)}
    When $s_i'=0$, creator $i$ receives revenue solely from direct traffic:
    \[
        R_i(x_i', 0) = \mu (x_i' + \xtotmi)^{\gamma} \frac{x_i'}{x_i' + (1+\alpha)\xtotmi}.
    \]

    Letting $Y = \xtotmi$, we define $h(Y) = (x_i' + Y)^{\gamma} \frac{x_i'}{x_i' + (1+\alpha)Y}$ and prove that $h(Y) \leq (x_i')^{\gamma}$ for all $Y \geq 0$. Taking the derivative:
    \begin{align*}
        h'(Y) &= \gamma (x_i' + Y)^{\gamma-1} \frac{x_i'}{x_i' + (1+\alpha)Y} - (x_i' + Y)^{\gamma} \frac{x_i'(1+\alpha)}{(x_i' + (1+\alpha)Y)^2} \\
        &= \frac{x_i'(x_i' + Y)^{\gamma-1}}{(x_i' + (1+\alpha)Y)^2}\left[x_i'(\gamma - (1+\alpha)) + Y(1+\alpha)(\gamma-1)\right].
    \end{align*}

    Since $\gamma \in [0,1]$, $\alpha > 0$, and $Y \geq 0$, both terms in the brackets are non-positive, yielding $h'(Y) \leq 0$ for all $Y \geq 0$. Thus, $h$ is non-increasing with maximum at $Y=0$, giving $h(Y) \leq h(0) = (x_i')^{\gamma}$. Therefore:
    \[
        R_i(x_i', 0) = \mu h(\xtotmi) \leq \mu (x_i')^{\gamma}.
    \]

    Combining both cases, we establish that for any deviation, $R_i(x_i', s_i') \leq \mu (x_i')^{\gamma}$. The condition $R_i(x_i', s_i') \geq c_i(x_i')$ thus necessitates:
    \[
        c_i(x_i') \leq \mu (x_i')^{\gamma}.
    \]

    Consequently, any deviation $(x_i', s_i')$ with $x_i' > x_i^{\max}$ yields negative utility for creator $i$. This completes the proof of \Cref{lem:rational_quality_bound}.
\end{proof}

\begin{lemma}\label{lem:quadratic_form_helper}
    Let $\x = (x_1, \dots, x_n)$ be a non-negative vector in $\mathbb{R}^n$ and $\mathbf{v} = (v_1, \dots, v_n)$ be an arbitrary real vector. For any $\gamma \in [0, 1]$, the following quadratic form is non-negative:
    \[
    \mathcal{Q}_{\gamma}(\x, \mathbf{v}) \coloneqq \sum_{i=1}^n v_i^2 \left(2\gamma x_i + 4 \sum_{\substack{j=1 \\ j \neq i}}^n x_j \right) + \sum_{\substack{i,j=1 \\ i \neq j}}^n v_i v_j \left( \gamma (x_i + x_j) + 2 \sum_{\substack{k=1 \\ k \neq i,j}}^n x_k \right) \geq 0.
    \]
\end{lemma}

\begin{proof}[\proofof{lem:quadratic_form_helper}]
    Our strategy is to decompose $\mathcal{Q}_{\gamma}(\x, \mathbf{v})$ as a sum $\sum_{r=1}^n x_r C_r$ and demonstrate that each coefficient $C_r \geq 0$. Since $x_r \geq 0$ for all $r$, this will establish that $\mathcal{Q}_{\gamma}(\x, \mathbf{v}) \geq 0$.

    \paragraph{Step 1: Collecting terms associated with each $x_r$}
    Fix an arbitrary index $r \in \{1,\ldots,n\}$ and identify all terms in $\mathcal{Q}_{\gamma}(\x, \mathbf{v})$ containing $x_r$.

    From the first summation $\sum_{i=1}^n v_i^2 (2\gamma x_i + 4 \sum_{j \neq i} x_j)$, we get:
    \begin{itemize}
        \item When $i = r$: The term $2\gamma v_r^2 x_r$.
        \item When $i \neq r$ and $j = r$: The term $4v_i^2 x_r$ for each $i \neq r$.
    \end{itemize}
    The total contribution from the first summation is therefore $2\gamma v_r^2 x_r + \sum_{i \neq r} 4v_i^2 x_r = x_r(2\gamma v_r^2 + 4\sum_{i \neq r} v_i^2)$.

    From the second summation $\sum_{i \neq j} v_i v_j (\gamma(x_i + x_j) + 2\sum_{k \neq i,j} x_k)$, we get:
    \begin{itemize}
        \item When $i = r, j \neq r$: The term $\gamma v_r v_j x_r$ for each $j \neq r$.
        \item When $i \neq r, j = r$: The term $\gamma v_i v_r x_r$ for each $i \neq r$.
        \item When $i,j \neq r$ and $k = r$: The term $2v_i v_j x_r$ for each pair $(i,j)$ where $i \neq j$ and both $i,j \neq r$.
    \end{itemize}

    Let $A_r = \sum_{i \neq r} v_i$ denote the sum of all $v_i$ except $v_r$. The contribution from the first two bullet points is $\gamma v_r \sum_{j \neq r} v_j + \gamma v_r \sum_{i \neq r} v_i = 2\gamma v_r A_r$.

    \paragraph{Step 2: Simplifying the cross-terms}
    For the third bullet point, we need to express $\sum_{i,j \neq r, i \neq j} v_i v_j$ in terms of $A_r$ and the sum of squares. Note that:
    \[
        A_r^2 = \left(\sum_{i \neq r} v_i\right)^2 = \sum_{i \neq r} v_i^2 + \sum_{i,j \neq r, i \neq j} v_i v_j.
    \]
    Rearranging, we get $\sum_{i,j \neq r, i \neq j} v_i v_j = A_r^2 - \sum_{i \neq r} v_i^2$. Thus, the contribution from the third bullet point is $2x_r(A_r^2 - \sum_{i \neq r} v_i^2)$.

    \paragraph{Step 3: Computing the total coefficient}
    Summing all contributions, the coefficient of $x_r$ is:
    \begin{align*}
        C_r &= 2\gamma v_r^2 + 4\sum_{i \neq r} v_i^2 + 2\gamma v_r A_r + 2A_r^2 - 2\sum_{i \neq r} v_i^2 \\
        &= 2\gamma v_r^2 + 2\sum_{i \neq r} v_i^2 + 2\gamma v_r A_r + 2A_r^2 \\
        &= 2\left(\gamma v_r^2 + \sum_{i \neq r} v_i^2 + \gamma v_r A_r + A_r^2\right).
    \end{align*}

    \paragraph{Step 4: Proving non-negativity through completing the square}
    To prove $C_r \geq 0$, we complete the square for terms involving $A_r$ and $v_r$. Observe that:
    \[
        A_r^2 + \gamma v_r A_r + \frac{\gamma^2}{4}v_r^2 = \left(A_r + \frac{\gamma}{2}v_r\right)^2.
    \]

    We can rewrite $\gamma v_r^2$ as $\frac{\gamma^2}{4}v_r^2 + (\gamma - \frac{\gamma^2}{4})v_r^2$. Substituting into our expression for $C_r$:
    \begin{align*}
        C_r &= 2\left(\frac{\gamma^2}{4}v_r^2 + \left(\gamma - \frac{\gamma^2}{4}\right)v_r^2 + \sum_{i \neq r} v_i^2 + \gamma v_r A_r + A_r^2\right) \\
        &= 2\left[\left(A_r + \frac{\gamma}{2}v_r\right)^2 + \left(\gamma - \frac{\gamma^2}{4}\right)v_r^2 + \sum_{i \neq r} v_i^2\right].
    \end{align*}

    Notice that $\gamma - \frac{\gamma^2}{4} \geq 0$ for all $\gamma \in [0,1]$. Therefore, as each term inside the brackets is non-negative, we have $C_r \geq 0$ for each $r$. Remembering that we started with $x_r \geq 0$, we conclude that:
    \[
        \mathcal{Q}_{\gamma}(\x, \mathbf{v}) = \sum_{r=1}^n x_r C_r \geq 0.
    \]
    This completes the proof of \Cref{lem:quadratic_form_helper}.
\end{proof}

\begin{lemma}\label{lem:dominance-positive-quality}
    For any creator $i$, any rational strategy must involve choosing $x_i > 0$.
\end{lemma}

\begin{proof}[\proofof{lem:dominance-positive-quality}]
Fix creator $i$ and strategies $(\x_{-i},\mathbf{s}_{-i})$. If $x_i=0$, then creator $i$ gets no direct-traffic share and no AI-revenue share,
and pays cost $c_i(0)=0$, hence $U_i(0,s_i)=0$ for every $s_i\in[0,1]$.

Next, let $Y=\norm{\x_{-i}}$ and $Z = \alpha \x_{-i}^\top \mathbf{s}_{-i}$. Consider the deviation $(x_i,s_i) = (\varepsilon,0)$ with $\varepsilon>0$. Then
\[
U_i(\varepsilon,0) = \mu(\varepsilon + Y)^\gamma \cdot \frac{\varepsilon}{\varepsilon + Y + Z} - c_i(\varepsilon) = \varepsilon \cdot g(\varepsilon) - c_i(\varepsilon),
\]
where $g(\varepsilon) \coloneqq \frac{\mu(\varepsilon+Y)^\gamma}{(\varepsilon+Y+Z)} > 0$ for $\varepsilon > 0$.

By the mean value theorem, for every $\varepsilon>0$ there exists $\xi\in(0,\varepsilon)$ such that
$c_i(\varepsilon)-c_i(0) = c_i'(\xi)\varepsilon$. Since $\lim_{x\to0^+}c_i'(x) = 0$, we get $\frac{c_i(\varepsilon)}{\varepsilon} \to 0 \quad\text{as }\varepsilon\to 0^+$, i.e., $c_i(\varepsilon)=o(\varepsilon)$.

If $Y>0$, then $g$ is continuous at $0$, and, for sufficiently small $\varepsilon$, $\abs{g(\varepsilon) - g(0)} \leq g(0)/2 \implies g(\varepsilon) \geq g(0)/2$. Using $c_i(\varepsilon) = o(\varepsilon)$, for all sufficiently small $\varepsilon$ also $c_i(\varepsilon)\le (g(0)/4)\varepsilon$, so
\[
U_i(\varepsilon,0)\;\ge\;\varepsilon\cdot \frac{g(0)}{2} - \frac{g(0)}{4}\varepsilon
\;=\;\frac{g(0)}{4}\varepsilon \;>\;0.
\]

If $Y=0$, then $U_i(\varepsilon,0)=\mu\varepsilon^\gamma-c_i(\varepsilon)$.
Since $\gamma\in[0,1]$, we have $\varepsilon = o(\varepsilon^\gamma)$ when $\gamma<1$
(and when $\gamma=1$ we already know $c_i(\varepsilon)=o(\varepsilon)$), so in all cases
$c_i(\varepsilon)=o(\varepsilon^\gamma)$. Therefore for all sufficiently small $\varepsilon$,
$c_i(\varepsilon)\le (\mu/2)\varepsilon^\gamma$, giving $U_i(\varepsilon,0)\ge (\mu/2)\varepsilon^\gamma>0$.

Thus for every $(\x_{-i},\mathbf{s}_{-i})$ and every $s_i$, $(x_i=0,s_i)$ is strictly dominated by
$(x_i=\varepsilon,s_i=0)$ for small enough $\varepsilon>0$. Hence, no rational strategy can play $x_i=0$.
\end{proof}

\begin{lemma}\label{lem:bound_on_total_quality_in_ese_power_costs}
    For any power cost instance where $a_{\min}, a_{\max} = O(1)$ and every $\rho \in [0,1]$, it holds that $\norm{\xr{\rho}}, \norm{\mathcal{X}_{-i}(\rho)}  = \Theta(n^{\frac{\theta-1}{\theta-\gamma}})$.
\end{lemma}

\begin{proof}[\proofof{lem:bound_on_total_quality_in_ese_power_costs}]
    Consider any creator $i$. Its FOC in the ESE is:
    \[
        \frac{\mu (1 + \alpha \rho)}{1 + \alpha} \norm{\xr{\rho}}^{\gamma - 2} \left[ (\gamma - 1) \mathcal{X}_i(\rho) + \norm{\xr{\rho}} \right] = a_i \theta (\mathcal{X}_i(\rho))^{\theta - 1}.
    \]
    Denoting $q_i = \frac{\mathcal{X}_i(\rho)}{\norm{\xr{\rho}}}$, we can rearrange the FOC to get:
    \begin{equation}\label{eq:proof-lem-bound-on-total-quality-in-ese-power-costs-foc-with-q}
        \frac{\mu (1 + \alpha \rho)}{1 + \alpha} \norm{\xr{\rho}}^{\gamma - 1} \left[ 1 - (1 - \gamma) q_i \right] = a_i \theta (\mathcal{X}_i(\rho))^{\theta - 1}.
    \end{equation}

    \paragraph{Step 1: Bound $q_i$ uniformly} Let $j^\star = \arg\min_{j \in [n]} q_j$. Since $q_{j^\star} \leq \frac{1}{n} \leq \frac{1}{2}$, we have:
    \[
        1 - (1-\gamma)q_{j^\star} \geq 1 - \frac{1-\gamma}{2} = \frac{1+\gamma}{2}.
    \]
    Now, we divide the FOC of any creator $i$, given by \Cref{eq:proof-lem-bound-on-total-quality-in-ese-power-costs-foc-with-q}, by the FOC of creator $j^\star$:
    \[
        \frac{1 - (1-\gamma) q_i}{1 - (1-\gamma) q_{j^\star}} = \frac{a_i \theta (\mathcal{X}_i(\rho))^{\theta - 1}}{a_{j^\star} \theta (\mathcal{X}_{j^\star}(\rho))^{\theta - 1}}.
    \]
    Noticing that the LHS is at most $\frac{1}{\frac{1+\gamma}{2}} = \frac{2}{1+\gamma}$, we have:
    \begin{equation}\label{eq:proof-lem-bound-on-total-quality-in-ese-power-costs-xi-ratio-bound}
        \mathcal{X}_i(\rho) \leq \left( \frac{2 a_{j^\star}}{(1+\gamma) a_{i}} \right)^{\frac{1}{\theta - 1}} \mathcal{X}_{j^\star}(\rho) \leq \left( \frac{2 a_{\max}}{(1+\gamma) a_{\min}} \right)^{\frac{1}{\theta - 1}} \mathcal{X}_{j^\star}(\rho) \eqqcolon C \mathcal{X}_{j^\star}(\rho).
    \end{equation}
    Notice that this is true for every creator $j$, not just $j^\star$. Substituting \Cref{eq:proof-lem-bound-on-total-quality-in-ese-power-costs-xi-ratio-bound} back into the definition of $q_i$, we get:
    \begin{equation}\label{eq:proof-lem-bound-on-total-quality-in-ese-power-costs-qi-bound}
        q_i = \frac{\mathcal{X}_i(\rho)}{\norm{\xr{\rho}}} \leq \frac{\max_{j \in [n]} \mathcal{X}_j(\rho)}{\max_{j \in [n]} \mathcal{X}_j(\rho) + (n-1)\min_{j \in [n]} \mathcal{X}_j(\rho)} \leq \frac{C}{C + n - 1} \leq \frac{C}{C+1}.
    \end{equation}
    Furthermore, we get that for every creator $i$:
    \begin{equation}\label{eq:proof-lem-bound-on-total-quality-in-ese-power-costs-1-minus-1-gamma-qi-bound}
        1 \geq 1 - (1-\gamma) q_i \geq 1 - (1-\gamma) \frac{C}{C+1} = \frac{1 + \gamma C}{C + 1} \coloneqq D > 0.
    \end{equation}

    \paragraph{Step 2: Bound $\norm{\xr{\rho}}$}
    Substituting \Cref{eq:proof-lem-bound-on-total-quality-in-ese-power-costs-1-minus-1-gamma-qi-bound} back into the FOC of any creator $i$ given by \Cref{eq:proof-lem-bound-on-total-quality-in-ese-power-costs-foc-with-q}, we have:
    \begin{equation}\label{eq:proof-lem-bound-on-total-quality-in-ese-power-costs-foc-theta-of-1}
        a_i \theta (\mathcal{X}_i(\rho))^{\theta - 1} = \Theta\left(\norm{\xr{\rho}}^{\gamma - 1}\right) \implies \mathcal{X}_i(\rho) = \Theta\left(\norm{\xr{\rho}}^{\frac{\gamma - 1}{\theta - 1}}\right).
    \end{equation}
    Summing \Cref{eq:proof-lem-bound-on-total-quality-in-ese-power-costs-foc-theta-of-1} over all creators $i \in [n]$, we get:
    \[
        \sum_{i=1}^{n} \mathcal{X}_i(\rho) = \norm{\xr{\rho}} = \Theta\left(n \cdot \norm{\xr{\rho}}^{\frac{\gamma - 1}{\theta - 1}}\right) \implies \norm{\xr{\rho}} = \Theta\left(n^{\frac{\theta - 1}{\theta - \gamma}}\right).
    \]
    Additionally, since $q_i \leq \frac{C}{C+1} < 1$ for every creator $i$ by \Cref{eq:proof-lem-bound-on-total-quality-in-ese-power-costs-qi-bound}, it means that $\norm{\mathcal{X}_{-i}(\rho)} = \Theta\left(\norm{\xr{\rho}}\right) = \Theta\left(n^{\frac{\theta - 1}{\theta - \gamma}}\right)$ as well. This completes the proof of \Cref{lem:bound_on_total_quality_in_ese_power_costs}.
\end{proof}

\subsection{Proofs of \Cref{sec:fse}'s Main Results}

\LemSChoiceThreshold*

\begin{proof}[\proofof{lem:s_choice_threshold}]
    We analyze how creator $i$'s utility varies with respect to their sharing level $s_i$ to determine the optimal choice. For notational simplicity, we denote the platform traffic by $T = \mu (x_i + \xtotmi)^\gamma$, which remains constant with respect to $s_i$. For the rest of the proof, as $(\x, \mathbf{s}_{-i})$ is fixed, we denote $U_i(s_i) \coloneqq U_i(\x, (s_i, \mathbf{s}_{-i}); f_\rho)$.

    Consider the special case where $x_i s_i + \x^{\top}_{-i} \mathbf{s}_{-i} = 0$ (implying $s_i = \x^{\top}_{-i} \mathbf{s}_{-i} = 0$). In this case, the utility of creator $i$ consists of its direct traffic minus its production cost:
    \begin{equation}
        \label{eq:utility_universal_no_sharing_wrt_si}
        U_i(0) = T \cdot \frac{x_i}{x_i + \xtotmi} - c_i(x_i).
    \end{equation}

    Now, assume that $x_i s_i + \x^{\top}_{-i} \mathbf{s}_{-i} > 0$. Denote $D(s_i) = (x_i + \xtotmi) + \alpha(x_i s_i + \x^{\top}_{-i} \mathbf{s}_{-i})$. The utility of creator $i$ can be expressed as:

    \begin{align*}
    U_i(s_i) &= T \cdot \frac{x_i}{D(s_i)} + T \cdot \frac{\alpha(x_i s_i + \x^{\top}_{-i} \mathbf{s}_{-i})}{D(s_i)} \cdot \rho \frac{x_i s_i}{x_i s_i + \x^{\top}_{-i} \mathbf{s}_{-i}} - c_i(x_i) \\
    &= T \cdot \frac{x_i}{D(s_i)} + T \cdot \frac{\alpha \rho x_i s_i}{D(s_i)} - c_i(x_i) \\
    &= T \left[ \frac{x_i(1 + \alpha\rho s_i)}{D(s_i)} \right] - c_i(x_i).
    \end{align*}

    Notice that this simplified expression is consistent with \Cref{eq:utility_universal_no_sharing_wrt_si} when $s_i = 0$. Next, we compute the derivative of $U_i(s_i)$ with respect to $s_i$:
    \begin{align*}
        \frac{\partial U_i(s_i)}{\partial s_i} &= T \frac{(x_i\alpha\rho) \cdot D(s_i) - x_i(1 + \alpha\rho s_i) \cdot (\alpha x_i)}{[D(s_i)]^2} \\
        &= T \frac{\alpha x_i}{[D(s_i)]^2} \left[ \rho D(s_i) - x_i(1 + \alpha\rho s_i) \right] \\
        &= T \frac{\alpha x_i}{[D(s_i)]^2} \left[ \rho \left((x_i + \xtotmi) + \alpha(x_i s_i + \x^{\top}_{-i} \mathbf{s}_{-i})\right) - x_i(1 + \alpha\rho s_i) \right] \\
        &= T \frac{\alpha x_i}{[D(s_i)]^2} \left[ \rho(x_i + \xtotmi + \alpha \x^{\top}_{-i} \mathbf{s}_{-i}) - x_i \right]. \\
    \end{align*}


    Since all factors outside the brackets ($T$, $\alpha$, $x_i$, and $[D(s_i)]^2$) are strictly positive, the sign of this derivative depends solely on the term in brackets. Specifically, $\frac{\partial U_i(s_i)}{\partial s_i} > 0$ if and only if:
    \[
    \rho(x_i + \xtotmi + \alpha \x^{\top}_{-i} \mathbf{s}_{-i}) - x_i > 0 \implies
    \rho > \frac{x_i}{x_i + \xtotmi + \alpha \x^{\top}_{-i} \mathbf{s}_{-i}} = \tau_i.
    \]

    This analysis yields the following conclusions:
    \begin{itemize}
    \item When $\rho > \tau_i$, the derivative is strictly positive, making $U_i(s_i)$ strictly increasing in $s_i$. Consequently, the maximum utility is attained at $s_i = 1$.
    \item When $\rho < \tau_i$, the derivative is strictly negative, making $U_i(s_i)$ strictly decreasing in $s_i$. The maximum utility is therefore attained at $s_i = 0$.
    \item When $\rho = \tau_i$, the derivative vanishes, making $U_i(s_i)$ constant with respect to $s_i$. In this case, any value $s_i \in [0,1]$ maximizes the utility.
    \end{itemize}

    This completes the proof of \Cref{lem:s_choice_threshold}
\end{proof}

\LemUniqueESE*

\begin{proof}[\proofof{lem:unique-ese}]
    For enforced sharing ($s_i=1$ for all $i$), our model bears structural similarities to \citet{10.5555/3692070.3694417}. An alternative route to proving this proposition would be to use a reduction to their model, which holds only for enforced sharing. Nevertheless, we provide a direct, self-contained proof tailored to our setting. We establish existence and uniqueness by demonstrating that the enforced sharing game is monotone and the utility functions are strictly concave in each player's action.

    \paragraph{Step 1: Utility structure and gradient computation}
    \Cref{lem:utility_full_sharing} asserts that, under enforced sharing, creator $i$'s utility is:
    \[
        U_i(\x, \mathbf{1}; f_\rho) = \frac{\mu (1 + \alpha \rho)}{1 + \alpha} \cdot (\xtot)^{\gamma - 1} \cdot x_i - c_i(x_i).
    \]
    Let us denote the first term by $R_i(\x)$, and $K = \frac{\mu (1 + \alpha \rho)}{1 + \alpha} > 0$. Taking the partial derivative of $R_i(\x)$ w.r.t. $x_i$:
    \begin{align*}
        \frac{\partial R_i}{\partial x_i} &= K \cdot \frac{\partial}{\partial x_i} \left[ (\xtot)^{\gamma - 1} \cdot x_i \right] \\
        &= K \left[ (\gamma - 1) \xtot^{\gamma - 2} x_i + \xtot^{\gamma - 1} \right]  \\
        &= K \xtot^{\gamma - 2} \left[ (\gamma - 1) x_i + \xtot \right]  \\
        &= K \xtot^{\gamma - 2} \left[ \gamma x_i + \xtotmi \right].
    \end{align*}

    We denote this gradient by $A_i$, and note that $v_i \coloneqq \frac{\partial U_i}{\partial x_i} = A_i - c_i'(x_i)$.

    \paragraph{Step 2: Strategy space consideration}
    \Cref{lem:rational_quality_bound} suggests that we can restrict attention to the effective strategy space $[0, x_i^{\max}]$ for each creator $i \in [n]$, where $x_i^{\max}$ is the unique positive root of $\mu x^\gamma = c_i(x)$. This strategy space is compact and convex. Additionally, \Cref{lem:dominance-positive-quality} ensures that all rational strategies involve strictly positive quality choices, allowing us to focus on the interior of this space.

    \paragraph{Step 3: Strict concavity of utility functions}
    We establish that $U_i$ is strictly concave with respect to $x_i$ for any fixed $\x_{-i}$ by demonstrating that $R_i$ is concave in $x_i$, and then using the strict concavity of $-c_i$ and the fact that the sum of a concave and a strictly concave function is strictly concave. Taking the derivative of $A_i$ with respect to $x_i$:
    \begin{align*}
        \frac{\partial A_i}{\partial x_i} &= K \left[ (\gamma - 2) \xtot^{\gamma - 3} \left( \gamma x_i + \xtotmi \right) + \xtot^{\gamma - 2} \gamma \right] \\
        &= K \xtot^{\gamma - 3} \left[ (\gamma - 2) (\gamma x_i + \xtotmi) + \gamma \xtot \right] \\
        &= K \xtot^{\gamma - 3} \left[ (\gamma^2 - 2\gamma) x_i + (\gamma - 2) \xtotmi + \gamma (x_i + \xtotmi) \right] \\
        &= K \xtot^{\gamma - 3} \left[ (\gamma^2 - \gamma) x_i + (2\gamma - 2) \xtotmi \right] \\
        &= K (\gamma - 1) \xtot^{\gamma - 3} \left[ \gamma x_i + 2 \xtotmi \right].
    \end{align*}
    For $\gamma \in [0,1]$, the term $K (\gamma - 1) \xtot^{\gamma - 3} \left[ \gamma x_i + 2 \xtotmi \right]$ is non-positive since $\gamma - 1 \leq 0$, while all other factors are non-negative. Therefore, $\frac{\partial^2 R_i}{\partial x_i^2} = \frac{\partial A_i}{\partial x_i} \leq 0$, establishing that $R_i$ is concave in $x_i$. Additionally, since $c_i$ is strictly convex, $-c_i$ is strictly concave. Thus, $U_i = R_i - c_i$ is strictly concave in $x_i$.

    \paragraph{Step 4: Monotonicity of the game}
    To establish monotonicity, we need to show for any distinct $\x, \x' \in \mathcal{X}$:
    \[
        \sum_{i=1}^{n} (v_i(\x') - v_i(\x))(x_i' - x_i) < 0.
    \]

    Substituting our gradient expressions:
    \begin{align*}
        \sum_{i=1}^{n} (v_i(\x') - v_i(\x))(x_i' - x_i) &= \sum_{i=1}^{n} \left[ (A_i(\x') - c_i'(x_i')) - (A_i(\x) - c_i'(x_i)) \right] (x_i' - x_i) \\
        &= \sum_{i=1}^{n} (A_i(\x') - A_i(\x))(x_i' - x_i) - \sum_{i=1}^{n} (c_i'(x_i') - c_i'(x_i))(x_i' - x_i).
    \end{align*}

    For the second term, since $c_i$ is strictly convex, we have $(c_i'(x_i') - c_i'(x_i))(x_i' - x_i) > 0$. Therefore:
    \[
        -\sum_{i=1}^{n} (c_i'(x_i') - c_i'(x_i))(x_i' - x_i) < 0.
    \]

    For the first term, we need to analyze $\sum_{i=1}^{n} (A_i(\x') - A_i(\x))(x_i' - x_i)$ in detail. We start with $A_i(\x') - A_i(\x)$. By the multivariable mean value theorem, there exists a point $\hat{\x}$ on the line segment between $\x$ and $\x'$ such that:
    \[
        A_i(\x') - A_i(\x) = \sum_{j=1}^n \frac{\partial A_i}{\partial x_j}(\hat{\x}) (x'_j - x_j).
    \]

    We now compute those partial derivatives. For $j = i$, it holds from our previous analysis that:
    \[
        \frac{\partial A_i}{\partial x_i}(\x) = K \xtot^{\gamma - 3} (\gamma - 1) \left[ \gamma x_i + 2 \xtotmi \right].
    \]

    For $j \neq i$:
    \[
        \frac{\partial A_i}{\partial x_j}(\x) = K \frac{\partial}{\partial x_j} \left[ \xtot^{\gamma - 2} (\gamma x_i + \xtotmi) \right].
    \]

    Since $\frac{\partial \xtot}{\partial x_j} = 1$ and $\frac{\partial \xtotmi}{\partial x_j} = 1$ (as $j \neq i$ implies $x_j$ is part of $\xtotmi$):
    \begin{align*}
        \frac{\partial A_i}{\partial x_j}(\x) &= K \left[ (\gamma - 2) \xtot^{\gamma - 3} (\gamma x_i + \xtotmi) + \xtot^{\gamma - 2} \right] \\
        &= K \xtot^{\gamma - 3} \left[ (\gamma - 2) (\gamma x_i + \xtotmi) + \xtot \right] \\
        &= K \xtot^{\gamma - 3} \left[ (\gamma^2 - 2\gamma) x_i + (\gamma - 2) \xtotmi + (x_i + \xtotmi) \right] \\
        &= K \xtot^{\gamma - 3} \left[ (\gamma^2 - 2\gamma + 1) x_i + (\gamma - 1) \xtotmi \right] \\
        &= K (\gamma - 1) \xtot^{\gamma - 3} \left[ (\gamma - 1) x_i + \xtotmi \right].
    \end{align*}

    Now, we return to our expression for the difference:
    \begin{align*}
        A_i(\x') - A_i(\x) &= \sum_{j=1}^n \frac{\partial A_i}{\partial x_j}(\hat{\x}) (x'_j - x_j) \\
        &= \frac{\partial A_i}{\partial x_i}(\hat{\x}) (x'_i - x_i) + \sum_{j \neq i} \frac{\partial A_i}{\partial x_j}(\hat{\x}) (x'_j - x_j) \\
        &= K \norm{\hat{\x}}^{\gamma - 3} (\gamma - 1) \left[ \gamma \hat{x}_i + 2 \norm{\hat{\x}_{-i}} \right] (x'_i - x_i) \\
        &+ \sum_{j \neq i} K \norm{\hat{\x}}^{\gamma - 3} (\gamma - 1) \left[ (\gamma - 1) \hat{x}_i + \norm{\hat{\x}_{-i}} \right] (x'_j - x_j).
    \end{align*}

    Multiplying each term by $(x'_i - x_i)$ and summing over all $i$:
    \begin{align*}
        \sum_{i=1}^{n} (A_i(\x') &- A_i(\x))(x_i' - x_i) = \sum_{i=1}^{n} K \norm{\hat{\x}}^{\gamma - 3} (\gamma - 1) \left[ \gamma \hat{x}_i + 2 \norm{\hat{\x}_{-i}} \right] (x'_i - x_i)^2 \\
        &+ \sum_{i=1}^{n} \sum_{j \neq i} K \norm{\hat{\x}}^{\gamma - 3} (\gamma - 1) \left[ (\gamma - 1) \hat{x}_i + \norm{\hat{\x}_{-i}} \right] (x'_j - x_j)(x'_i - x_i).
    \end{align*}

    Factoring out $K \norm{\hat{\x}}^{\gamma - 3} (\gamma - 1)$, which is a constant with respect to the summation:
    \begin{align*}
        = K \norm{\hat{\x}}^{\gamma - 3} (\gamma - 1) &\bigl[ \sum_{i=1}^{n} \left[ \gamma \hat{x}_i + 2 \norm{\hat{\x}_{-i}} \right] (x'_i - x_i)^2 \\
        &+ \sum_{i=1}^{n} \sum_{j \neq i} \left[ (\gamma - 1) \hat{x}_i + \norm{\hat{\x}_{-i}} \right] (x'_j - x_j)(x'_i - x_i) \bigr].
    \end{align*}

    Next, denote $\mathbf{d} = \x' - \x$, and let $B_{ij} = \left((\gamma - 1) \hat{x}_i + \norm{\hat{\x}_{-i}} \right)$ for $i \neq j$. Notice that the double summation can be rewritten as $\sum_{i \neq j} B_{ij} d_i d_j$. We now simplify this expression further. The sum $B_{ij} + B_{ji}$ can be expressed as:
    \[
        B_{ij} + B_{ji} = (\gamma - 1)(\hat{x}_i + \hat{x}_j) + \norm{\hat{\x}_{-i}} + \norm{\hat{\x}_{-j}} = \gamma (\hat{x}_i + \hat{x}_j) + 2 \sum_{k \neq i,j} \hat{x}_k.
    \]
    Because the double summation counts each of the pairs $(i,j)$ and $(j,i)$, we can rewrite it as:
    \[
        \sum_{i=1}^{n} \sum_{j \neq i} B_{ij} d_i d_j = \frac{1}{2} \sum_{i=1}^{n} \sum_{j \neq i} (B_{ij} + B_{ji}) d_i d_j = \frac{1}{2} \sum_{i=1}^{n} \sum_{j \neq i} \left[ \gamma (\hat{x}_i + \hat{x}_j) + 2 \sum_{k \neq i,j} \hat{x}_k \right] d_i d_j.
    \]

    Thus, we can factor $0.5$ out of both the first expression and the double summation, leading to:
    \begin{align*}
        &\sum_{i=1}^{n} (A_i(\x') - A_i(\x))(x_i' - x_i) \\
        &\qquad= \frac{1}{2} K (\gamma - 1) \norm{\hat{\x}}^{\gamma - 3} \Bigg[
            \sum_{i=1}^{n} d_i^2 \left( 2\gamma \hat{x}_i + 4 \norm{\hat{\x}_{-i}} \right)
            + \sum_{i=1}^{n} \sum_{j \neq i} d_i d_j \left( \gamma (\hat{x}_i + \hat{x}_j) + 2 \sum_{k \neq i,j} \hat{x}_k \right)
        \Bigg].
    \end{align*}

    The expression inside the square brackets matches the structure of $\mathcal{Q}_\gamma(\hat{\x}, \mathbf{d})$ as defined in \Cref{lem:quadratic_form_helper}. By \Cref{lem:quadratic_form_helper}, we know that $\mathcal{Q}_\gamma(\hat{\x}, \mathbf{d}) \geq 0$ for any $\hat{\x} \geq 0$ and any $\mathbf{d}$. Since $K > 0$, $\gamma - 1 \leq 0$ (as $\gamma \in [0,1]$), and $\norm{\hat{\x}}^{\gamma-3} > 0$ (because $\norm{\hat{\x}} > 0$ for non-trivial strategy profiles), we have:
    \[
        K(\gamma-1)\norm{\hat{\x}}^{\gamma-3}\mathcal{Q}_\gamma(\hat{\x}, \mathbf{(x' - x)}) \leq 0.
    \]

    Therefore:
    \[
        \sum_{i=1}^{n} (A_i(\x') - A_i(\x))(x_i' - x_i) \leq 0.
    \]

    Combining this with our earlier inequality for the cost function terms:
    \[
        \sum_{i=1}^{n} (v_i(\x') - v_i(\x))(x_i' - x_i) < 0.
    \]

    This establishes that the game is monotone.

    \paragraph{Step 5: Existence and uniqueness of equilibrium}
    Since the enforced sharing game satisfies:
    \begin{itemize}
        \item Strategy spaces $[0, x_i^{\max}]$ are compact and convex.
        \item Utility functions are continuous.
        \item Utility functions are strictly concave in each player's own strategy.
        \item The game is monotone.
    \end{itemize}

    By Theorem 2 in \citet{rosen1965existence}, the subgame $G(\rho)$ admits a unique pure Nash equilibrium in the enforced sharing game, which is the unique ESE.
\end{proof}

\PropUniquePNERhoOne*

\begin{proof}[\proofof{prop:unique_pne_rho1}]
    We prove the result for the more general GenAI quality function $\qai = \alpha(\stot)^\beta$ with $\beta \in (0,1]$; the linear case $\beta=1$ used in the main text follows as a special case. The proof proceeds in two parts: We first establish the existence and uniqueness of an FSE, and then show that no other PNE exists.

    \paragraph{Existence and uniqueness of FSE}
    We begin by showing that under $\rho=1$, full sharing is the unique optimal choice for any creator with positive quality. Fix a creator $i$ with $x_i > 0$, and suppose $\xtotmi > 0$. Excluding the cost term $c_i(x_i)$, which is independent of $s_i$, creator $i$'s utility as a function of $s_i$ can be written as
    \[
    \widehat{U}_i(s_i) = \Tx \left[ \frac{x_i + \alpha x_i s_i (\stot)^{\beta-1}}{\xtot + \alpha (\stot)^\beta} \right].
    \]
    Differentiating with respect to $s_i$ and simplifying yields
    \[
    \frac{d\widehat{U}_i}{ds_i} = \Tx \cdot \frac{\alpha x_i (\stot)^{\beta-2}}{(\xtot + \alpha (\stot)^\beta)^2} \left\{ \beta x_i s_i \xtotmi + \mathbf{x}^{\top}_{-i} \mathbf{s}_{-i} \left[ (1-\beta) x_i + \xtotmi + \alpha (\stot)^\beta \right] \right\}.
    \]
    Under our assumptions ($\alpha > 0$, $x_i > 0$, $\xtotmi > 0$), the fraction is strictly positive for all $s_i \in (0,1)$. Moreover, the first term in braces is strictly positive, and the second is non-negative. Hence $\frac{d\widehat{U}_i}{ds_i} > 0$ on $(0,1)$, implying that $\widehat{U}_i$ is strictly increasing, and therefore $s_i = 1$ is the unique optimum.

    Next, we show that under $\rho = 1$ and full sharing, the utility function simplifies to a form that does not depend on $\beta$. When $\mathbf{s} = \mathbf{1}$, we have $\stot = \xtot$ and $f_{i,1}(\mathbf{x}, \mathbf{1}) = \frac{x_i}{\xtot}$. Substituting into the general utility expression:
    \[
    U_i(\mathbf{x}, \mathbf{1}; f_1) = \mu \xtot^\gamma \cdot \frac{x_i}{\xtot + \alpha \xtot^\beta} + \mu \xtot^\gamma \cdot \frac{\alpha \xtot^\beta}{\xtot + \alpha \xtot^\beta} \cdot \frac{x_i}{\xtot} - c_i(x_i).
    \]
    Combining the revenue terms over a common denominator and factoring, we obtain
    \[
    U_i(\mathbf{x}, \mathbf{1}; f_1) = \frac{\mu \xtot^\gamma x_i (\xtot + \alpha \xtot^\beta)}{(\xtot + \alpha \xtot^\beta) \xtot} - c_i(x_i) = \mu \xtot^{\gamma - 1} x_i - c_i(x_i).
    \]
    This simplified form is identical to that analyzed in the proof of \Cref{lem:unique-ese} for $\beta = 1$. Since that analysis depends only on the structure $\mu \xtot^{\gamma-1} x_i - c_i(x_i)$ and not on $\beta$, we conclude that the enforced sharing game under $\rho = 1$ admits a unique ESE, which we denote by $\mathbf{x}^\star$.

    To establish that $(\mathbf{x}^\star, \mathbf{1})$ is the unique FSE, consider any deviation by creator $i$ to $(x_i', s_i')$. By a straightforward extension of \Cref{lem:dominance-positive-quality} to general $\beta$, we have $x_j^\star > 0$ for all $j$, ensuring $\norm{\mathbf{x}^\star_{-i}} > 0$. By the analysis above, any profitable deviation must have $s_i' = 1$. But then the best response in quality is $x_i' = x_i^\star$, since $\mathbf{x}^\star$ is an ESE. Thus $(\mathbf{x}^\star, \mathbf{1})$ is an FSE, and uniqueness follows from the uniqueness of the ESE.

    \paragraph{No other PNE exists}
    Suppose $(\mathbf{x}', \mathbf{s}')$ is an arbitrary PNE. By the extension of \Cref{lem:dominance-positive-quality}, we must have $x_i' > 0$ for all $i \in [n]$. Since $n \geq 2$, this implies $\norm{\mathbf{x}'_{-i}} > 0$ for each creator $i$. By the optimality of full sharing established above, we must have $\mathbf{s}' = \mathbf{1}$. But then the quality profile $\mathbf{x}'$ must constitute an ESE, and by uniqueness, $\mathbf{x}' = \mathbf{x}^\star$. Therefore, $(\mathbf{x}^\star, \mathbf{1})$ is the unique PNE.
\end{proof}

\ThmFSEExistence*

\begin{proof}[\proofof{thm:fse-existence}]
    We establish that when $\rho$ satisfies the condition in the theorem statement, the unique ESE with full sharing constitutes the unique FSE of the game. The proof proceeds in three main steps: First, we analyze properties of the ESE; second, we demonstrate that full sharing is strictly dominant under the condition on $\rho$; and finally, we verify that the resulting profile is the unique FSE.

    \paragraph{Step 1: Properties of the ESE profile}
    Recall that the existence and uniqueness of the ESE profile $\xrho$ are guaranteed by \Cref{lem:unique-ese}. By \Cref{lem:dominance-positive-quality}, any rational strategy must involve a strictly positive quality choice, which implies $\mathcal{X}_i(\rho) > 0$ for all $i \in [n]$. This allows us to use \Cref{lem:s_choice_threshold} in the next step.

    \paragraph{Step 2: Optimality of full sharing under the specified condition}
    We examine when full sharing ($s_i = 1$) is the unique best response for creator $i$, given that all other creators $j \neq i$ play $(x_j^\star, s_j = 1)$.

    By \Cref{lem:s_choice_threshold}, for any fixed quality $x_i > 0$, full sharing is the unique optimal choice if and only if:
    \[
        \rho > \tau_i(x_i) \coloneqq \frac{x_i}{x_i + \norm{\mathcal{X}_{-i}(\rho)} + \alpha\norm{\mathcal{X}_{-i}(\rho)}} = \frac{x_i}{x_i + (1+\alpha)\norm{\mathcal{X}_{-i}(\rho)}}.
    \]

    The derivative of $\tau_i$ with respect to $x_i$ is:
    \[
        \frac{\partial \tau_i}{\partial x_i} = \frac{(1+\alpha)\norm{\mathcal{X}_{-i}(\rho)}}{(x_i + (1+\alpha)\norm{\mathcal{X}_{-i}(\rho)})^2} > 0.
    \]

    This confirms that $\tau_i(x_i)$ is strictly increasing in $x_i > 0$. Recall that \Cref{lem:rational_quality_bound} ensures that for any rational quality choice, we must have $x_i \leq x_i^{\max}$. Since $\tau_i(x_i)$ is strictly increasing in $x_i$ and any rational quality choice satisfies $x_i \leq x_i^{\max}$, the maximum threshold value for creator $i$ occurs at $x_i = x_i^{\max}$, yielding:
    \[
        \max_{x_i \in (0, x_i^{\max}]} \tau_i(x_i) = \tau_i(x_i^{\max}) = \frac{x_i^{\max}}{x_i^{\max} + (1+\alpha)\norm{\mathcal{X}_{-i}(\rho)}}.
    \]

    The theorem's condition requires:
    \[
    \rho > \max_{i \in [n]} \frac{x_i^{\max}}{x_i^{\max} + (1+\alpha)\norm{\mathcal{X}_{-i}(\rho)}}.
    \]

    This guarantees that for all creators $i \in [n]$ and all rational quality choices $x_i \in (0, x_i^{\max}]$, the inequality $\rho > \tau_i(x_i)$ holds. Consequently, full sharing is the strictly dominant strategy for each creator at any rational quality level.

    \paragraph{Step 3: Verification of FSE properties}
    Consider an arbitrary creator $i$ contemplating a deviation from $(\mathcal{X}_i(\rho), 1)$ to any alternative strategy $(x_i', s_i')$.

    By \Cref{lem:rational_quality_bound}, if $(x_i', s_i')$ is to be a beneficial deviation, we must have $x_i' \leq x_i^{\max}$. For any such rational quality choice $x_i' > 0$, Step 2 has established that $s_i = 1$ strictly dominates any other sharing level when $\rho$ satisfies the theorem's condition. Therefore, regardless of the quality choice $x_i'$, creator $i$ would maximize utility by choosing $s_i = 1$ rather than any $s_i' < 1$.

    This means that any rational deviation must effectively reduce to a quality-only deviation with full sharing, i.e., $(x_i', 1)$. However, from the definition of an ESE, $\mathcal{X}_i(\rho)$ is the unique best response to $(\x_{-i}^\star, \mathbf{1}_{-i})$ among all strategies with $s_i = 1$. Therefore, no deviation to any $(x_i', 1)$ can be profitable.

    Combining these observations, we conclude that no creator $i$ has a profitable deviation from $(\mathcal{X}_i(\rho), 1)$, making $(\xrho, \mathbf{1})$ a PNE. Moreover, uniqueness follows from the observation that any FSE must be an ESE when restricted to the subgame with $\mathbf{s} = \mathbf{1}$. Since \Cref{lem:unique-ese} establishes that the ESE is unique, there can be no other FSE. This completes the proof of \Cref{thm:fse-existence}.
\end{proof}

\PropAsymptoticStabilityInterval*

\begin{proof}[\proofof{prop:asymptotic_stability_interval}]
    Fix $w \in (0,1]$. By \Cref{thm:fse-existence}, a sufficient condition for $(\xr{\rho}, \mathbf{1})$ to be the unique FSE is:
    \[
        \rho > \max_{i \in [n]} \frac{x_i^{\max}}{x_i^{\max} + (1+\alpha) \norm{\mathcal{X}_{-i}(\rho)}}.
    \]
    Since we consider $\rho \in (w, 1]$, it suffices to show that for sufficiently large $n$, the RHS is bounded by $w$.

    First, we bound the numerator. Recall that $x_i^{\max}$ satisfies $a_i (x_i^{\max})^\theta = \mu (x_i^{\max})^\gamma$. Thus, for all $i$:
    \begin{equation} \label{eq:numerator_bound}
        x_i^{\max} = \left( \frac{\mu}{a_i} \right)^{\frac{1}{\theta - \gamma}} \le \left( \frac{\mu}{a_{\min}} \right)^{\frac{1}{\theta - \gamma}} \eqqcolon M.
    \end{equation}
    
    Next, we bound the denominator using \Cref{lem:bound_on_total_quality_in_ese_power_costs}, which establishes that $\norm{\mathcal{X}_{-i}(\rho)} = \Theta(n^{\frac{\theta-1}{\theta-\gamma}})$ for any $\rho \in [0,1]$. Specifically, there exists a constant $K > 0$ (depending on $\mu, \alpha, \theta, \gamma, a_{\min}, a_{\max}$, but not on $n$) such that for all sufficiently large $n$:
    \begin{equation} \label{eq:denominator_bound}
         \norm{\mathcal{X}_{-i}(\rho)} \ge K \cdot n^{\frac{\theta-1}{\theta-\gamma}}.
    \end{equation}
    
    Substituting \eqref{eq:numerator_bound} and \eqref{eq:denominator_bound} into the sufficient condition, we require:
    \[
        \frac{M}{M + (1+\alpha) K n^{\frac{\theta-1}{\theta-\gamma}}} \le w.
    \]
    Rearranging to isolate $n$:
    \begin{align*}
        M &\le w \left( M + (1+\alpha) K n^{\frac{\theta-1}{\theta-\gamma}} \right) \\
        M (1-w) &\le w (1+\alpha) K n^{\frac{\theta-1}{\theta-\gamma}} \\
        n^{\frac{\theta-1}{\theta-\gamma}} &\ge \frac{M(1-w)}{w(1+\alpha)K}.
    \end{align*}
    Since $\theta > 1$ and $\theta > \gamma$, the exponent $\frac{\theta-1}{\theta-\gamma}$ is positive. Thus, there exists a threshold $N(w)$ such that for all $n > N(w)$, the inequality holds. This ensures the FSE condition is met for all $\rho \in (w, 1]$.
\end{proof}

\PropMonotonicityHomogeneous*

\begin{proof}[\proofof{prop:monotonicity_homogeneous}]
    By symmetry and \Cref{lem:unique-ese}, the equilibrium is unique and symmetric, characterized by a single quality level $x(\rho)$--we have $x_i = x(\rho)$ for all $i$, which implies $\xtot = n x(\rho)$. Differentiating $U_i$ with respect to $x_i$ yields:
    \[
        \frac{\partial U_i}{\partial x_i} = \frac{\mu(1+\alpha\rho)}{1+\alpha} \xtot^{\gamma-2} \left[ (\gamma-1)x_i + \xtot \right] - c'(x_i).
    \]
    Setting the derivative to zero and substituting the power cost derivative $c'(x) = a \theta x^{\theta-1}$ yields the following FOC (which is enforcing, as \Cref{lem:dominance-positive-quality} asserts that all the qualities are positive):
    \[
        \frac{\mu (1 + \alpha \rho)}{1 + \alpha} (n x(\rho))^{\gamma-2} \left[ (\gamma-1)x(\rho) + n x(\rho) \right] = a \theta x(\rho)^{\theta - 1}.
    \]
    Simplifying the term in brackets to $x(\rho)(n+\gamma-1)$ and rearranging to isolate $x(\rho)$, we obtain:
    \[
        x(\rho)^{\theta - \gamma} = C \cdot (1 + \alpha \rho),
    \]
    where $C > 0$ is a constant independent of $\rho$. Since $\alpha > 0$ and $\theta > 1 \ge \gamma$ (implying $\theta - \gamma > 0$), the right-hand side is strictly increasing in $\rho$. Consequently, $x(\rho)$, and thus total quality $\norm{\xr{\rho}} = n x(\rho)$, must be strictly increasing in $\rho$.
\end{proof}

%% file: appendix/section4_proofs.tex
\section{Omitted Proofs From \Cref{sec:platform-optimization}}\label{sec:sec4_proofs}

This section is mainly devoted to proving \Cref{thm:alg_main}, along with the other results of \Cref{sec:platform-optimization}. The road to the actual proof of \Cref{thm:alg_main} contains extensive Lipschitz continuity analysis of various auxiliary functions. The structure of this section is as follows:
\begin{enumerate}
    \item We first introduce some notations and prove a lower bound on the quality produced by each creator in any ESE in \Cref{app:ese_quality_lower_bound_and_notations}.
    \item Next, in \Cref{app:technical_lipschitz_continuity_results}, we prove several technical Lipschitz continuity results that will be used in subsequent sections.
    \item In \Cref{app:lipshitz_results_for_main_algorithm}, we prove Lipschitz continuity results that are directly used in the proof of \Cref{alg:main}, which build upon the technical results of \Cref{app:technical_lipschitz_continuity_results}.
\end{enumerate}

\subsection{ESE Quality Lower Bound and Notations}\label{app:ese_quality_lower_bound_and_notations}

\begin{remark}
    For the rest of this section, we will denote for every player $i \in [n]$:
    \[
        x_i^{\min} \coloneqq \left( c_i' \right)^{-1} \left( \frac{\mu}{1 + \alpha} \left( \sum_{j=1}^{n} x_j^{\max} \right)^{\gamma - 1} \gamma \right).
    \]
\end{remark}

\begin{lemma}\label{lem:ese_quality_lower_bound}
    Fix an arbitrary $\rho$. Then, for every creator $i \in [n]$, it holds that $\mathcal{X}_i(\rho) \ge x_i^{\min}$.
\end{lemma}

\begin{proof}[\proofof{lem:ese_quality_lower_bound}]
    Fix any creator $i$. Recall that its FOC in the ESE is:
    \[
        \frac{\mu (1 + \alpha \rho)}{1 + \alpha} \norm{\xr{\rho}}^{\gamma - 2} \left[ (\gamma - 1) \mathcal{X}_i(\rho) + \norm{\xr{\rho}} \right] = c_i'(\mathcal{X}_i(\rho)).
    \]

    Next, we rearrange the LHS of the FOC and bound it from below:
    \begin{align*}
        c_i'(\mathcal{X}_i(\rho)) &= \frac{\mu (1 + \alpha \rho)}{1 + \alpha} \norm{\xr{\rho}}^{\gamma - 2} \left[ (\gamma - 1) \mathcal{X}_i(\rho) + \norm{\xr{\rho}} \right]\\
        &\geq \frac{\mu}{1 + \alpha} \norm{\xr{\rho}}^{\gamma - 1} \left[ 1 + (\gamma - 1) \frac{\mathcal{X}_i(\rho)}{\norm{\xr{\rho}}} \right] \\
        &\geq \frac{\mu}{1 + \alpha} \norm{\xr{\rho}}^{\gamma - 1} \gamma \\
        &\geq \frac{\mu}{1 + \alpha} \left( \sum_{j=1}^{n} x_j^{\max} \right)^{\gamma - 1} \gamma,
    \end{align*}
    Where the last inequality follows from \Cref{lem:rational_quality_bound} and the fact that $\gamma - 1 \leq 0$. Recall that $c_i'(\cdot)$ is strictly increasing since $c_i(\cdot)$ is strictly convex; thus, we have:
    \[
        \mathcal{X}_i(\rho) \geq \left( c_i' \right)^{-1} \left( \frac{\mu}{1 + \alpha} \left( \sum_{j=1}^{n} x_j^{\max} \right)^{\gamma - 1} \gamma \right) = x_i^{\min}.
    \]
\end{proof}

\begin{remark}\label{rem:B-notation}
    Throughout this section, we use the following notation:
    \begin{itemize}
        \item $X_{LB} =\sum_{i = 1}^{n} x_i^{\min}$ and $X_{UB} = \sum_{j=1}^{n} x_j^{\max}$ are lower and upper bounds on the total quality produced in any ESE, respectively.
        \item $Y_{LB} =\sum_{i = 1}^{n} x_i^{\min} - \max_{i \in [n]} x_i^{\min}$ and $Y_{UB} = \sum_{j=1}^{n} x_j^{\max} - \min_{i \in [n]} x_i^{\max}$ are lower and upper bounds on the total quality produced by all creators except any single creator in any ESE, respectively.
        \item $x^{\max} = \max_{i \in [n]} x_i^{\max}$ is an upper bound on the quality any single creator can produce.
        \item $M_{c'} = \max_{i \in [n]} \sup_{x \in [0, x^{\max}]} |c_i'(x)|$ is a uniform bound on the derivatives of all cost functions over the interval $[0, x^{\max}]$.
        \item $B(q,u,v,w,t)$ is a parametrized bound defined as below for non-negative scalars $q, u, v, w, t$:
        \[
            B(q,u,v,w,t) = q \cdot M_{c'} + \frac{\mu \cdot u \cdot n^{v}(1 + x^{\max} + Y_{UB})^{w}}{(Y_{LB})^{t}}.
        \]
        Notice that:
        \begin{enumerate}
            \item Increasing both $w$ and $t$ by the same amount increases $B$.
            \item \begin{align*}
                &\max\{B(q_1,u_1,v_1,w_1,t), B(q_2,u_2,v_2,w_2,t)\} \\
                &\leq B(\max\{q_1,q_2\}, \max\{u_1,u_2\}, \max\{v_1,v_2\}, \max\{w_1,w_2\}, t).
            \end{align*}
        \end{enumerate}
        \item We will also use the notation $\tilde B(q,u,v,w,t)$ to denote a bound of the same form as $B(q,u,v,w,t)$ but with $Y_{LB}, Y_{UB}$ replaced by $0.5Y_{LB}, 2Y_{UB}$, respectively.
    \end{itemize}
\end{remark}

\subsection{Technical Lipschitz Continuity Results}\label{app:technical_lipschitz_continuity_results}

\begin{lemma}\label{lem:lipschitz_phi_0}
    For $(y, z, \sigma) \in [Y_{LB}, Y_{UB}] \times [0, x^{\max}] \times [0,1]$, define the function
    \[
        \phi^0(y, z, \sigma) = \max_{x \in [0, x^{\max}]} \left\{ \frac{\mu x (x+y)^\gamma}{x + (1+\alpha)y} - c(x) \right\} - \frac{\mu (1 + \alpha\sigma)}{1+\alpha} z (z+y)^{\gamma-1} + c(z).
    \]
    Then for all $(y_1, z_1, \sigma_1)$, $(y_2, z_2, \sigma_2)$ in this domain,
    \begin{align*}
        &\abs{\phi^0(y_1, z_1, \sigma_1) - \phi^0(y_2, z_2, \sigma_2)} \\
        &\leq B(1, 1, 0, 1, 2 - \gamma) \abs{z_1 - z_2} + B(0, 2 + \alpha, 0, 2, 3 - \gamma) \abs{y_1 - y_2} +  B(0, 1, 0, 1, 1 - \gamma) \abs{\sigma_1 - \sigma_2}.
    \end{align*}
\end{lemma}

\begin{proof}[\proofof{lem:lipschitz_phi_0}]
    Define the following auxiliary functions:
    \begin{align*}
        F(x, y) &\coloneqq \frac{\mu x (x+y)^\gamma}{x + (1+\alpha)y} - c(x) , \\
        U(y) &\coloneqq \max_{x \in [0, x^{\max}]} F(x, y), \\
        A(y, z, \sigma) &\coloneqq \frac{\mu (1 + \alpha\sigma)}{1+\alpha} z (z+y)^{\gamma-1} - c(z), \\
    \end{align*}
    so that $\phi^0(y, z, \sigma) = U(y) - A(y, z, \sigma)$. We bound the sensitivity of $\phi^0$ with respect to each variable.

    \paragraph{Lipschitz constant in $z$}
    $U(y)$ does not depend on $z$. We consider the $z$-derivative of $-A(y, z, \sigma)$:
    \begin{align*}
        \frac{\partial -A}{\partial z}(y, z, \sigma)
        &= -\frac{\mu (1 + \alpha\sigma)}{1+\alpha} \left((z+y)^{\gamma-1} + z(\gamma-1)(z+y)^{\gamma-2}\right) + c'(z) \\
        &= -\frac{\mu (1 + \alpha\sigma)}{1+\alpha} (z+y)^{\gamma-2} (y + \gamma z) + c'(z).
    \end{align*}

    Bounding the absolute value of this derivative yields:
    \[
        \abs{\frac{\partial}{\partial z}\left(-A(y, z, \sigma)\right)} \leq \mu (Y_{LB})^{\gamma-2} (Y_{UB} + x^{\max}) + M_{c'} \leq B(1, 1, 0, 1, 2 - \gamma).
    \]

    \paragraph{Lipschitz constant in $\sigma$}
    $U(y)$ does not depend on $\sigma$. The only dependence comes from $A(y, z, \sigma)$. Compute:
    \[
        \frac{\partial}{\partial \sigma} \left( -A(y, z, \sigma) \right)
        = -\frac{\partial A}{\partial \sigma}(y, z, \sigma)
        = -\frac{\mu \alpha}{1+\alpha} z (z+y)^{\gamma-1}.
    \]
    Taking absolute values and bounding:
    \[
        \abs{\frac{\partial}{\partial \sigma} \left( -A(y, z, \sigma) \right)} \leq \mu x^{\max} (Y_{LB})^{\gamma-1} \leq B(0, 1, 0, 1, 1 - \gamma).
    \]

    \paragraph{Lipschitz constant in $y$}
    Both $U(y)$ and $A(y, z, \sigma)$ depend on $y$.

    First, for $U(y)$:
    Let $y_1, y_2 \in [Y_{LB}, Y_{UB}]$.
    Let $x_1^\star$ maximize $F(x, y_1)$, $x_2^\star$ maximize $F(x, y_2)$. By optimality,
    \[
        U(y_1) - U(y_2) \leq F(x_1^\star, y_1) - F(x_1^\star, y_2).
    \]
    By the Mean Value Theorem, for some $\tilde{y}_1$ between $y_1$ and $y_2$,
    \[
        F(x_1^\star, y_1) - F(x_1^\star, y_2) = \frac{\partial F}{\partial y}(x_1^\star, \tilde{y}_1) (y_1 - y_2).
    \]
    Similarly,
    \[
        U(y_2) - U(y_1) \leq \frac{\partial F}{\partial y}(x_2^\star, \tilde{y}_2) (y_2 - y_1).
    \]
    Thus,
    \[
        \abs{U(y_1) - U(y_2)} \leq \sup_{x \in [0, x^{\max}],\ y \in [Y_{LB}, Y_{UB}]} \abs{ \frac{\partial F}{\partial y}(x, y) } \abs{y_1 - y_2}.
    \]

    Now, we bound the derivative $\frac{\partial F}{\partial y}(x, y)$. This derivative is given by:
    \[
        \frac{\partial F}{\partial y}(x, y) = \mu x \frac{(x+y)^{\gamma-1} \left[ (\gamma-1-\alpha)x + (\gamma-1)(1+\alpha)y \right]}{(x + (1+\alpha)y)^2}.
    \]

    Bounding this expression yields:
    \[
        \sup_{x, y} \abs{ \frac{\partial F}{\partial y}(x, y) } \leq \frac{\mu x^{\max}(1+\alpha)}{(Y_{LB})^{3-\gamma}} (x^{\max} + Y_{UB}).
    \]

    Next, for $-A(y, z, \sigma)$, the derivative is
    \[
        -\frac{\partial A}{\partial y} = -\frac{\mu (1+\alpha\sigma)}{1+\alpha} z (\gamma-1) (z+y)^{\gamma-2}.
    \]
    This expression's absolute value is bounded by:
    \[
        \abs{ -\frac{\partial A}{\partial y} } \leq \mu x^{\max} (Y_{LB})^{\gamma-2}.
    \]

    Therefore, the bound for the Lipschitz constant in $y$ is:
    \[
        \mu x^{\max} (Y_{LB})^{\gamma-2} + \frac{\mu x^{\max}(1+\alpha)}{(Y_{LB})^{3-\gamma}} (x^{\max} + Y_{UB}) \leq B(0, 2 + \alpha, 0, 2, 3 - \gamma).
    \]

    \paragraph{Conclusion} Combining the three Lipschitz constants, we have:
    \begin{align*}
        &\abs{\phi^0(y_1, z_1, \sigma_1) - \phi^0(y_2, z_2, \sigma_2)} \\
        &\leq B(1, 1, 0, 1, 2 - \gamma) \abs{z_1 - z_2} + B(0, 2 + \alpha, 0, 2, 3 - \gamma) \abs{y_1 - y_2} + B(0, 1, 0, 1, 1 - \gamma) \abs{\sigma_1 - \sigma_2}.
    \end{align*}
    This completes the proof of \Cref{lem:lipschitz_phi_0}.
\end{proof}

\begin{lemma}\label{lem:lipschitz_phi_1}
    Define $K(\sigma) \coloneqq \frac{\mu (1 + \alpha\sigma)}{1+\alpha}$, and define for any $(z, y, \sigma) \in [0, x^{\max}] \times [Y_{LB}, Y_{UB}] \times [0, 1]$:
    \[
        \phi^1(z, y, \sigma) = \max_{x \in [0, x^{\max}]} \left\{ K(\sigma) x (x+y)^{\gamma-1} - c(x) \right\} - \left[ K(\sigma) z (z+y)^{\gamma-1} - c(z) \right],
    \]

    Then $\phi^1$ is Lipschitz continuous on its domain. Specifically, for any $(z_1, y_1, \sigma_1)$ and $(z_2, y_2, \sigma_2)$,
    \begin{align*}
        &\abs{\phi^1(z_1, y_1, \sigma_1) - \phi^1(z_2, y_2, \sigma_2)} \\
        &\leq B(1, 1, 0, 1, 2 - \gamma) \abs{z_1 - z_2} + B(0, 2, 0, 1, 2 - \gamma) \abs{y_1 - y_2} + B(0, 2, 0, 1, 1 - \gamma) \abs{\sigma_1 - \sigma_2}.
    \end{align*}
\end{lemma}

\begin{proof}[\proofof{lem:lipschitz_phi_1}]
    For clarity, define the following auxiliary functions:
    \begin{align*}
        G(x, y, \sigma) &\coloneqq K(\sigma) x (x+y)^{\gamma-1} - c(x), \\
        V_1(y, \sigma) &\coloneqq \max_{x \in [0, x^{\max}]} G(x, y, \sigma), \\
        H(z, y, \sigma) &\coloneqq K(\sigma) z (z+y)^{\gamma-1} - c(z).
    \end{align*}
    Then $\phi^1(z, y, \sigma) = V_1(y, \sigma) - H(z, y, \sigma)$.

    \paragraph{Lipschitz continuity with respect to $z$}
    $V_1(y, \sigma)$ does not depend on $z$. The function $H(z, y, \sigma)$ is differentiable in $z$:
    \[
        \frac{\partial H}{\partial z}(z, y, \sigma) = K(\sigma) (z+y)^{\gamma-2} (y + \gamma z) - c'(z).
    \]
    Its absolute value is bounded as follows:
    \[
        \abs{ \frac{\partial H}{\partial z}(z, y, \sigma) } \leq \mu (Y_{LB})^{\gamma-2} (Y_{UB} + x^{\max}) + M_{c'} \leq B(1, 1, 0, 1, 2 - \gamma).
    \]

    \paragraph{Lipschitz continuity with respect to $y$}
    Both $V_1(y, \sigma)$ and $H(z, y, \sigma)$ depend on $y$.
    \begin{itemize}
        \item For $V_1(y, \sigma)$: By the envelope theorem,
        \[
            \frac{\partial V_1}{\partial y}(y, \sigma) = K(\sigma) x^\star(y, \sigma) (\gamma-1) (x^\star(y, \sigma) + y)^{\gamma-2},
        \]
        where $x^\star(y, \sigma)$ is the unique maximizer of $G(\cdot, y, \sigma)$ (due to strict concavity of $G$). A bound on this expression is
        \[
            \abs{ \frac{\partial V_1}{\partial y}(y, \sigma) } \leq \mu x^{\max} (Y_{LB})^{\gamma-2}.
        \]
        \item For $H(z, y, \sigma)$:
        \[
            \frac{\partial H}{\partial y}(z, y, \sigma) = K(\sigma) z (\gamma-1) (z + y)^{\gamma-2}.
        \]
        With $z \leq x^{\max}$, $\abs{\gamma-1} \leq 1$, $(z + y)^{\gamma-2} \leq (Y_{LB})^{\gamma-2}$,
        \[
            \abs{ \frac{\partial H}{\partial y}(z, y, \sigma) } \leq \mu x^{\max} (Y_{LB})^{\gamma-2}.
        \]
    \end{itemize}
    Adding these, we get:
    \[
        2\mu x^{\max} (Y_{LB})^{\gamma-2} \leq B(0, 2, 0, 1, 2 - \gamma).
    \]

    \paragraph{Lipschitz continuity with respect to $\sigma$}
    Both $V_1(y, \sigma)$ and $H(z, y, \sigma)$ depend on $\sigma$. Let $K'(\sigma) = \frac{dK}{d\sigma} = \frac{\mu \alpha}{1+\alpha} \leq \mu$.
    \begin{itemize}
        \item For $V_1(y, \sigma)$: By the envelope theorem,
        \[
            \frac{\partial V_1}{\partial \sigma}(y, \sigma) = K'(\sigma) x^\star(y, \sigma) (x^\star(y, \sigma) + y)^{\gamma-1},
        \]
        where $x^\star(y, \sigma) \leq x^{\max}$, $(x^\star(y, \sigma) + y)^{\gamma-1} \leq (Y_{LB})^{\gamma-1}$,
        \[
            \abs{ \frac{\partial V_1}{\partial \sigma}(y, \sigma) } \leq \mu x^{\max} (Y_{LB})^{\gamma-1}.
        \]
        \item For $H(z, y, \sigma)$:
        \[
            \frac{\partial H}{\partial \sigma}(z, y, \sigma) = K'(\sigma) z (z + y)^{\gamma-1}.
        \]
        Similarly, $z \leq x^{\max}$, $(z + y)^{\gamma-1} \leq (Y_{LB})^{\gamma-1}$,
        \[
            \abs{ \frac{\partial H}{\partial \sigma}(z, y, \sigma) } \leq \mu x^{\max} (Y_{LB})^{\gamma-1}.
        \]
    \end{itemize}
    Summing, we get:
    \[
        2\mu x^{\max} (Y_{LB})^{\gamma-1} \leq B(0, 2, 0, 1, 1 - \gamma).
    \]

    \paragraph{Conclusion} Combining the three Lipschitz constants, we have:
    \begin{align*}
        &\abs{\phi^1(z_1, y_1, \sigma_1) - \phi^1(z_2, y_2, \sigma_2)} \\
        &\leq B(1, 1, 0, 1, 2 - \gamma) \abs{z_1 - z_2} + B(0, 2, 0, 1, 2 - \gamma) \abs{y_1 - y_2} + B(0, 2, 0, 1, 1 - \gamma) \abs{\sigma_1 - \sigma_2}. 
    \end{align*}
    This completes the proof of \Cref{lem:lipschitz_phi_1}
\end{proof}

\begin{lemma}\label{lem:liphsit_of_deviation}
    Let $\phi_k(\mathbf{z}, \sigma)$ be the maximum utility gain player $k$ can achieve by unilaterally deviating from the profile $(\mathbf{z}, \mathbf{1})$ when the platform's revenue sharing parameter is $\sigma$. For any profiles $\mathbf{x}, \mathbf{x}'$ such that $\norm{\mathbf{x} - \mathbf{x}'}_2 \leq \delta$ and $Y_{LB} \leq \norm{\mathbf{x}_{-k}}, \norm{\mathbf{x}'_{-k}} \leq Y_{UB}$, and any $\rho, \rho'$ such that $\abs{\rho - \rho'} \leq \delta$, we have:
    \[
        \abs{\phi_k(\mathbf{x}, \rho) - \phi_k(\mathbf{x}', \rho')} \leq \delta \left( B(1, 1, 0, 1, 2 - \gamma) + B(0, (2 + \alpha), 0.5, 2, 3 - \gamma) + B(0, 2, 0, 1, 1 - \gamma) \right)
    \]
\end{lemma}

\begin{proof}[\proofof{lem:liphsit_of_deviation}]
    The deviation gain function $\phi_k$ quantifies how much a player $k$ can benefit by deviating from the strategy profile $(\mathbf{z}, \mathbf{1})$. To analyze this function, we introduce the notation $z_k$ for player $k$'s quality and $y_k = \norm{\mathbf{z}_{-k}}$ for the sum of all other players' qualities. This allows us to express $\phi_k$ as a function of these scalar arguments: $\phi_k(\mathbf{z}, \sigma) = \phi_k(z_k, y_k, \sigma)$.

    \paragraph{Step 1: Lipschitz continuity of deviation gain function $\phi_k$}

    According to \Cref{lem:s_choice_threshold}, the optimal deviation involves either full sharing or no sharing. Thus, we can express the gain function as:
    \[
        \phi_k(z_k, y_k, \sigma) = \max\left(\phi_k^{(s'=0)}(z_k, y_k, \sigma), \phi_k^{(s'=1)}(z_k, y_k, \sigma)\right),
    \]
    where:
    \begin{align*}
        \phi_k^{(s'=1)}(z_k, y_k, \sigma) &= \left[\max_{x_k' \in [0, x^{\max}]} U_k((x_k', \mathbf{z}_{-k}),\mathbf{1};f_\sigma)\right] - U_k((z_k, \mathbf{z}_{-k}),\mathbf{1};f_\sigma), \\
        \phi_k^{(s'=0)}(z_k, y_k, \sigma) &= \left[\max_{x_k' \in [0, x^{\max}]} U_k((x_k', \mathbf{z}_{-k}),(0, \mathbf{1}_{-k});f_\sigma)\right] - U_k((z_k, \mathbf{z}_{-k}),\mathbf{1};f_\sigma).
    \end{align*}

    Recall that the Lipschitz constant of a maximum of two functions is at most the maximum of the Lipschitz constants of the two functions. Therefore, we can bound the Lipschitz constant of $\phi_k$ by taking the maximum of the Lipschitz constants of $\phi_k^{(s'=0)}$ and $\phi_k^{(s'=1)}$:
    \begin{itemize}
        \item For the $z$ coordinate:
        \[
            \max \{B(1, 1, 0, 1, 2 - \gamma), B(1, 1, 0, 1, 2 - \gamma)\} = B(1, 1, 0, 1, 2 - \gamma).
        \]
        \item For the $y$ coordinate:
        \[
            \max \{B(0, 2 + \alpha, 0, 2, 3 - \gamma), B(0, 2, 0, 1, 2 - \gamma)\} = B(0, 2 + \alpha, 0, 2, 3 - \gamma).
        \]
        \item For the $\sigma$ coordinate:
        \[
            \max \{B(0, 1, 0, 1, 1 - \gamma), B(0, 2, 0, 1, 1 - \gamma)\} = B(0, 2, 0, 1, 1 - \gamma).
        \]
    \end{itemize}
    Writing it explicitly, we have:
    \begin{align*}
        &\abs{\phi_k(z_{k,1}, y_{k,1}, \sigma_1) - \phi_k(z_{k,2}, y_{k,2}, \sigma_2)} \\
        &\leq B(1, 1, 0, 1, 2 - \gamma) \abs{z_{k,1} - z_{k,2}} + B(0, 2 + \alpha, 0, 2, 3 - \gamma) \abs{y_{k,1} - y_{k,2}} + B(0, 2, 0, 1, 1 - \gamma) \abs{\sigma_1 - \sigma_2}.
    \end{align*}

    \paragraph{Step 2: Bounding the difference $\abs{\phi_k(\mathbf{x}, \rho) - \phi_k(\mathbf{x}', \rho')}$}

    For profiles $\mathbf{x}$ and $\mathbf{x}'$ with $\norm{\mathbf{x} - \mathbf{x}'}_2 \leq \delta$, we have $\abs{x_k - x'_k} \leq \norm{\mathbf{x} - \mathbf{x}'}_2 \leq \delta$ for any component $k$. For the sum of qualities of other players, we have:
    \[
        \abs{y_k - y'_k} = \abs{\norm{\mathbf{x}_{-k}} - \norm{\mathbf{x}'_{-k}}} \leq \norm{\mathbf{x}_{-k} - \mathbf{x}'_{-k}} \leq \sqrt{n-1}\norm{\mathbf{x}_{-k} - \mathbf{x}'_{-k}}_2 \leq \sqrt{n}\delta.
    \]

    Using the Lipschitz constant derived above and the mean value theorem, we can bound the difference:
    \[
        \abs{\phi_k(\mathbf{x}, \rho) - \phi_k(\mathbf{x}', \rho')} \leq B(1, 1, 0, 1, 2 - \gamma) \delta + B(0, 2 + \alpha, 0, 2, 3 - \gamma) \sqrt{n}\delta + B(0, 2, 0, 1, 1 - \gamma) \delta.
    \]
    Notice we can write the expression $B(0, 2 + \alpha, 0, 2, 3 - \gamma) \sqrt{n}$ as $B(0, (2 + \alpha), 0.5, 2, 3 - \gamma)$, as the third parameter affects the exponent of $n$ in the definition of $B$. Thus, we can combine the terms to get:
    \[
        \abs{\phi_k(\mathbf{x}, \rho) - \phi_k(\mathbf{x}', \rho')} \leq \delta \cdot \left( B(1, 1, 0, 1, 2 - \gamma) + B(0, (2 + \alpha), 0.5, 2, 3 - \gamma) + B(0, 2, 0, 1, 1 - \gamma) \right).
    \]
\end{proof}

\begin{lemma}\label{lem:bounded_x_rho_difference_implies_bounded_revenue_difference}
    Fix some $\delta > 0$. Let two quality profiles $\mathbf{x}$ and $\mathbf{x}'$ such that $\norm{\mathbf{x} - \mathbf{x}'}_2 \leq \delta$, and two revenue sharing parameters $\rho, \rho'$ such that $\abs{\rho - \rho'} \leq \delta$. Then, the absolute difference in the platform's revenue is bounded by:
    \[
        \abs{U_P(\mathbf{x}, \mathbf{1}; f_{\rho}) - U_P(\mathbf{x}', \mathbf{1}; f_{\rho'})} \leq B(0, 1+\gamma, 0.5, 1, 1 - \gamma) \cdot \delta.
    \]
\end{lemma}

\begin{proof}[\proofof{lem:bounded_x_rho_difference_implies_bounded_revenue_difference}]
    We apply the triangle inequality to decompose the difference in utilities into two components:
    \begin{align*}
        \abs{U_P(\x, \mathbf{1}; f_{\rho}) - U_P(\x', \mathbf{1}; f_{\rho'})}
        &\leq \abs{U_P(\x, \mathbf{1}; f_{\rho}) - U_P(\x', \mathbf{1}; f_{\rho})}\\
        &+ \abs{U_P(\x', \mathbf{1}; f_{\rho}) - U_P(\x', \mathbf{1}; f_{\rho'})}.
    \end{align*}

    Letting $C = \frac{\mu\alpha}{1+\alpha}$, we can express this as:
    \begin{align*}
        \abs{U_P(\x, \mathbf{1}; f_{\rho}) - U_P(\x', \mathbf{1}; f_{\rho'})}
        &= \abs{C(1-\rho)\norm{\x}^\gamma - C(1-\rho)\norm{\x'}^\gamma}\\
        &+ \abs{C(1-\rho)\norm{\x'}^\gamma - C(1-\rho')\norm{\x'}^\gamma}\\
        &= C(1-\rho) \cdot \abs{\norm{\x}^\gamma - \norm{\x'}^\gamma}
        + C\norm{\x'}^\gamma \cdot \abs{\rho' - \rho}.
    \end{align*}

    To bound the difference of powers $\abs{\norm{\x}^\gamma - \norm{\x'}^\gamma}$, we use the Mean Value Theorem. Let $g(v) = v^\gamma$ for $v > 0$, where $\gamma \in [0,1]$. By the Mean Value Theorem, there exists a value $\xi$ between $\norm{\x}$ and $\norm{\x'}$ such that:
    \[
        \abs{\norm{\x}^\gamma - \norm{\x'}^\gamma} = \abs{g(\norm{\x}) - g(\norm{\x'})} = \abs{g'(\xi)} \cdot \abs{\norm{\x} - \norm{\x'}}.
    \]

    The derivative is $g'(v) = \gamma v^{\gamma-1}$. Since $\xi \in [X_{LB}, X_{UB}]$ and $\gamma-1 \leq 0$ for $\gamma \in [0,1]$, the function $v^{\gamma-1}$ is non-increasing, so:
    \[
        \abs{g'(\xi)} = \gamma \xi^{\gamma-1} \leq \gamma (X_{LB})^{\gamma-1}.
    \]

    Therefore,
    \[
        \abs{\norm{\x}^\gamma - \norm{\x'}^\gamma} \leq \gamma (X_{LB})^{\gamma-1} \cdot \abs{\norm{\x} - \norm{\x'}}.
    \]

    By the reverse triangle inequality and the relation between $\ell_1$ and $\ell_2$ norms in $\mathbb{R}^n$, we have:
    \[
        \abs{\norm{\x} - \norm{\x'}} \leq \norm{\x - \x'} \leq \sqrt{n} \cdot \norm{\x - \x'}_2 \leq \sqrt{n} \cdot \delta.
    \]

    Combining, the first term is bounded by:
    \[
        T_1 = C\abs{1-\rho} \cdot \abs{\norm{\x}^\gamma - \norm{\x'}^\gamma} \leq C \cdot \gamma (X_{LB})^{\gamma-1} \cdot \sqrt{n} \cdot \delta.
    \]

    For the second term, since $\norm{\x'} \leq X_{UB}$ and $\abs{\rho' - \rho} \leq \delta$, we have:
    \[
        T_2 = C\norm{\x'}^\gamma \cdot \abs{\rho' - \rho} \leq C \cdot (X_{UB})^\gamma \cdot \delta.
    \]

    Combining both bounds and factoring out $\delta$, we obtain:
    \begin{align*}
        \abs{U_P(\x, \mathbf{1}; f_{\rho}) - U_P(\x', \mathbf{1}; f_{\rho'})}
        &\leq T_1 + T_2 \\
        &\leq C \cdot \gamma (X_{LB})^{\gamma-1} \cdot \sqrt{n} \cdot \delta + C \cdot (X_{UB})^\gamma \cdot \delta\\
        &= C \cdot \left( \gamma (X_{LB})^{\gamma-1} \cdot \sqrt{n} + (X_{UB})^\gamma \right) \cdot \delta.
    \end{align*}

    Since $\frac{\mu\alpha}{1+\alpha} \leq \mu$ for any $\alpha > 0$, we have:
    \begin{align*}
        \abs{U_P(\x, \mathbf{1}; f_{\rho}) - U_P(\x', \mathbf{1}; f_{\rho'})}
        &\leq \mu \cdot \left( \gamma (X_{LB})^{\gamma-1} \sqrt{n} + (X_{UB})^\gamma \right) \cdot \delta\\
        &\leq \left( B(0, \gamma, 0.5, 0, 1 - \gamma) + B(0, 1, 0, \gamma, 0) \right) \cdot \delta\\
        &\leq B(0, 1+\gamma, 0.5, 1, 1 - \gamma) \cdot \delta.
    \end{align*}
    This completes the proof of \Cref{lem:bounded_x_rho_difference_implies_bounded_revenue_difference}.
\end{proof}

\subsection{Lipschitz Continuity Results Used in the Proof of \Cref{alg:main}}\label{app:lipshitz_results_for_main_algorithm}

\begin{lemma}\label{lem:transfer_of_approx_fse}
    Let $\rho, \rho' \in [0,1]$ such that $\abs{\rho - \rho'} \le \delta$, and $\mathbf{x}, \mathbf{x}'$ such that $\norm{\mathbf{x} - \mathbf{x}'}_2 \le \delta$ and one of those profiles is an ESE in either $G(\rho)$ or $G(\rho')$. Assuming $0.5 Y_{LB} > \sqrt{n}\delta$, if $\x$ is an $\eta$-FSE in $G(\rho)$, then $\x'$ is a $(\eta + B(1, (5 + \alpha) \cdot 2^{6 - \gamma}, 0.5, 3, 3 - \gamma) \delta)$-FSE in $G(\rho')$.
\end{lemma}

\begin{proof}[\proofof{lem:transfer_of_approx_fse}]
    Define $\bar \x \in \{\x, \x'\}$ to be the ESE profile among the two. Let $\phi_k(\mathbf{z}, \sigma)$ denote the maximum utility gain player $k$ can achieve by unilaterally deviating from the profile $(\mathbf{z}, \mathbf{1})$ when the platform's revenue sharing parameter is $\sigma$. To apply \Cref{lem:liphsit_of_deviation}, we need to establish bounds on other players' contributions in the profiles $\x$ and $\x'$.

    From \Cref{lem:ese_quality_lower_bound}, we know that $\norm{\bar \x_{-k}} \geq \sum_{i} x_i^{\min} - \max_{k} x_k^{\min} \eqqcolon Y_{LB}$ for any player $k$ and any $\rho \in [0,1]$. Additionally, we know from \Cref{lem:rational_quality_bound} that $\norm{\bar \x_{-k}} \leq \sum_{j=1}^{n} x_j^{\max} - \min_{k} x_k^{\max} \eqqcolon Y_{UB}$.

    We now derive lower bounds for $\norm{\x_{-k}}, \norm{\x'_{-k}}$ as well via the assumptions that $\norm{\x - \bar \x}_2, \norm{\x' - \bar \x}_2 \leq \delta$ and $0.5 Y_{LB} > \sqrt{n}\delta$. Specifically, for any player $k$:
    \[
        \abs{\norm{\bar \x_{-k}} - \norm{\x_{-k}}} \leq \norm{\bar \x_{-k} - \x_{-k}} \leq \sqrt{n-1}\norm{\bar \x_{-k} - \x_{-k}}_2 \leq \sqrt{n}\delta,
    \]
    and the same holds for $\x'$. Since $\norm{\bar \x_{-k}} \geq Y_{LB}$ and $0.5 Y_{LB} \geq \sqrt{n}\delta$, we have $\norm{\x_{-k}} \geq Y_{LB} - \sqrt{n}\delta > 0.5 Y_{LB} \eqqcolon \tilde Y_{LB} > 0$. Additionally, we have $\norm{\x_{-k}} \leq Y_{UB} + \sqrt{n}\delta < 2 Y_{UB} \eqqcolon \tilde Y_{UB}$. These bounds hold for $\x'$ as well. Notice that the validity of these bounds does not depend on which of the two profiles is the ESE.

    Now, we analyze the deviation gain of any player from the profile $\x'$. By the assumption that $\x$ is $\eta$-FSE in $G(\rho)$, we know that $\phi_k(\x, \rho) \leq \eta$ for all players $k$. Applying the triangle inequality, we have:
    \[
        \phi_k(\x', \rho') \leq \abs{\phi_k(\x, \rho) - \phi_k(\x', \rho')} + \phi_k(\x, \rho).
    \]
    Now, applying \Cref{lem:liphsit_of_deviation} with $\tilde Y_{LB}, \tilde Y_{UB}$ as the bounds for other players' contributions, we get:
    \[
        \abs{\phi_k(\x, \rho) - \phi_k(\x', \rho')} \leq \delta \cdot \left(\tilde B(1, 1, 0, 1, 2 - \gamma) + \tilde B(0, (2 + \alpha), 0.5, 2, 3 - \gamma) + \tilde B(0, 2, 0, 1, 1 - \gamma)\right),
    \]
    where $\tilde B$ is defined similarly to $B$ but with $\tilde Y_{LB}, \tilde Y_{UB}$ instead of $Y_{LB}, Y_{UB}$. Thus, we have:
    \[
        \phi_k(\x', \rho') \leq \delta \cdot \left(\tilde B(1, 1, 0, 1, 2 - \gamma) + \tilde B(0, (2 + \alpha), 0.5, 2, 3 - \gamma) + \tilde B(0, 2, 0, 1, 1 - \gamma)\right) + \eta.
    \]
    We now want to simplify the bracketed term. First, notice that
    \begin{align*}
        &\tilde B(1, 1, 0, 1, 2 - \gamma) + \tilde B(0, (2 + \alpha), 0.5, 2, 3 - \gamma) + \tilde B(0, 2, 0, 1, 1 - \gamma) \\
        &\leq \tilde B(1, 1, 0, 2, 3- \gamma) + \tilde B(0, (2 + \alpha), 0.5, 2, 3 - \gamma) + \tilde B(0, 2, 0, 3, 3-\gamma) \\
        &\leq \tilde B(1, (1+2 + \alpha + 2), 0.5, 3, 3 - \gamma) \\
        &= \tilde B(1, (5 + \alpha), 0.5, 3, 3 - \gamma) \\
        &\leq B(1, (5 + \alpha) \cdot 2^{3} \cdot 2^{3-\gamma}, 0.5, 3, 3 - \gamma) \\
        &= B(1, (5 + \alpha) \cdot 2^{6 - \gamma}, 0.5, 3, 3 - \gamma),
    \end{align*}
    where the second to last inequality follows from the definition of $B$ and the fact that $\tilde Y_{LB} = 0.5 Y_{LB}$ and $\tilde Y_{UB} = 2 Y_{UB}$. After this simplification, we get that:
    \[
        \phi_k(\x', \rho') \leq B(1, (5 + \alpha) \cdot 2^{6 - \gamma}, 0.5, 3, 3 - \gamma) \cdot \delta + \eta,
    \]
    which means that $\x'$ is a $(\eta + B(1, (5 + \alpha) \cdot 2^{6 - \gamma}, 0.5, 3, 3 - \gamma)\delta)$-FSE in $G(\rho')$, completing the proof of \Cref{lem:transfer_of_approx_fse}.
\end{proof}

\begin{lemma}\label{lem:bounded_revenue_difference}
    Let $\rho, \rho' \in [0,1]$ such that $\abs{\rho - \rho'} \le \delta$, and $\x, \x'$ such that $\norm{\x - \x'}_2 \le \delta$ and one of those profiles is an ESE in either $G(\rho)$ or $G(\rho')$. Assuming $0.5X_{LB} \geq \sqrt{n}\delta$, it holds that
    \[
        \abs{U_P(\x, \mathbf{1}; f_{\rho}) - U_P(\x', \mathbf{1}; f_{\rho'})} \leq B(0, 4(1+\gamma), 0.5, 1, 1 - \gamma) \delta.
    \]
\end{lemma}

\begin{proof}[\proofof{lem:bounded_revenue_difference}]
    Without loss of generality, we assume $\x$ is the unique ESE in $G(\rho)$. Following the same steps as in the proof of \Cref{lem:transfer_of_approx_fse}, we get that $0.5X_{LB} \leq \norm{\x} \leq 2X_{UB}$, as well as $\norm{\x'}$ bounds. Next, we apply \Cref{lem:bounded_x_rho_difference_implies_bounded_revenue_difference} with those bounds to get:
    \begin{align*}
        \abs{U_P(\x, \mathbf{1}; f_{\rho}) - U_P(\x', \mathbf{1}; f_{\rho'})} &\leq \tilde B(0, 1+\gamma, 0.5, 1, 1 - \gamma) \cdot \delta \\
        &\leq B(0, 4(1+\gamma), 0.5, 1, 1 - \gamma) \cdot \delta,
    \end{align*}
    which completes the proof of \Cref{lem:bounded_revenue_difference}.
\end{proof}

\begin{lemma}\label{lem:deviation_lipschitz_s_0}
    Let $f(x_i) = \frac{\mu (x_i + \norm{\mathbf{x}_{-i}})^{\gamma} x_i}{x_i + (1+\alpha)\norm{\mathbf{x}_{-i}}} - c_i(x_i)$ be the utility function for player $i$ under no sharing. Assume we have a profile $\mathbf{x}$ such that $\norm{\mathbf{x} - \xr{\rho}}_2 \leq \delta$, where $\xr{\rho}$ is the ESE in $G(\rho)$. Assuming that $0.5Y_{LB} > \delta \sqrt{n}$, $f(x_i)$ is Lipschitz continuous in $x_i$ with Lipschitz constant $B(1, 5 \cdot 2^{5-\gamma}, 0, 2, 3-\gamma)$.
\end{lemma}

\begin{proof}[\proofof{lem:deviation_lipschitz_s_0}]
    Following the logic in the proof of \Cref{lem:transfer_of_approx_fse}, we have that $\norm{\mathbf{x}_{-i}}$ is bounded as $0.5Y_{LB} < \norm{\mathbf{x}_{-i}} < 2Y_{UB}$. Let $y = \norm{\mathbf{x}_{-i}}$ for simplicity. We compute the derivative of $f(x_i)$ using the quotient rule:
    \begin{align*}
        f'(x_i) &= \mu \cdot \frac{[\gamma(x_i + y)^{\gamma-1}x_i + (x_i + y)^{\gamma}][x_i + (1+\alpha)y] - (x_i + y)^{\gamma}x_i }{[x_i + (1+\alpha)y]^2} - c_i'(x_i) \\
        &= \mu \cdot \frac{\gamma(x_i + y)^{\gamma-1}x_i[x_i + (1+\alpha)y] + (x_i + y)^{\gamma}[x_i + (1+\alpha)y] - (x_i + y)^{\gamma}x_i}{[x_i + (1+\alpha)y]^2} - c_i'(x_i).
    \end{align*}

    Using our established bounds $y \geq 0.5Y_{LB}$, $y \leq 2Y_{UB}$, and $x_i \leq x_i^{\max}$:
    \begin{align*}
        \abs{f'(x_i)} &\leq \mu \Bigg[\frac{ \gamma (0.5Y_{LB})^{\gamma-1} (x_i^{\max})^2 }{ [(1+\alpha)(0.5Y_{LB})]^2 } + \frac{ \gamma (0.5Y_{LB})^{\gamma-1} x_i^{\max} (1+\alpha)(2Y_{UB}) }{ [(1+\alpha)(0.5Y_{LB})]^2 } \\
        &+ \frac{ (x_i^{\max}+2Y_{UB})^{\gamma} x_i^{\max} }{ [(1+\alpha)(0.5Y_{LB})]^2 } + \frac{ (x_i^{\max} + 2Y_{UB})^{\gamma} (1+\alpha)(2Y_{UB}) }{ [(1+\alpha)(0.5Y_{LB})]^2 } \\
        &+ \frac{ (2Y_{UB})^{\gamma} x_i^{\max} }{ [(1+\alpha)(0.5Y_{LB})]^2 } \Bigg] + \abs{c_i'(x_i^{\max})} \\
        &\leq \tilde B(0, 1, 0, 2, 3-\gamma) + \tilde B(0, 1, 0, 2, 3-\gamma) + \tilde B(0, 1, 0, \gamma+1, 2) + \tilde B(0, 1, 0, \gamma+1, 2) \\
        &+ \tilde B(0, 1, 0, \gamma+1, 2) + c_i'(x_i^{\max}) \\
        &\leq \tilde B(1, 5, 0, 2, 3-\gamma) \\
        & \leq B(1, 5 \cdot 2^{5-\gamma}, 0, 2, 3-\gamma).
    \end{align*}
    This completes the proof of \Cref{lem:deviation_lipschitz_s_0}.
\end{proof}

\subsection{Strong Monotonicity and Lipschitz Continuity of ESE}\label{app:strongly_monotone_and_lipschitz_ese}

\begin{lemma}\label{lem:strongly_monotone}
    Assume that all the cost functions $c_i(x_i)$ are $\beta$-strongly convex. Then, for any $\rho \in [0,1]$, the enforced sharing game w.r.t. $\rho$ is strongly monotone.
\end{lemma}

\begin{proof}[\proofof{lem:strongly_monotone}]
    We extend the monotonicity result from \Cref{lem:unique-ese} to show that with $\beta$-strongly convex cost functions, the game becomes strongly monotone.

    Recall from the proof of \Cref{lem:unique-ese} that:
    \[
        \sum_{i=1}^{n} (v_i(\mathbf{x}') - v_i(\mathbf{x}))(x_i' - x_i) = \sum_{i=1}^{n} (A_i(\mathbf{x}') - A_i(\mathbf{x}))(x_i' - x_i) - \sum_{i=1}^{n} (c_i'(x_i') - c_i'(x_i))(x_i' - x_i).
    \]

    We already established that the first term is non-positive:
    \[
        \sum_{i=1}^{n} (A_i(\mathbf{x}') - A_i(\mathbf{x}))(x_i' - x_i) \leq 0.
    \]

    The key difference now is that with $\beta$-strongly convex cost functions, we have a stronger bound on the second term. Specifically, for $\beta$-strongly convex functions:
    \[
        (c_i'(x_i') - c_i'(x_i))(x_i' - x_i) \geq \beta(x_i' - x_i)^2.
    \]

    Summing over all players:
    \[
        \sum_{i=1}^{n} (c_i'(x_i') - c_i'(x_i))(x_i' - x_i) \geq \beta\sum_{i=1}^{n}(x_i' - x_i)^2 = \beta\|\mathbf{x}' - \mathbf{x}\|_2^2.
    \]

    Therefore:
    \[
        \sum_{i=1}^{n} (v_i(\mathbf{x}') - v_i(\mathbf{x}))(x_i' - x_i) \leq 0 - \beta\|\mathbf{x}' - \mathbf{x}\|_2^2 = -\beta\|\mathbf{x}' - \mathbf{x}\|_2^2.
    \]

    This establishes that the game is strongly monotone with parameter $\beta$.
\end{proof}

\begin{lemma}\label{lem:vector_ese_lipschitz}
    Assume that all the cost functions $c_i(x_i)$ are $\beta$-strongly convex. Then, $\xr{\cdot}$ is Lipschitz continuous with respect to the $\ell_2$-norm on the interval $[0,1]$, with Lipschitz constant at most $B(0, 1/\beta, 0.5, 1, 2-\gamma)$.
\end{lemma}

\begin{proof}[\proofof{lem:vector_ese_lipschitz}]
    The proof proceeds in three main steps: (1) characterizing the equilibrium conditions, (2) establishing differentiability through the Implicit Function Theorem, and (3) bounding the derivative to obtain the Lipschitz constant.

    \paragraph{Step 1: Characterizing the equilibrium} Denote $K(\rho) = \frac{\mu(1+\alpha\rho)}{1+\alpha}$. The utility for creator $i$ under enforced full sharing ($\mathbf{s}=\mathbf{1}$) is given by \Cref{lem:utility_full_sharing}:
    \[
        U_i(\mathbf{x}, \mathbf{1}; f_\rho) = K(\rho) \norm{\mathbf{x}}^{\gamma - 1} x_i - c_i(x_i),
    \]

    Let us define $g_i(\mathbf{x}, \rho) = \frac{\partial U_i}{\partial x_i} = K(\rho)\norm{\mathbf{x}}^{\gamma-2}\left[\gamma x_i + \norm{\mathbf{x}_{-i}}\right] - c_i'(x_i)$. Recall that the FOC for creator $i$ in the ESE is $g_i(\xr{\rho}, \rho) = 0$.

    \paragraph{Step 2: Continuous differentiability of $\xr{\cdot}$}
    To establish Lipschitz continuity, we first show that $\xr{\cdot}$ is continuously differentiable using the Implicit Function Theorem (IFT).

    We first verify that $\mathbf{g}(\mathbf{x}, \rho)$ is continuously differentiable in both $\mathbf{x}$ and $\rho$:
    \begin{itemize}
        \item $K(\rho) = \frac{\mu(1+\alpha\rho)}{1+\alpha}$ is clearly $C^1$ in $\rho$.
        \item By assumption, each $c_i(x_i)$ is at least twice differentiable, so $c_i'(x_i)$ is $C^1$.
        \item From \Cref{lem:dominance-positive-quality} and the arguments in the proof of \Cref{thm:fse-existence}, we know that $\mathcal{X}_i(\rho) > 0$ for all $i \in [n]$. Consequently, terms like $\norm{\x}^{\gamma-2}$ are well-defined and $C^1$ in a neighborhood of $\xr{\rho}$ and are continuously differentiable as a composition of $C^1$ functions.
    \end{itemize}

    Next, we need to verify that the Jacobian matrix $J_x(\mathbf{x}, \rho) = \frac{\partial \mathbf{g}}{\partial \mathbf{x}}$ is invertible at $\mathcal{X}(\rho)$. \Cref{lem:strongly_monotone} establishes that when the cost functions $c_i(x_i)$ are $\beta$-strongly convex, the enforced sharing game is strongly monotone. Specifically, for any distinct $\mathbf{x}, \mathbf{x}'$ in the strategy space, the gradient map $\mathbf{g}(\mathbf{x}, \rho)$ satisfies:
    \[
        (\mathbf{g}(\mathbf{x}', \rho) - \mathbf{g}(\mathbf{x}, \rho))^T(\mathbf{x}' - \mathbf{x}) \leq -\beta\|\mathbf{x}' - \mathbf{x}\|_2^2.
    \]
    This condition of strong monotonicity for a continuously differentiable map $\mathbf{g}$ implies that the symmetric part of its Jacobian, $J_S(\mathbf{x}, \rho) = \frac{1}{2}(J_x(\mathbf{x}, \rho) + J_x(\mathbf{x}, \rho)^T)$, is negative definite. This, in turn, implies that the Jacobian $J_x(\mathcal{X}(\rho), \rho)$ itself is invertible. To see why, assume for contradiction that $J_x(\mathcal{X}(\rho), \rho)$ is not invertible. Then, there exists a non-zero vector $\mathbf{u}$ such that $J_x(\mathcal{X}(\rho), \rho)\mathbf{u} = \mathbf{0}$. This would lead to:
    \[
        \mathbf{u}^T J_S(\mathcal{X}(\rho), \rho) \mathbf{u} = \frac{1}{2} (\mathbf{u}^T J_x(\mathcal{X}(\rho), \rho) \mathbf{u} + \mathbf{u}^T J_x(\mathcal{X}(\rho), \rho)^T \mathbf{u}).
    \]
    Since $J_x(\mathcal{X}(\rho), \rho)\mathbf{u} = \mathbf{0}$, we get that $\mathbf{u}^T J_S(\mathcal{X}(\rho), \rho) \mathbf{u} = 0$, which contradicts the fact that $J_S(\mathcal{X}(\rho), \rho)$ is negative definite. Therefore, $J_x(\mathcal{X}(\rho), \rho)$ must be invertible.

    By the Implicit Function Theorem, $\xr{\cdot}$ is a $C^1$ function of $\rho$ in $[0,1]$, with derivative:
    \[
        \frac{d\xr{\cdot}}{d\rho} = -\left(J_x(\xr{\rho}, \rho)\right)^{-1} \frac{\partial \mathbf{g}}{\partial \rho}(\xr{\rho}, \rho).
    \]

    \paragraph{Step 3: Bounding the derivative norm}
    We now bound $\norm{\frac{d\mathbf{x}^\star}{d\rho}}_2$ to establish the Lipschitz constant. By the operator norm inequality:
    \[
        \norm{\frac{d\xr{\cdot}}{d\rho}}_2 \leq \norm{J_x^{-1}}_2 \cdot \norm{\frac{\partial \mathbf{g}}{\partial \rho}}_2.
    \]

    First, we bound $\norm{J_x(\mathcal{X}(\rho), \rho)^{-1}}_2$. Using standard results from variational inequalities (e.g., \cite{scutari2012monotone}), we know that strong monotonicity of $\mathbf{g}(\mathbf{x}, \rho)$ implies that the symmetric part of the Jacobian $J_S(\mathcal{X}(\rho), \rho)$ is negative definite with the following property: For any non-zero vector $\mathbf{v}$,
    \[
        \mathbf{v}^T J_S(\mathcal{X}(\rho), \rho) \mathbf{v} \leq -\beta \norm{\mathbf{v}}_2^2.
    \]
    This means that the largest eigenvalue of the symmetric part $J_S(\mathcal{X}(\rho), \rho)$ satisfies $\lambda_{\max}(J_S) \leq -\beta$. Since $\beta > 0$, we have $\lambda_{\max}(J_S) < 0$.
    For a real invertible matrix $A$, if its symmetric part $S_A = (A+A^T)/2$ is negative definite and its largest eigenvalue is $\lambda_{\max}(S_A) = -m < 0$ (where $m>0$), it is a known result that the norm of the inverse is bounded by $\norm{A^{-1}}_2 \leq \frac{1}{m}$.
    Applying this result to $J_x(\mathcal{X}(\rho), \rho)$, with $m = \beta$, we obtain:
    \[ \norm{J_x(\mathcal{X}(\rho), \rho)^{-1}}_2 \leq \frac{1}{\beta}. \]

    Next, we bound $\norm{\frac{\partial \mathbf{g}}{\partial \rho}}_2$:
    For each $i \in [n]$, we have:
    \begin{align*}
        \frac{\partial g_i}{\partial \rho}(\mathbf{x}, \rho) &= \frac{\partial K(\rho)}{\partial \rho} \norm{\mathbf{x}}^{\gamma-2} \left(\gamma x_i + \norm{\mathbf{x}_{-i}}\right)\\
        &= \frac{\mu\alpha}{1+\alpha} \norm{\mathbf{x}}^{\gamma-2} \left(\gamma x_i + \norm{\mathbf{x}_{-i}}\right).
    \end{align*}

    \Cref{lem:ese_quality_lower_bound} ensures that $X_{LB} \leq \norm{\xr{\rho}}$. Additionally, \Cref{lem:rational_quality_bound} suggests that $\norm{\xr{\rho}} \leq X_{UB}$. Using the fact that $\gamma-2 < 0$, we have $\norm{\xr{\rho}}^{\gamma-2} \leq (X_{LB})^{\gamma-2}$. Also, $\gamma \mathcal{X}_i(\rho) + \norm{\mathcal{X}_{-i}(\rho)} \leq X_{UB}$ and $\frac{\mu\alpha}{1+\alpha} \leq \mu$. Therefore:
    \[
        \abs{\frac{\partial g_i}{\partial \rho}(\xr{\rho}, \rho)} \leq \mu (X_{LB})^{\gamma-2} X_{UB}.
    \]

    Using the Euclidean norm:
    \[
        \norm{\frac{\partial \mathbf{g}}{\partial \rho}(\xr{\rho}, \rho)}_2 = \sqrt{\sum_{i=1}^n \left(\frac{\partial g_i}{\partial \rho}\right)^2} \leq \sqrt{n} \cdot \mu (X_{LB})^{\gamma-2} X_{UB}.
    \]

    Combining our bounds:
    \[
        \norm{\frac{d\xr{\cdot}}{d\rho}}_2 \leq \frac{1}{\beta} \cdot \mu\sqrt{n} \cdot  (X_{LB})^{\gamma-2} X_{UB} = \frac{\mu\sqrt{n}}{\beta}  \frac{X_{UB}}{(X_{LB})^{2-\gamma}}.
    \]

    Since $\xr{\cdot}$ is continuously differentiable on the compact interval $[0,1]$, its derivative is bounded. By the Mean Value Theorem, for any $\rho, \rho' \in [0,1]$:
    \[
        \norm{\mathbf{x}^\star(\rho) - \mathbf{x}^\star(\rho')}_2 \leq L \cdot \abs{\rho - \rho'},
    \]
    where $L = \sup_{\rho \in [0,1]} \norm{\frac{d\mathbf{x}^\star}{d\rho}(\rho)}_2 \leq \frac{\mu \sqrt{n}}{\beta} \frac{X_{UB}}{(X_{LB})^{2-\gamma}} \leq B(0, 1/\beta, 0.5, 1, 2-\gamma)$. The proof of \Cref{lem:vector_ese_lipschitz} is now complete.
\end{proof}

\subsection{Proofs of \Cref{sec:platform-optimization}'s Main Results}

\ThmAlgMain*

\begin{proof}[\proofof{thm:alg_main}]
    The proof is structured in two main parts: First, establishing the approximation guarantees of \Cref{alg:main}, and second, analyzing its computational complexity under an efficient implementation.

    Before starting the proof, we address a special case of the choice of a large $\varepsilon$. Although we don't expect the algorithm to be executed with such a large $\varepsilon$, we want to show that the algorithm will prevail in this case as well. Notice that if $\varepsilon > \mu (1+X_{UB})$, then no matter which output the algorithm returns, it will satisfy the approximation guarantees:
    \begin{itemize}
        \item First, notice that the utilities of the players are bounded by $\mu (1+X_{UB})$. Thus, no player can have a deviation that is larger than $\mu (1+X_{UB})$, and the algorithm must return an $\varepsilon$-FSE.
        \item In the same line of reasoning, the platform revenue is also bounded by $\mu (1+X_{UB})$, thus, the algorithm will return a solution with revenue not smaller than $\text{OPT} - \varepsilon$.
    \end{itemize}

    Given these reasons, we limit our attention to the case when $\varepsilon \leq \mu X_{UB}$.

    \paragraph{Part 1: Approximation guarantees}
    First, let us specify the choice of the constant $\delta$ in the algorithm. We set $\delta = \frac{\varepsilon}{A}$, where $A = 4B(1, (5 + \alpha) \cdot 2^{6 - \gamma}, 0.5, 3, 3 - \gamma) \cdot (1+B(0, 1/\beta, 0.5, 1, 2-\gamma))$. The reasons for this choice will become clear in the next paragraphs.

    Let $\rho^\star \in \mathcal{F}$ be a feasible solution (induces an FSE) such that $U_P(\xr{\rho^\star}, \mathbf{1}; f_{\rho^\star}) \geq \text{OPT} - \frac{1}{2}\varepsilon$. Our objective is to show that \Cref{alg:main} returns a parameter $\hat{\rho}$ such that $(\xr{\hat \rho}, \mathbf{1})$ is an $\varepsilon$-FSE in $G(\hat \rho)$ and $U_P(\xr{\hat \rho}, \mathbf{1}; f_{\hat \rho}) \geq \text{OPT} - \varepsilon$.

    Let $\tilde{\rho} = \frac{\lfloor \rho^\star \cdot \lfloor 1/\delta \rfloor \rfloor}{\lfloor 1/\delta \rfloor}$ be the closest grid point less than or equal to $\rho^\star$ that is evaluated by the algorithm in Line~\ref{line:loop-start}. Note that $|\tilde{\rho} - \rho^\star| \leq \delta$ due to the grid resolution. Additionally, let $\tilde{\mathbf{x}}$ be some arbitrary quality profile that could have been found in Line~\ref{line:compute-ese} when evaluating $\tilde{\rho}$. Namely, $\norm{\tilde{\mathbf{x}} - \xr{\tilde \rho}}_2 \leq \delta$.

    We start from demonstrating that $\tilde{\mathbf{x}}$ will pass the $\eta$-FSE check in Line~\ref{line:check-fse} (recall that the algorithm is executed with $\eta = \frac{\varepsilon}{4}$). Recall that we demonstrated in \Cref{lem:vector_ese_lipschitz} that $\xr{\cdot}$ is Lipschitz continuous with a Lipschitz constant $L = B(0, 1/\beta, 0.5, 1, 2-\gamma)$. By the Lipschitz property and the triangle inequality, we have:

    \[
    \norm{\tilde{\mathbf{x}} - \xr{\rho^\star}}_2 \leq \norm{\tilde{\mathbf{x}} - \xr{\tilde \rho}}_2 + \norm{\xr{\tilde \rho} - \xr{\rho^\star}}_2 \leq \delta + L\delta = (1+L)\delta.
    \]

    Since $(\xr{\rho^\star}, \mathbf{1})$ is an FSE in $G(\rho^\star)$, it is also trivially $0$-FSE. We now want to use \Cref{lem:transfer_of_approx_fse} to show that $\tilde{\mathbf{x}}$ is an approximate FSE in $G(\tilde \rho)$ within the required tolerance. To do so, we need the inequality $0.5Y_{LB} > \sqrt{n}(1+L)\delta = \frac{\sqrt{n}(1+L)\varepsilon}{A}$ to hold. Remember that if $\varepsilon > \mu (1+X_{UB})$, then the algorithm will return an $\varepsilon$-FSE regardless of our analysis, so we can continue by assuming that $\varepsilon \leq \mu (1+X_{UB})$. As we constructed $A$ such that $A = 4B(1, (5 + \alpha) \cdot 2^{6 - \gamma}, 0.5, 3, 3 - \gamma) \cdot (1+L) \geq \frac{\sqrt{n} (1+L) \mu (1+X_{UB})}{0.5Y_{LB}}$, then we have $0.5Y_{LB} \geq \frac{\sqrt{n} \mu (1+X_{UB})}{A} \geq \frac{\sqrt{n}(1+L) \varepsilon}{A}$, which implies that the condition holds. Thus, employing \Cref{lem:transfer_of_approx_fse}, we can conclude that $\tilde{\mathbf{x}}$ is an $B(1, (5 + \alpha) \cdot 2^{6 - \gamma}, 0.5, 3, 3 - \gamma)(1+L)\delta$-FSE in $G(\tilde \rho)$. As we have $\delta = \frac{\varepsilon}{A}$ and we set $A = 4B(1, (5 + \alpha) \cdot 2^{6 - \gamma}, 0.5, 3, 3 - \gamma) \cdot (1+L)$, we have that $\tilde{\mathbf{x}}$ is an $\frac{\varepsilon}{4}$-FSE in $G(\tilde \rho)$. Therefore, $\tilde \x$ satisfies the condition in Line~\ref{line:check-fse}, and the algorithm will find some feasible output.

    Now, we show that the evaluated utility $U_P(\tilde{\mathbf{x}}, \mathbf{1}; f_{\tilde{\rho}})$ is close to $\text{OPT}$. For that, we once again use the validity of the condition $0.5Y_{LB} > \sqrt{n}(1+L)\delta$ (from the same reasoning as above), and employ \Cref{lem:bounded_revenue_difference} to get:
    \begin{align*}
        U_P(\tilde{\mathbf{x}}, \mathbf{1}; f_{\tilde{\rho}}) &\geq U_P(\xr{\rho^\star}, \mathbf{1}; f_{\rho^\star}) - B(0, 4(1+\gamma), 0.5, 1, 1 - \gamma)(1+L)\delta \\
        &\geq \text{OPT} - \frac{1}{2}\varepsilon - B(0, 4(1+\gamma), 0.5, 1, 1 - \gamma)(1+L)\delta.
    \end{align*}

    As we constructed $A$ such that $A \geq 4B(0, 4(1+\gamma), 0.5, 1, 1 - \gamma)(1+L)$, we have that
    \[
    U_P(\tilde{\mathbf{x}}, \mathbf{1}; f_{\tilde{\rho}}) \geq \text{OPT} - \frac{1}{2}\varepsilon - B(0, 4(1+\gamma), 0.5, 1, 1 - \gamma)(1+L)\frac{\varepsilon}{A} \geq \text{OPT} - \frac{1}{2}\varepsilon - \frac{1}{4}\varepsilon = \text{OPT} - \frac{3}{4}\varepsilon.
    \]
    As a result, the best profile that the algorithm will evaluate in Line~\ref{line:loop-end} must be with utility at least $\text{OPT} - \frac{3}{4}\varepsilon$.

    To wrap up the approximation guarantees, we need to show that the returned induced profile, denoted as $\xr{\hat \rho}$, is $\varepsilon$-FSE and that $U_P(\xr{\hat \rho}, \mathbf{1}; f_{\hat \rho}) \geq \text{OPT} - \varepsilon$. Notice that this is not immediate from the above analysis, as the corresponding evaluated profile, denoted as $\hat \x$, is not necessarily the induced profile from the output $\hat \rho$. To overcome this mismatch, we once again use \Cref{lem:transfer_of_approx_fse} and \Cref{lem:bounded_revenue_difference}.

    Starting from the $\varepsilon$-FSE guarantee, remember that $\norm{\hat \x - \xr{\hat \rho}}_2 \leq \delta$. Thus, \Cref{lem:transfer_of_approx_fse} implies that given that $\hat \x$ is an $\eta$-FSE with respect to $\hat \rho$, we have that $\xr{\hat \rho}$ is an $\eta + B(1, (5 + \alpha) \cdot 2^{6 - \gamma}, 0.5, 3, 3 - \gamma) \delta$-FSE with respect to $\hat \rho$. As $B(1, (5 + \alpha) \cdot 2^{6 - \gamma}, 0.5, 3, 3 - \gamma) \delta = B(1, (5 + \alpha) \cdot 2^{6 - \gamma}, 0.5, 3, 3 - \gamma) \frac{\varepsilon}{A} \leq \frac{1}{4}\varepsilon$, we have that $\xr{\hat \rho}$ is an $\frac{1}{2}\varepsilon$-FSE with respect to $\hat \rho$. Moving on to the utility guarantee, give that $U_P(\hat \x, \mathbf{1}; f_{\hat \rho}) \geq \text{OPT} - \frac{3}{4}\varepsilon$, \Cref{lem:bounded_revenue_difference} implies that $U_P(\xr{\hat \rho}, \mathbf{1}; f_{\hat \rho}) \geq U_P(\hat \x, \mathbf{1}; f_{\hat \rho}) - B(0, 4(1+\gamma), 0.5, 1, 1 - \gamma)\delta$. As we set $A$ such that $A \geq 4B(0, 4(1+\gamma), 0.5, 1, 1 - \gamma)(1+L)$, we have that $B(0, 4(1+\gamma), 0.5, 1, 1 - \gamma)\frac{\varepsilon}{A} \leq \frac{1}{4}\varepsilon$, and thus $U_P(\xr{\hat \rho}, \mathbf{1}; f_{\hat \rho}) \geq \text{OPT} - \varepsilon$.

    Therefore, any output of the algorithm satisfies the approximation guarantees, and we conclude that the algorithm is correct.

    \paragraph{Part 2: Computational complexity}
    We now analyze the computational complexity of an efficient implementation of \Cref{alg:main}. The algorithm performs $\lceil 1/\delta \rceil$ iterations, where in each iteration it computes an approximate ESE and checks if it constitutes an $\eta$-FSE.

    For computing the approximate ESE (Line~\ref{line:compute-ese}), we leverage \Cref{lem:strongly_monotone}, which states that under the assumption $\forall i \in [n], x_i \in [0, x_i^{\max}] : c''_i(x_i) \geq \beta$, the enforced sharing game is strongly monotone. For strongly monotone games, the Multi-Agent Mirror Descent (MAMD) algorithm with full gradient feedback converges at a rate of $O(1/T)$, as established in \cite{10.5555/3327345.3327469} (see \Cref{alg:mamd_ese} for a full specification of the algorithm). Therefore, to find a $\delta$-approximation of an ESE (w.r.t. the $l_2$ norm), it suffices to run MAMD for $T = O(1/\delta)$ iterations.

    Each iteration of MAMD takes $O(n)$ time to compute the gradients for all players. Therefore, the total time complexity of approximating the ESE is $O(n/\delta)$.

    For the $\eta$-FSE check (Line~\ref{line:check-fse}), we need to verify if any player can profitably deviate from the profile $(\mathbf{x}, \mathbf{1})$ by more than $\eta$. By \Cref{lem:s_choice_threshold}, for any fixed quality choice $x_i$, the optimal sharing level is either $s_i = 0$ or $s_i = 1$. Therefore, for each player $i$, we need to check two types of deviations: 1) Deviations where $s_i = 1$ (changing only quality); and 2) Deviations where $s_i = 0$ (changing quality and not sharing).

    From the proof of \Cref{lem:transfer_of_approx_fse}, we know that if $\x$ is $\delta$ close to the ESE, then the maximal deviation gain by only varying the quality is at most $\delta B(1, (5 + \alpha) \cdot 2^{6 - \gamma}, 0.5, 3, 3 - \gamma)$. To see why, think of a similar maximal gain from a deviations function that only allows for deviations in $x_i$ while keeping the sharing level $s_i=1$. As we chose $\delta$ such that $B(1, (5 + \alpha) \cdot 2^{6 - \gamma}, 0.5, 3, 3 - \gamma) \delta \leq \frac{1}{4}\varepsilon$, we can ignore this deviation type as it will not be larger than $\eta$. On the other hand, the second deviation type needs to be checked (recall that not every ESE is an FSE; stability w.r.t. sharing is not guaranteed). We check for such deviations in \Cref{alg:fse_stability_check}.

    \begin{algorithm}
    \caption{FSE Stability Check}
    \label{alg:fse_stability_check}
    \begin{algorithmic}[1]
        \Require $\mathbf{x}$, $\rho$, $\eta$
        \Ensure Boolean indicating FSE stability
        
        \For{$i = 1$ \textbf{to} $n$}
            \State Define $f_0(x'_i) = \frac{\mu (x'_i + \|\mathbf{x}_{-i}\|_1)^{\gamma} x'_i}{x'_i + (1+\alpha)\|\mathbf{x}_{-i}\|_1} - c_i(x'_i)$
            \State Compute $L_0$ using \Cref{lem:deviation_lipschitz_s_0}
            \State Find $\max_{x'_i \in [0, x_i^{\max}]} f_0(x'_i)$ using grid search with step size $\frac{\eta}{L_0}$
            \State Let $\phi_i \gets \max_{x'_i} f_0(x'_i) - U_i(\mathbf{x}, \mathbf{1}; f_{\rho})$
            \If{$\phi_i > 2\eta$}
                \State \Return \textbf{False}
            \EndIf
        \EndFor
        
        \State \Return \textbf{True}
    \end{algorithmic}
    \end{algorithm}

    For each player, we perform a grid search over the interval $[0, x_i^{\max}]$ to find the best deviation when switching to withholding all data. By \Cref{lem:deviation_lipschitz_s_0}, the utility function for $s_i = 0$ while the rest of the creators fully share is Lipschitz continuous with constant $L_0 = B(1, 5 \cdot 2^{5-\gamma}, 0, 2, 3-\gamma)$. Therefore, we can perform a grid search with step size $\frac{\eta}{L_0}$ to find an approximate maximum of the deviation utility function. This ensures that the discretization error is at most $\eta$.

    With step size $\frac{\eta}{L_0}$ for grid search, the discretization error is at most $\eta$. This means that:
    \begin{itemize}
        \item Solutions that are $\eta$-FSE will be accepted, as the maximal deviation will be at most $\eta$ and the discretization error will be at most $\eta$.
        \item No solution that has a larger deviation than $3\eta$ will be accepted, as even with the discretization error reducing the score, the found deviation will still be larger than $2\eta$.
        \item Solutions that are between $\eta$ and $3\eta$ might be accepted. This does not bother us, as one can use the analysis we did before to show that solutions that are $3\eta$-FSE will induce a solution that is at most $4\eta$-FSE, which is $\varepsilon$-FSE, and thus we are okay with accepting them, although they do not exactly stand the $\eta$-FSE test.
    \end{itemize}

    Then, the correctness of the algorithm remains intact, although we used this approximation trick. We have previously analyzed $\tilde \rho$, which we showed to be $\eta$-FSE and therefore will pass this test. This solution is also very close to optimal. Additionally, as we just explained, no solution that induces more than $\varepsilon$-FSE will be returned.

    Wrapping up the complexity analysis; the time complexity of the $\eta$-FSE check for all players is $O\left(\frac{n \cdot L_0 \cdot x^{\max}}{\eta}\right)$. Combined with the complexity of finding the approximate ESE, the total time complexity per iteration of the main algorithm is:
    \[
    O\left(\frac{n}{\delta} + \frac{n \cdot L_0 \cdot x^{\max}}{\eta}\right) = O\left(\frac{n \cdot (1 + L_0 \cdot x^{\max})}{\delta}\right).
    \]

    With $\lceil 1/\delta \rceil$ iterations in total, the overall time complexity is:
    \[
    O\left(\frac{n \cdot (1 + L_0 \cdot x^{\max})}{\delta^2}\right).
    \]

    Therefore, by taking $B = n \cdot (1 + B(1, 5 \cdot 2^{5-\gamma}, 0, 2, 3-\gamma) \cdot x^{\max})$, we can express the time complexity as $O\left(\frac{B}{\delta^2}\right) = O\left(\frac{A^2 B}{\varepsilon^2}\right)$, where $A$ is the constant we defined earlier. This completes the proof of the computational complexity and ends the proof of \Cref{thm:alg_main}.
\end{proof}

\PropPowerCostScaling*

\begin{proof}[\proofof{prop:power_cost_scaling}]
    We prove the proposition by using the tight bounds on the ESE qualities for power cost functions established in \Cref{lem:bound_on_total_quality_in_ese_power_costs}. Notice that this does not alter any of the previously mentioned results, as it still correctly bounds the ESE qualities from above and below.

    \paragraph{Step 1: Constants simplification} Notice that under the assumption that $a_{\min}, a_{\max} = O(1)$, we have $x^{\max}, M_{c'} = O(1)$. Additionally, from \Cref{lem:bound_on_total_quality_in_ese_power_costs}, we have that $Y_{LB}, Y_{UB} = \Theta\left(n^{\frac{\theta-1}{\theta-\gamma}}\right)$. Therefore, we can simplify the function $B(\cdot)$ as follows:
    \begin{equation}\label{eq:prop_power_cost_scaling_proof_simplified_B}
        B(q, u, v, w, t) = q \cdot M_{c'} + \frac{u n^v \left(1 + x^{\max} + Y_{UB}\right)^w}{(Y_{LB})^t} = O\left(n^{\max\{0, v + \frac{\theta - 1}{\theta - \gamma} \cdot (w - t)\}}\right).
    \end{equation}

    \paragraph{Step 2: Analyzing the constants in \Cref{thm:alg_main}} We now analyze the constants that appear in the proof of \Cref{thm:alg_main}. Starting with the constant $A$, we have:
    \begin{align*}
        A &= 4B(1, (5+\alpha) \cdot 2^{6 - \gamma}, 0.5, 3, 3 - \gamma) \cdot (1+B(0, 1/\beta, 0.5, 1, 2-\gamma)) \\
        &= O\left(n^{\max\{0, 0.5 + \frac{\theta - 1}{\theta - \gamma} \cdot (3 - (3 - \gamma))\}} \cdot n^{\max\{0, 0.5 + \frac{\theta - 1}{\theta - \gamma} \cdot (1 - (2 - \gamma))\}}\right) \\
        &= O\left(n^{0.5 + \frac{\theta - 1}{\theta - \gamma} \cdot \gamma + \max\{0, 0.5 + \frac{\theta - 1}{\theta - \gamma} \cdot (\gamma - 1)\}}\right).
    \end{align*}
    In our regime of interest, where $\gamma \in [0,1]$ and $\theta \in (1,2]$, the power of $n$ in the above expression is maximized at $\gamma = 1$, and we have
    \[
        0.5 + \frac{\theta - 1}{\theta - \gamma} \cdot \gamma + \max\{0, 0.5 + \frac{\theta - 1}{\theta - \gamma} \cdot (\gamma - 1)\} \leq 2.
    \]
    Therefore, we have $A = O\left(n^{2}\right)$. Moving on to the constant $B$ from the complexity analysis, we have:
    \begin{align*}
        B &= n \cdot (1 + B(1, 5 \cdot 2^{5-\gamma}, 0, 2, 3-\gamma)\cdot x^{\max}) \\
        &= O\left(n \cdot n^{\max\{0, 0 + \frac{\theta - 1}{\theta - \gamma} \cdot (\gamma - 1)\}}\right) = O(n).
    \end{align*}
    To conclude, we have $A = O\left(n^{2}\right)$ and $B = O(n)$, which completes the proof of \Cref{prop:power_cost_scaling}.
\end{proof}

\PropApproxFixedRho*

\begin{proof}[\proofof{prop:approx_fixed_rho}]
    We begin by invoking a result adapted from \citet{10.5555/3692070.3694417}. We use a reduction from the enforced sharing game to the exclusive competition model analyzed in \citet{10.5555/3692070.3694417}. Under enforced sharing, the utility of player $i$ is given by \Cref{lem:utility_full_sharing}. It turns out that the 1D exclusive competition in \citet{10.5555/3692070.3694417} is equivalent to our enforced sharing game with the following parameter substitution: $\mu'=\mu(1+\alpha\rho)$, $\tilde{\gamma}=1-\gamma$, $\tilde{\beta}=1$ and $\tilde \alpha = \alpha$.
    
    \citet{10.5555/3692070.3694417} derived asymptotic bounds on the total quality in the PNE of the exclusive competition, which directly translate to bounds for our ESE concept. Specifically, we adapt Theorem 3 from \citet{10.5555/3692070.3694417} to our use case. We stress again that our main results in this paper are novel and require new ideas and techniques beyond those covered in \citet{10.5555/3692070.3694417}. More precisely, our two main points of interest were the sharing decisions of the creators and the allocation rule of the platform, which are novel aspects of our paper. The adapted result is as follows:
    \begin{quote}
        \textbf{Adapted Theorem.} Assume power cost functions $c_i(x_i) = a_i x_i^\theta$ for all $i$. Fix any $\rho \in [0,1]$ and let $\mu' = \mu(1+\alpha\rho)$. For any sufficiently large $n$ and $\mu$, $\norm{\xrho}$ satisfies:
        \[
             \frac{C_\theta}{2\alpha + 2} < \frac{(\norm{\xrho})^{\theta-\gamma}}{\mu' \cdot \norm{\mathbf{a}^{-1}}_{\frac{1}{\theta - 1}}} < C_\theta,
        \]
        where $C_\theta$ is a constant depending only on $\theta$, $\mathbf{a}^{-1} = (a_1^{-1}, \dots, a_n^{-1})$, and $\norm{\mathbf{a}^{-1}}_{\frac{1}{\theta - 1}} = \left( \sum_{i=1}^n (a_i^{-1})^{\frac{1}{\theta - 1}} \right)^{\theta-1}$.
    \end{quote}

    We now proceed to derive the approximation. We rely on the key relationships between $\norm{\xr{\rho}}$ and its lower bound $X_{LB}(\rho)$ established in the theorem above, which guarantees that for any $\rho \in [0,1]$:
    \[
        X_{LB}(\rho) \leq \norm{\xr{\rho}} \leq R \cdot X_{LB}(\rho),
    \]
    where $R = (2\alpha+2)^{1/(\theta-\gamma)}$ and $X_{LB}(\rho) = \left(\frac{C_\theta \mu (1+\alpha\rho) \norm{\mathbf{a}^{-1}}_{\frac{1}{\theta - 1}}}{2\alpha+2}\right)^{\frac{1}{\theta - \gamma}}$. Since the platform's revenue is $U_p(\xr{\rho}, \mathbf{1} ; f_{\rho}) = \frac{\mu\alpha}{1+\alpha}(1-\rho)(\xrho)^\gamma$, these quality bounds imply corresponding revenue bounds:
    \[
        V_{LB}(\rho) \leq U_p(\xr{\rho}, \mathbf{1} ; f_{\rho}) \leq R^\gamma \cdot V_{LB}(\rho),
    \]
    where $V_{LB}(\rho) = \frac{\mu\alpha}{1+\alpha}(1-\rho)(X_{LB}(\rho))^\gamma$.
    
    Next, we derive an approximation by finding the $\rho$ value that maximizes $V_{LB}(\rho)$. Substituting and simplifying:
    \begin{align*}
        V_{LB}(\rho) &= \frac{\mu\alpha}{1+\alpha}(1-\rho) \left[ \left( \frac{C_\theta \mu (1+\alpha\rho) \norm{\mathbf{a}^{-1}}_{\frac{1}{\theta - 1}}}{2\alpha+2} \right)^{\frac{1}{\theta-\gamma}} \right]^\gamma \\
        &= K \cdot (1-\rho)(1+\alpha\rho)^p,
    \end{align*}
    where $K$ is a positive constant and $p = \frac{\gamma}{\theta-\gamma}$. Taking the derivative of $g(\rho) = (1-\rho)(1+\alpha\rho)^p$:
    \begin{align*}
        g'(\rho) &= -(1+\alpha\rho)^p + (1-\rho)p\alpha(1+\alpha\rho)^{p-1} \\
        &= (1+\alpha\rho)^{p-1}[(p\alpha - 1) - \alpha\rho(1+p)].
    \end{align*}
    
    Setting $g'(\rho)=0$ and solving for $\rho$:
    \[
        \rho_{crit} = \frac{p\alpha - 1}{\alpha(1+p)} = \frac{\alpha\gamma +\gamma - \theta}{\alpha \theta}.
    \]
    
    Under our assumption that $(\alpha+1)\gamma > \theta$, we have $\rho_{crit} \in (0,1)$, and the second derivative confirms this critical point maximizes $V_{LB}(\rho)$. We set $\rho_{approx} = \rho_{crit}$.
    
    For sufficiently large $n$ and $\mu$, \Cref{prop:asymptotic_stability_interval} guarantees that $\rho_{approx}$ induces an FSE. This is crucial because it ensures our approximation is feasible within the constraints of the platform's optimization problem.
    
    Now we establish the approximation bound. Denote by $\rho^\star$ an allocation parameter that achieves $\text{OPT}$ (or a very close approximation). We leverage that $\rho_{approx}$ maximizes $V_{LB}(\rho)$, so $V_{LB}(\rho_{approx}) \geq V_{LB}(\rho^\star)$. From our revenue bounds above, $V_{LB}(\rho^\star) > \text{OPT}/R^\gamma$. Combining these:
    \[
        V_{LB}(\rho_{approx}) \geq V_{LB}(\rho^\star) \geq \frac{\text{OPT}}{R^\gamma}.
    \]
    
    Since $U_p(\xr{\rho_{approx}}, \mathbf{1} ; f_{\rho_{approx}}) > V_{LB}(\rho_{approx})$, we have:
    \[
        U_p(\xr{\rho_{approx}}, \mathbf{1} ; f_{\rho_{approx}}) \geq \frac{\text{OPT}}{R^\gamma} = \frac{\text{OPT}}{(2\alpha+2)^{\frac{\gamma}{\theta-\gamma}}}.
    \]
    
    Therefore, $U_p(\xr{\rho_{approx}}, \mathbf{1} ; f_{\rho_{approx}}) \geq (2\alpha+2)^{-\frac{\gamma}{\theta-\gamma}} \cdot \text{OPT}$, establishing that our fixed-parameter approach achieves the claimed approximation ratio.
\end{proof}

%% file: appendix/extensions_proofs.tex
\section{Omitted Proofs from \Cref{sec:extensions}}\label{sec:extensions_proofs}

In this section, we provide the formal proofs and derivations for the extensions discussed in \Cref{sec:extensions}.

\subsection{Proofs for Multiple Topics}\label{subsec:proofs_multiple_topics}

In this extension, each creator $i$ chooses a quality vector $\x_i = (x_{i,1}, \ldots, x_{i,K})$ and a sharing vector $\s_i = (s_{i,1}, \ldots, s_{i,K})$. The utility function is separable in revenue across topics but coupled in costs:
\[
    U_i(\x, \s; f_\rho) = \sum_{k=1}^K R_{i,k}(\x, \s; f_\rho) - c_i(\x_i).
\]

\begin{proposition}[Extension of \Cref{lem:s_choice_threshold} to Multiple Topics]
    Fix a quality profile $\x$ and sharing decisions for all topics other than $k$, and for all creators other than $i$. The optimal sharing level $s_{i,k}$ for creator $i$ in topic $k$ is given by the threshold policy:
    \[
        s_{i,k} =
        \begin{cases}
            1 &    \rho > \tau_{i,k} \\
            [0,1] &    \rho  = \tau_{i,k} \\
            0 &    \rho < \tau_{i,k} \\
        \end{cases},
        \quad \text{where } \tau_{i,k} = \frac{x_{i,k}}{x_{i,k} + \sum_{j \neq i} x_{j,k} + \alpha \sum_{j \neq i} x_{j,k} s_{j,k}}.
    \]
\end{proposition}

\begin{proof}
    The utility function $U_i$ separates into topic-specific revenue terms. Crucially, the cost function $c_i(\x_i)$ depends only on quality $\x_i$ and is independent of the sharing decision $\s_i$. Furthermore, the revenue $R_{i,m}$ for any topic $m \neq k$ is independent of $s_{i,k}$. Therefore, maximizing $U_i$ with respect to $s_{i,k}$ is equivalent to maximizing the revenue term $R_{i,k}$ alone. The partial derivative is:
    \[
        \frac{\partial U_i}{\partial s_{i,k}} = \frac{\partial R_{i,k}}{\partial s_{i,k}}.
    \]
    The functional form of $R_{i,k}$ is identical to the single-topic revenue function analyzed in \Cref{lem:s_choice_threshold}, substituting the scalar quality $x_i$ with the topic-specific quality $x_{i,k}$. Following the derivation in the proof of \Cref{lem:s_choice_threshold} (specifically the derivative of Eq. \ref{eq:utility_universal_no_sharing_wrt_si}), we obtain the same condition: $\frac{\partial R_{i,k}}{\partial s_{i,k}} > 0$ if and only if $\rho > \tau_{i,k}$.
\end{proof}

\begin{proposition}[Extension of \Cref{lem:unique-ese} to Multiple Topics]
    For any fixed $\rho$, the subgame $G(\rho)$ under enforced sharing ($\s = \mathbf{1}$) admits a unique ESE.
\end{proposition}

\begin{proof}
    We show that the game satisfies the Diagonally Strictly Concave (DSC) condition (or equivalently, that the gradient operator of the game is strictly monotone). The combined utility gradient vector for all players can be written as $\mathbf{g}(\x) = \nabla_{\x} \mathbf{R}(\x) - \nabla_{\x} \mathbf{C}(\x)$, where $\mathbf{R}$ represents the vector of revenues and $\mathbf{C}$ represents the vector of costs.
    
    First, consider the revenue component. Since total revenue is the sum of independent revenues from topics $k=1 \dots K$, the Jacobian of the revenue vector is block-diagonal, where each block corresponds to a single topic $k$. As established in the proof of \Cref{lem:unique-ese}, the single-topic revenue game is monotone (i.e., satisfies the negative definite property for the symmetric part of the Jacobian). The direct sum of monotone mappings is itself monotone. Thus, $\sum_{i=1}^n (\nabla_{\x_i} R_i(\x') - \nabla_{\x_i} R_i(\x))^T (\x_i' - \x_i) \leq 0$.
    
    Second, consider the cost component. We assume $c_i(\x_i)$ is strictly convex. By definition of strict convexity, its gradient $\nabla c_i$ is strictly monotone:
    \[
        (\nabla c_i(\x_i') - \nabla c_i(\x_i))^T (\x_i' - \x_i) > 0 \quad \text{for } \x_i' \neq \x_i.
    \]
    Even if $c_i$ is not separable across topics (i.e., cross-derivatives $\frac{\partial^2 c_i}{\partial x_{i,k} \partial x_{i,m}} \neq 0$), strict convexity ensures the Hessian is positive definite.
    
    Combining these, the game gradient $\mathbf{g}(\x)$ is strictly monotone. Along with the convexity and compactness of the strategy space (bounded by rational quality limits), this guarantees a unique equilibrium via Rosen's theorem~\cite{rosen1965existence}.
\end{proof}

\begin{theorem}[Extension of \Cref{thm:fse-existence} to Multiple Topics]
    There exist topic-specific bounds $x_{i,k}^{\max}$ such that if $\rho > \max_{i,k} \frac{x_{i,k}^{\max}}{x_{i,k}^{\max} + (1+\alpha) \norm{\mathcal{X}_{-i, k}(\rho)}}$, the unique ESE is also the unique FSE.
\end{theorem}

\begin{proof}[Proof Sketch]
    We generalize the concept of "rational quality bound" from \Cref{lem:rational_quality_bound}. Since $c_i(\x_i)$ is strictly convex and superlinear, the set of strategies yielding non-negative utility is bounded. Let $S_i$ be the set of rational strategies for player $i$. We define $x_{i,k}^{\max} = \sup \{ x_{i,k} \mid \x_i \in S_i \}$. Because $\tau_{i,k}$ is increasing in $x_{i,k}$ (analogous to the single-topic case), the condition on $\rho$ ensures that for every rational quality profile $\x$, $\rho > \tau_{i,k}(x_{i,k})$. This makes $s_{i,k}=1$ a strictly dominant strategy for every topic $k$ and every creator $i$. Since full sharing is dominant, the unique equilibrium must be the unique ESE of the enforced sharing game.
\end{proof}

\subsection{Proofs for Prior GenAI Data}\label{subsec:proofs_prior_data}

In this extension, we introduce a non-strategic creator (indexed by 0) who provides a fixed body of content $x_0 > 0$ that is always shared ($s_0 = 1$). This modifies the GenAI quality function to $Q_{\text{AI}}(\x, \s) = \alpha \left( x_0 + \sum_{j=1}^n x_j s_j \right)$.

\begin{proposition}[Extension of \Cref{lem:s_choice_threshold} with Prior Data]
    Fix a creator $i$ with quality $x_i > 0$ and a profile $(\x_{-i}, \s_{-i})$. The optimal sharing level $s_i$ is determined by the threshold:
    \[
        \tau_i = \frac{x_i}{x_i + \norm{\x_{-i}} + \alpha \x^{\top}_{-i} \s_{-i} + (1+\alpha)x_0}.
    \]
\end{proposition}

\begin{proof}
    Since $x_0$ is fixed in the eyes of each creator $i$, it can be thought of as an additional creator who chooses $s = 1$. Thus, the threshold $\tau_i$ is adapted directly to the stated form.
    
\end{proof}

\begin{proposition}[Robustness of \Cref{lem:unique-ese}]
    The introduction of fixed prior data $x_0$ does not violate the conditions for the uniqueness of the ESE.
\end{proposition}

\begin{proof}
    In the enforced sharing game ($\s=\mathbf{1}$), the utility function becomes:
    \[
        U_i(\x) = \frac{\mu(1+\alpha\rho)}{1+\alpha} (\xtot + x_0)^{\gamma-1} x_i - c_i(x_i).
    \]

    Let $X_{tot} = \xtot + x_0$. The revenue term is $R_i(\x) = C \cdot (X_{tot})^{\gamma-1} x_i$.
    The second derivative of revenue with respect to $x_i$ is:
    \[
        \frac{\partial^2 R_i}{\partial x_i^2} = C (\gamma-1) (X_{tot})^{\gamma-3} [ \gamma x_i + 2(X_{tot} - x_i) ].
    \]
    Since $X_{tot} > x_i$ (as $x_0 > 0$ and other creators exist), the term in the brackets is positive. With $\gamma \le 1$, the second derivative remains negative, preserving concavity.
    
    The logic for the monotonicity of the gradient vector (Step 4 in \Cref{lem:unique-ese}) relies on the structure of the quadratic form $\mathcal{Q}_\gamma$. Adding a positive constant $x_0$ to the total quality term $X_{tot}$ does not alter the sign of the quadratic form, as $x_0$ effectively acts as a shift in the operating point $\hat{\x}$. Since $\mathcal{Q}_\gamma \ge 0$ holds for any non-negative vector, shifting the vector by a positive constant preserves the inequality. Thus, the game remains monotone, and the ESE remains unique.
\end{proof}

\subsection{Proofs for Alternative Optimization Objectives}\label{subsec:proofs_alternative_objectives}

In this subsection, the platform aims to maximize
\[
    V(\rho)=f(\xr{\rho},\rho)
\]
over $\rho\in[0,1]$ subject to the FSE constraint. We assume $f$ is Lipschitz on $\mathcal{D}\times[0,1]$, where $\mathcal{D}$ contains all profiles examined by the algorithm, with constants $L_f$ (w.r.t.\ $\ell_2$ in $\x$) and $L_\rho$ (w.r.t.\ $|\cdot|$ in $\rho$). Then for any $\norm{\x-\x'}_2\le\delta$ and $\abs{\rho-\rho'} \leq \delta$,
\[
    |f(\x,\rho)-f(\x',\rho')|\le (L_f+L_\rho)\delta.
\]

Replacing \Cref{lem:bounded_revenue_difference} with this bound, the proof of \Cref{thm:alg_main} extends directly. \Cref{lem:transfer_of_approx_fse} remains unchanged, while the objective discretization error becomes at most $(L_f+L_\rho)(1+L_{\text{ESE}})\delta$, where $L_{\text{ESE}}=B(0,1/\beta,0.5,1,2-\gamma)$ is the Lipschitz constant of $\xr{\cdot}$ from \Cref{lem:vector_ese_lipschitz}. Consequently, choosing $\delta=\varepsilon/A'$ with
\[
    A'=\max\Big\{A_{\text{stability}},\;4(L_f+L_\rho)(1+L_{\text{ESE}})\Big\}
\]
yields the same $\varepsilon$-approximation and runtime guarantees as in \Cref{thm:alg_main}, with constants depending on $L_f$ and $L_\rho$.

\subsection{Proofs for GenAI-Dependent Traffic}\label{subsec:proofs_genai_traffic}

In this extension, the user traffic $T(\x, \s)$ depends on the total quality of both human and AI content: $T(\x, \s) = \mu \left( \xtot + Q_{\text{AI}}(\x, \s) \right)^\gamma$.

\begin{proposition}[Threshold Derivation for GenAI-Dependent Traffic]
    Fix an arbitrary creator $i$ with a content $x_i >0$, and a profile $(\mathbf x_{-i}, \s_{-i})$. Further, define $\tau_i$ as 
    \[
        \tau_i = \frac{1}{\alpha} \left[ \left( \frac{\xtot + \alpha \left(x_i + \sum_{j\neq i}^n x_j s_j\right)}{\xtot + \alpha \sum_{j \neq i} x_j s_j} \right)^{1-\gamma} - 1 \right].
    \]
    The optimal sharing level $s$ of creator $i$ w.r.t. $(x_i, \mathbf x_{-i}, \mathbf s_{-i})$, i.e., $s\in \argmax_{s\in [0,1]} U_i ((x_i,\mathbf x_{-i}),(s, \s_{-i}); f_{\rho}) $ is given by
    \begin{equation}
        s^\star =
        \begin{cases}
            1 &    \rho > \tau_i \\
            \textnormal{any number from } [0,1] &    \rho  = \tau_i \\
            0 &    \rho < \tau_i \\
        \end{cases}.
    \end{equation}
\end{proposition}

\begin{proof}
    Let $D(s_i) = \xtot + Q_{\text{AI}}(\x, (s_i, \s_{-i}))$. Note that $D(s_i) = A + \alpha x_i s_i$, where $A = \xtot + \alpha \sum_{j \neq i} x_j s_j$ is independent of $s_i$.
    
    The utility of creator $i$ as a function of the continuous variable $s_i \in [0,1]$ is:
    \[
        U_i(s_i) = \mu (D(s_i))^\gamma \cdot \frac{x_i (1 + \alpha \rho s_i)}{D(s_i)} - c_i(x_i) = \mu x_i (D(s_i))^{\gamma-1} (1 + \alpha \rho s_i) - c_i(x_i).
    \]
    Differentiating $U_i$ with respect to $s_i$:
    \begin{align*}
        \frac{\partial U_i}{\partial s_i} &= \mu x_i \left[ (\gamma-1)(D(s_i))^{\gamma-2} (\alpha x_i) (1 + \alpha \rho s_i) + (D(s_i))^{\gamma-1} (\alpha \rho) \right] \\
        &= \mu x_i \alpha (D(s_i))^{\gamma-2} \left[ (\gamma-1) x_i (1 + \alpha \rho s_i) + \rho D(s_i) \right].
    \end{align*}
    Substituting $D(s_i) = A + \alpha x_i s_i$ into the bracketed term:
    \begin{align*}
        [\dots] &= (\gamma-1)x_i + (\gamma-1)\alpha \rho x_i s_i + \rho A + \rho \alpha x_i s_i \\
        &= \rho A - (1-\gamma)x_i + \gamma \alpha \rho x_i s_i.
    \end{align*}
    Since $\gamma, \alpha, \rho, x_i > 0$, the term $\gamma \alpha \rho x_i s_i$ is strictly increasing in $s_i$. This implies that $\frac{\partial U_i}{\partial s_i}$ is strictly increasing with respect to $s_i$. Consequently, the maximum of $U_i(s_i)$ on $[0,1]$ must occur at one of the boundaries, $s_i=0$ or $s_i=1$.
    
    Creator $i$ chooses $s_i=1$ if and only if $U_i(1) > U_i(0)$. Using the utility expression:
    \begin{align*}
        \mu x_i (D(1))^{\gamma-1} (1 + \alpha \rho) &> \mu x_i (D(0))^{\gamma-1} \\
        1 + \alpha \rho &> \left( \frac{D(0)}{D(1)} \right)^{\gamma-1} = \left( \frac{D(1)}{D(0)} \right)^{1-\gamma}.
    \end{align*}
    Solving for $\rho$, we obtain the threshold condition from the proposition statement:
    \[
        \rho > \frac{1}{\alpha} \left[ \left( \frac{D(1)}{D(0)} \right)^{1-\gamma} - 1 \right] = \tau_i.
    \]
\end{proof}

\begin{proposition}[ESE Uniqueness under GenAI-Dependent Traffic]
    For any $\rho \in [0,1]$, the subgame $G(\rho)$ under enforced sharing ($\s=\mathbf{1}$) admits a unique ESE.
\end{proposition}

\begin{proof}
    Under enforced sharing, the total quality is $D(\mathbf{1}) = \xtot + \alpha \xtot = (1+\alpha) \xtot$.
    The traffic function simplifies to:
    \[
        T(\x, \mathbf{1}) = \mu \left( (1+\alpha) \xtot \right)^\gamma = \mu (1+\alpha)^\gamma \xtot^\gamma.
    \]
    The utility function for creator $i$ becomes:
    \begin{align*}
        U_i(\x, \mathbf{1}) &= T(\x, \mathbf{1}) \frac{x_i (1+\alpha \rho)}{D(\mathbf{1})} - c_i(x_i) \\
        &= \mu (1+\alpha)^\gamma \xtot^\gamma \cdot \frac{x_i (1+\alpha \rho)}{(1+\alpha)\xtot} - c_i(x_i) \\
        &= \mu (1+\alpha)^{\gamma-1} (1+\alpha \rho) \xtot^{\gamma-1} x_i - c_i(x_i).
    \end{align*}
    Let $\tilde{\mu} = \mu (1+\alpha)^{\gamma-1} (1+\alpha) = \mu (1+\alpha)^\gamma$. The utility takes the form $\frac{\tilde{\mu}(1+\alpha \rho)}{1+\alpha} \xtot^{\gamma-1} x_i - c_i(x_i)$.
    
    This expression has the same form as the utility function in the baseline model (\Cref{lem:utility_full_sharing}), differing only by the constant scalar $\tilde{\mu}$ instead of $\mu$. Since the scaling of the revenue term by a positive constant does not affect the strict concavity of the utility or the monotonicity of the game gradient, the existence and uniqueness of the ESE are preserved.
\end{proof}

\subsection{Proofs for GenAI Quality Beyond Linearity}\label{subsec:proofs_genai_quality_beta}

In this extension, the GenAI quality is given by $Q_{\text{AI}}(\x, \s) = \alpha \left( \sum_{i=1}^n x_i s_i \right)^\beta$ with $\beta \in (0,1]$. 

\begin{proposition}[Sharing Incentive under Universal Withholding]
    Fix a creator $i$ and a quality profile $\x$. Assume no other creator shares content ($\s_{-i} = \mathbf{0}$). Then, the utility of creator $i$ is strictly increasing in $s_i$ on the interval $(0, 1]$ if $\rho > \frac{x_i}{\xtot}$, and strictly decreasing if $\rho < \frac{x_i}{\xtot}$. Consequently, creator $i$ has an incentive to deviate to full sharing ($s_i=1$) if and only if $\rho > \frac{x_i}{\xtot}$.
\end{proposition}

\begin{proof}
    Let $X = \xtot$ be the total human quality. Since $\s_{-i} = \mathbf{0}$, the total shared content is $S(s_i) = x_i s_i$. 
    If $s_i = 0$, the utility is simply the direct traffic revenue share: $U_i(0) = \mu X^\gamma \frac{x_i}{X} - c_i(x_i)$.
    
    For any $s_i > 0$, creator $i$ is the sole contributor to the GenAI. Thus, their share of the AI revenue is $100\%$ of the allocated pool (i.e., $\frac{x_i s_i}{S(s_i)} = 1$). The GenAI quality is $Q(s_i) = \alpha (x_i s_i)^\beta$.
    The utility function for $s_i > 0$ becomes:
    \[
        U_i(s_i) = \mu X^\gamma \left[ \frac{x_i}{X + Q(s_i)} + \frac{Q(s_i)}{X + Q(s_i)} \cdot \rho \cdot 1 \right] - c_i(x_i).
    \]
    Let $Y(s_i) = Q(s_i) = \alpha (x_i s_i)^\beta$. Since $\beta > 0$, $Y(s_i)$ is strictly increasing in $s_i$. We can analyze the utility with respect to $Y$:
    \[
        U_i(Y) = \mu X^\gamma \left[ \frac{x_i + \rho Y}{X + Y} \right] - c_i(x_i).
    \]
    Differentiating with respect to $Y$:
    \begin{align*}
        \frac{dU_i}{dY} &= \mu X^\gamma \cdot \frac{\rho(X + Y) - (x_i + \rho Y)}{(X + Y)^2} \\
        &= \mu X^\gamma \cdot \frac{\rho X + \rho Y - x_i - \rho Y}{(X + Y)^2} \\
        &= \mu X^\gamma \cdot \frac{\rho X - x_i}{(X + Y)^2}.
    \end{align*}
    The sign of the derivative is determined entirely by the numerator $\rho X - x_i$, which is constant with respect to $s_i$ (and $Y$).
    
    \begin{itemize}
        \item If $\rho X - x_i > 0 \iff \rho > \frac{x_i}{X}$, then $\frac{dU_i}{dY} > 0$. Since $Y$ increases with $s_i$, $U_i(s_i)$ is strictly increasing on $(0, 1]$. The optimal strategy is $s_i = 1$.
        \item If $\rho X - x_i < 0 \iff \rho < \frac{x_i}{X}$, then $U_i(s_i)$ is strictly decreasing on $(0, 1]$. The optimal strategy approaches $s_i \to 0$ (or $s_i=0$ exactly, by continuity from the right).
    \end{itemize}
    Thus, the condition $\rho > \frac{x_i}{\xtot}$ is necessary and sufficient for creator $i$ to profitably deviate from a state of non-sharing to full sharing.
\end{proof}

We also note that \Cref{prop:unique_pne_rho1} generalizes to this new modeling, and refer the reader to its proof in \Cref{sec:sec3_proofs}, where we already proved it for this generalized model.

%% file: appendix/other_allocation_rules.tex
\section{Other Allocation Rules}\label{sec:other_allocation_rules}

In this section, we analyze two alternative functional forms to the proportional allocation rule examined in the main text. We demonstrate that the stability of FSEs depends not only on the total amount of revenue shared but also on how that revenue is distributed among creators. Specifically, we show that some intuitively appealing rules may fail to induce FSE even when the platform redistributes its entire GenAI-driven revenue ($\rho=1$).

We consider two alternative allocation rules parameterized by $\rho \in [0,1]$:

\begin{enumerate}
    \item \textbf{Winner-Takes-All (WTA):} The entire $\rho$-fraction of AI revenue is allocated to the creators who provide the maximal effective contribution. Let $W(\mathbf{x}, \mathbf{s}) = \argmax_{j \in [n]} \{x_j s_j\}$. Then:
    \[
    f_{i,\rho}^{WTA}(\mathbf{x}, \mathbf{s}) =
    \begin{cases}
    \frac{\rho}{|W(\mathbf{x}, \mathbf{s})|}  & \text{if } i \in W(\mathbf{x}, \mathbf{s}) \text{ and } \max_j x_j s_j > 0 \\
    0 & \text{otherwise}
    \end{cases}
    \]

    \item \textbf{Binary Threshold with Equal Shares (BTES):} The $\rho$-fraction is distributed equally among all creators who fully opt-in ($s_i=1$), regardless of their quality contribution. Let $S_1(\mathbf{s}) = \{ j : s_j = 1 \}$. Then:
    \[
    f_{i,\rho}^{BTES}(\mathbf{x}, \mathbf{s}) =
    \begin{cases}
    \frac{\rho}{|S_1(\mathbf{s})|} & \text{if } i \in S_1(\mathbf{s}) \\
    0 & \text{otherwise}
    \end{cases}
    \]
\end{enumerate}

While \Cref{prop:unique_pne_rho1} and \Cref{thm:fse-existence} establish that the proportional allocation rule admits FSE for sufficiently high $\rho$, the following proposition shows that such stability guarantees do not extend to these alternative mechanisms.

\begin{proposition}[Instability under Alternative Rules] \label{prop:wta_btes_instability}
Under either WTA or BTES allocation rules with $\alpha>0$ and $n \ge 2$, there exist instances where no FSE exists, even with full revenue redistribution ($\rho=1$).
\end{proposition}

\begin{proof}
    We construct specific instances demonstrating the non-existence of FSE under each rule for $\rho=1$.
    
    \paragraph{Part 1: WTA non-existence.}
    Consider any candidate full-sharing profile $(\mathbf{x}, \mathbf{1})$. If all qualities are identical, any small perturbation breaks the tie and guarantees the deviator a large utility increase, so the original profile cannot have been an FSE. If creators have different quality levels, let $i$ be a creator such that $x_i < \max_k x_k$. Creator $i$ is not in the winner set $W(\mathbf{x}, \mathbf{1})$ and thus receives zero AI revenue. Their utility is purely derived from direct traffic:
    \[
        U_{i}(\mathbf{x}, \mathbf{1}; f_1^{WTA}) = \mu (\norm{\mathbf{x}})^\gamma \frac{x_i}{\norm{\mathbf{x}} + \qai} - c_i(x_i).
    \]
    Consider a deviation to $s_i' = 0$. The total shared content decreases, reducing the AI quality to $\qai' = \alpha (\norm{\mathbf{x}} - x_i) < \qai$. Creator $i$ still receives zero AI revenue, but benefits from reduced competition:
    \[
        U_{i}((\mathbf{x}, (0, \mathbf{1}_{-i})); f_1^{WTA}) = \mu (\norm{\mathbf{x}})^\gamma \frac{x_i}{\norm{\mathbf{x}} + \qai'} - c_i(x_i).
    \]
    Since $\qai' < \qai$, the traffic share $\frac{x_i}{\norm{\mathbf{x}} + \qai'}$ strictly increases, making the deviation profitable. If all qualities are identical, any creator can slightly increase their quality $x_i$ to become the unique winner, capturing the entire AI revenue pot, which also destabilizes the profile. Thus, no FSE exists.

    \paragraph{Part 2: BTES non-existence.}
    We provide a counterexample with $n=2$, $\alpha=\gamma=1$, and $\mu=10$. Let the cost functions be $c_1(x) = 0.1 x^2$ and $c_2(x) = 0.4 x^2$.
    
    First, we determine the unique ESE $\mathbf{x}^\star$ under full sharing ($s_1=s_2=1$). In this state, both creators share, so $|S_1|=2$. The AI quality is $\qai = \alpha(x_1 + x_2) = \xtot$. The utility for creator $i$ is:
    \[ 
        U_{i}(\mathbf{x}, \mathbf{1}; f^{BTES}) = 10 \xtot \frac{x_i}{2 \xtot} + \frac{1}{2} \left( 10 \xtot \frac{\xtot}{2 \xtot} \right) - c_i(x_i) = 5 x_i + 2.5 \xtot - c_i(x_i).
    \]
    The first-order conditions $\frac{\partial U_i}{\partial x_i} = 0$ yield:
    \[ 
        5 + 2.5 - 0.2 x_1 = 0 \implies x_1^\star = 37.5,
    \]
    \[ 
        5 + 2.5 - 0.8 x_2 = 0 \implies x_2^\star = 9.375.
    \]
    The resulting total quality is $\xtot^\star = 46.875$. Creator 1's utility in this potential FSE is:
    \[
        U_1(\mathbf{x}^\star, \mathbf{1}) = 5(37.5) + 2.5(46.875) - 0.1(37.5)^2 = 187.5 + 117.1875 - 140.625 = 164.0625.
    \]
    
    Now consider a deviation by Creator 1 to $s_1'=0$ (keeping $x_1=37.5$). They forfeit their share of AI revenue, but the AI quality drops to $\qai' = \alpha x_2^\star = 9.375$. Creator 1's new utility is derived solely from direct traffic:
    \[
        U_1^{dev} = 10 (46.875) \frac{37.5}{46.875 + 9.375} - 0.1(37.5)^2 = 468.75 \left( \frac{37.5}{56.25} \right) - 140.625 = 171.875.
    \]
    Thus, the deviation is profitable. Thus, the unique ESE is not stable under strategic sharing, proving the non-existence of FSE.
\end{proof}

The instability of these rules stems from a misalignment between contribution and reward. Under WTA, non-winning creators have no incentive to share, as doing so only strengthens the GenAI competitor without providing any compensation. Under BTES, high-quality creators (like Creator 1 in the proof) are undercompensated relative to their large contribution to the AI's quality; they prefer to withhold data to weaken the AI rather than accept a diluted share of revenue equal to that of lower-quality contributors. The proportional allocation rule resolves this by linking rewards directly to the magnitude of contribution, ensuring that stronger creators are adequately incentivized.

%% file: appendix/modeling_assumptions_and_limitations.tex
\section{Modeling Assumptions and Limitations}\label{sec:assumptions}

This section outlines key assumptions in our modeling framework and discusses their implications and limitations.

\paragraph{User Traffic Modeling} Our baseline model assumes that platform user traffic is driven entirely by human-created content, with AI-generated content affecting allocation but not the traffic volume. This encompasses the scenario of a fixed number of users as a special case when $\gamma = 0$. In \Cref{sec:extensions}, we relax this assumption by considering a traffic model that incorporates both human and AI-generated content, and show that our main results carry over with appropriately adjusted constants. More general traffic models--such as those capturing nonlinear interactions between human and AI quality--remain unexplored.

\paragraph{Generative AI Quality Modeling} We model the quality of GenAI content as a linear function of the total shared creator data: $Q_{\text{AI}}(\mathbf{x}, \mathbf{s}) = \alpha \cdot \sum_i x_i s_i$. In practice, the relationship between training data quantity and model quality typically exhibits diminishing returns, consistent with established scaling laws in machine learning~\cite{DBLP:journals/corr/abs-2001-08361}. To address this, \Cref{sec:extensions} considers a generalized quality function $Q_{\text{AI}} = \alpha (\sum_i x_i s_i)^\beta$ with $\beta \in [0,1]$. While certain theoretical guarantees valid under the linear model do not hold in this generalized setting, our simulations suggest that the key qualitative insights persist. Other factors that may affect GenAI quality in practice--such as data diversity, representativeness, and complementarity among creators' contributions--are not captured by our model and remain directions for future work.

\paragraph{Scalar Representation of Content Competitiveness} Each creator's quality choice is represented as a single scalar $x_i \in [0,\infty)$, embodying the overall ``competitiveness'' of their content. This quantity encompasses various dimensions, including quality, quantity, trend alignment, expressiveness, domain specificity, and novelty. This modeling choice is a common simplification in the literature~\cite{10.5555/3737916.3740675, 10.5555/3692070.3694417}, enabling tractable analysis while preserving the essence of competitive dynamics. However, real-world scenarios may involve more intricate representations, such as continuous vectors of features~\cite{hron2022modeling}, requiring entirely different analytical tools.

\paragraph{Continuous Data Sharing Decisions} Our model enables creators to select their sharing level continuously, represented by $s_i \in [0,1]$. In practice, platforms usually provide binary options. Nevertheless, all our results remain valid under binary sharing choices, as \Cref{lem:s_choice_threshold} ensures that only binary sharing decisions are rational.

\paragraph{Static Game Formulation} Our model captures a one-shot interaction between the platform and creators. In practice, creator-platform relationships are ongoing, and the quality of GenAI systems depends on data accumulated over time. Such dynamic settings introduce richer strategic considerations, including the timing of sharing decisions and the platform's ability to update its allocation rule. Modeling these interactions as repeated or sequential games constitutes a natural extension of our framework.